\documentclass[12pt, letter]{article}
\usepackage{kpfonts}
\usepackage{comment}
\usepackage{bm}
\usepackage{amssymb}
\usepackage{amsmath}
\usepackage{xfrac}
\usepackage{stmaryrd,longtable}
\usepackage{float}
\usepackage{etoolbox}

\usepackage{amsthm}

\usepackage{graphicx}

\newtheorem{thm}{Theorem}
\newtheorem{cor}{Corollary}

\newtheorem{lem}{Lemma}
\newtheorem{prop}{Proposition}
\newtheorem{rmk}{Remark}
\newtheorem{ex}{Example}
\newtheorem{cl}{Claim}
\newtheorem{df}{Definition}

\newtheorem{ass}{Assumption}

\usepackage{enumerate}
\usepackage{paralist}
\usepackage{indentfirst}
\usepackage{setspace}	
\usepackage[normalem]{ulem}
\usepackage{natbib}	
\usepackage{bibentry}

\usepackage{appendix}
\usepackage{xspace,setspace}

\usepackage{color}
\usepackage[usenames, dvipsnames, svgnames, table]{xcolor}

\usepackage{float}
\floatstyle{boxed}
\newfloat{aside}{h}{aside}
\usepackage{caption}
\newif\ifsubmission

\renewcommand{\arraystretch}{0.7}

\usepackage{geometry}
\usepackage{algorithm}
\usepackage[noend]{algpseudocode}

\usepackage[pagebackref=true]{hyperref}
\hypersetup{
	pdfborder={0 0 0.5 [2 1]},
	colorlinks=true,
	linkcolor=blue,
	citecolor=Maroon,
	urlcolor=blue,
	linkbordercolor=blue,
	citebordercolor=PineGreen,
	urlbordercolor=blue
}

\usepackage[capitalize,nameinlink]{cleveref}
\crefname{ineq}{Inequality}{Inequalities}
\crefname{ass}{Assumption}{Assumptions}
\creflabelformat{ineq}{#2{\upshape(#1)}#3}
\Crefname{cor}{Corollary}{Corollaries}
\Crefname{conj}{Conjecture}{Conjectures}
\Crefname{thm}{Theorem}{Theorems}
\Crefname{rmk}{Remark}{Remarks}
\Crefname{prop}{Proposition}{Propositions}
\Crefname{lem}{Lemma}{Lemmata}
\Crefname{ex}{Example}{Examples}
\Crefname{cl}{Claim}{Claims}
\Crefname{df}{Definition}{Definitions}

\usepackage{tikz}
\usetikzlibrary{positioning,arrows.meta}

\newcommand{\assum}[1]{\begin{ass}#1\end{ass}}

\newcommand{\pisd}{\mathrel P_{i}^{sd}}
\newcommand{\risd}{\mathrel {R}_{i}^{sd}}
\newcommand{\nisd}{\mathrel N_{i}^{sd}}

\newcommand{\rul}{mechanism\xspace}
\newcommand{\ruls}{mechanisms\xspace}
\newcommand{\Ruls}{Mechanisms\xspace}

\DeclareMathOperator*{\sd}{\pi^\text{$P$}}
\DeclareMathOperator*{\rp}{\pi^\text{$P$}}
\DeclareMathOperator*{\md}{\textsl{$MD$}}
\newcommand{\len}{\text{length}}

\DeclareMathOperator*\integ{int}

\DeclareMathOperator*\supp{supp}

\newcommand{\defn}[1]{{\boldmath \bf #1}}
\usepackage{amsmath}
\usepackage{accents}
\newlength{\dhatheight}
\newcommand{\doublehat}[1]{%
    \settoheight{\dhatheight}{\ensuremath{\hat{#1}}}%
    \addtolength{\dhatheight}{-0.2ex}%
    \hat{\vphantom{\rule{1pt}{\dhatheight}}%
    \smash{\hat{#1}}}}

\newlength{\dbarheight}

\usepackage{soul}
\usepackage{natbib}
\newcommand{\citeposs}[1]{\citeauthor{#1}'s \citeyearpar{#1}}

\newcommand\doubleplus{\ensuremath{\mathbin{+\mkern-2mu+}}}

\newcommand\sdsp{strategy-proof\xspace}
\newcommand\sdspness{strategy-proofness\xspace}
\newcommand\Sdspness{Strategy-proofness\xspace}

\newcommand\robustefficiency{unambiguous efficiency\xspace}

\newcommand\strongefficiency{unambiguous efficiency\xspace}
\newcommand\strongefficient{unambiguously efficient\xspace}
\newcommand\Strongefficiency{Unambiguous efficiency\xspace}
\newcommand\StrongEfficiency{Unambiguous Efficiency\xspace}

\newcommand\sdefficiency{ambiguous efficiency\xspace}
\newcommand\sdefficient{ambiguously efficient\xspace}
\newcommand\Sdefficiency{Ambiguous efficiency\xspace}
\newcommand\SdEfficiency{Ambiguous Efficiency\xspace}

\makeatletter
\long\def\@makecaption#1#2{%
  \vskip\abovecaptionskip
  \sbox\@tempboxa{\textbf{#1}: #2}%
  \ifdim \wd\@tempboxa >\hsize
    \textbf{#1:} #2\par
  \else
    \global \@minipagefalse
    \hb@xt@\hsize{\box\@tempboxa\hfil}%
  \fi
  \vskip\belowcaptionskip}
\makeatother

\title{
Unambiguous Efficiency of Random Allocations
}

\author{
\begin{tabular}{ccc}
\begin{tabular}{c}
Eun Jeong Heo\\
{\footnotesize University of Seoul }\\[7pt]
{\footnotesize eunjheo@uos.ac.kr }\\
\end{tabular}&
\begin{tabular}{c}
Vikram Manjunath\\
{\footnotesize University of Ottawa}\\[7pt]
{\footnotesize vikram@dosamobile.com}\\
\end{tabular}&
\begin{tabular}{c}
Samson Alva\\
{\footnotesize University of Texas}\\[-1pt] {\footnotesize at  San Antonio}\\[-1pt]
{\footnotesize samson.alva@utsa.edu}
\end{tabular}
\end{tabular}
}

\begin{document}
\sloppy
\maketitle
\setstretch{1}
\begin{abstract}
When allocating indivisible objects via lottery, planners often use
ordinal mechanisms, which elicit agents' rankings of objects rather
than their full preferences over lotteries. In such an ordinal informational
environment, planners cannot differentiate between utility profiles that
induce the same ranking of objects. We propose the criterion of \emph{unambiguous efficiency}: regardless
of how each agent extends their preferences over objects to
lotteries, the allocation is Pareto efficient with respect to the extended
preferences. We compare
this with the predominant efficiency criterion used for ordinal mechanisms.
As an application to mechanism design, we characterize all efficient and
strategy-proof mechanisms satisfying certain regularity conditions.

\bigskip\noindent
{\it JEL} classification: C70, D47, D61.

\end{abstract}

\setstretch{1.3}

\section{Introduction}
\label{sec:introduction}

It would not be  controversial for us to say that
efficiency is the  central normative consideration
in most economic analyses.
We have long converged on a
formulation of this concept that says: if moving from
$x$ to $y$
makes someone better off without making anyone worse off, $x$ is not efficient. This concept has been studied across different economic
contexts. We investigate its application to object allocation
through lotteries. We begin our analysis by tracing the intellectual lineage
of efficiency, from its origins to the recent literature on random
allocation \ruls.

We attribute this efficiency notion to Pareto, but  Edgeworth refers
to it even earlier  in
his \emph{Mathematical {Psychics}} (\citeauthor{Edgeworth1881},
\citeyear{Edgeworth1881}, p. 27,  emphasis in original):
  ``the   utility of any
  one contractor must be a maximum \emph{relative to} the utilities of
  other contractors being constant, or not decreasing.''
Like Edgeworth, Pareto describes  his now-eponymous efficiency
concept while  formalizing the sense in which an
equilibrium allocation is optimal. In \cite{Pareto:1902}, having drawn
 indifference
curves of two individuals through a point, he remarks on the point of
tangency, $m$, as
follows (the English translation is from
  \cite{Pareto-en:1902}):\footnote{Pareto used a common  origin
  for both individuals
  rather than   what we have come to call the Edgeworth box.
  The  diagram in  \cite{Edgeworth1881} is not what we now know as the
  Edgeworth  box. Pareto
  subsequently discovered the box in
  \cite{Pareto:1906}.   So, perhaps we should call them  the
  Pareto-box and   Edgeworth-efficiency. We thank William Thomson for
  bringing up the murky issues of attribution for these ideas. For
  more, see \cite{TarascioWEJ1972}.}
\begin{quote}
  Il punto $m$ gode dunque della propriet\`a che da esso non ci si pu\`o
  allontanare, col baratto, o con altro simile ordinamento tale che
  ci\`o che \`e tolto ad un individuo viene dato all'altro, per modo di
  fare crescere entrambe le ofelimit\`a totali dei due individui.

  [The point $m$ therefore enjoys the property that it is not possible
  in departing from it, by barter, or by some similar arrangement
  such that  what is  taken away from one individual is given to
  another, to increase the total ophelimities of both individuals.]
\end{quote}
This description  applies to allocations of indivisible goods, but
\cite{Arrow1951} expands its scope.  In discussing
  the  compensation principle,  \citeauthor{Arrow1951} states an
  efficiency   criterion that amounts to Pareto-efficiency and says
  ``this formulation is certainly not debatable except perhaps on a
  philosophy of systematically denying people whatever they want.''
  By the 1970s, Pareto-efficiency was accepted as a fundamental
  normative criterion without need for an explanation. For instance, in
  imposing the requirement of Pareto-efficiency on voting schemes,
  \cite{ZeckhauserAPSR1973} merely says that it is ``unambiguous,''
  yet makes an   argument for the other requirement (non-dictatorship).

 In first studying  the problem we
 study here, the allocation of indivisible
goods via lottery,
\citeposs{HyllandZeckhauserJPE1979}  primary objective is
 Pareto-efficiency. In their model, each agent has
a von Neumann-Morgenstern (vNM)
utility function with which they compare lotteries over the
objects. Edgeworth's and Pareto's descriptions of Pareto-efficiency
 translate
naturally to such a setting and they mean the same thing. The
equivalence relies on  agents'  preferences over lotteries being
complete.

A subsequent line of research on allocation via  lotteries pertains to
\emph{ordinal} allocation \ruls, building on
\cite{BogomolnaiaMoulinJET2001}, and the current paper contributes to
this literature. So, we consider allocation \ruls that
respond only to the  ordinal content
of agents' preferences over objects without regard for intensities. The restriction
to ordinal \ruls may be driven by
cognitive or  informational
constraints. Specifically,  vNM preferences may not be adequate to
represent agents' preferences (see, for instance,
  \cite{KagelRoth1995}). Even if they are adequate, incentive
  compatibility in such
settings (alongside some regularity conditions) also implies
ordinality.\footnote{Under a strong property
called ``uniform cone continuity,''
\cite{EhlersMajumdarMishraSenJME2020}
show that strategy-proofness
alone implies ordinality in a very general model. Even with a
substantially weaker notion of continuity, \cite{HeoManjunath2025} show that  incentive
compatibility alongside efficiency, neutrality and non-bossiness
implies ordinality in the case of 3 agents.}

The analysis of
efficiency in this context
is less straightforward than informationally unconstrained settings.
Rather than comparing lotteries on
the basis of vNM utilities, a natural approach is
to compare based on (first order)
stochastic dominance.
Both Edgeworth's and Pareto's descriptions of efficiency are that an
agent cannot be made better off subject to a
constraint. Edgeworth's constraint is not making others
worse off, while Pareto's is making others better
off as well. If the welfare comparisons are made using complete relations,
these two descriptions render equivalent efficiency concepts.
However, the stochastic
dominance  relation is incomplete.
\cite{BogomolnaiaMoulinJET2001}   define
``ordinal efficiency''
so
that an
allocation $\pi$ is efficient if there is no allocation $\pi'$ that stochastically
dominates $\pi$ for each agent whose assignment is changed.
This corresponds to Pareto's formulation using the
stochastic dominance relation for welfare comparisons.
It has been subsequently used by a vast
literature, where it is often called ``sd-efficiency''.

Perhaps because
Pareto's efficiency description is now standard, the efficiency concept of \cite{BogomolnaiaMoulinJET2001} has
been accepted and built upon with little discussion of its
normative interpretation in terms of complete preferences over lotteries.
It is not the
only way to define  efficiency  of an ordinal
\rul.
We argue for an alternative that has a more clear-cut normative
interpretation.  If each
agent has a well-defined preference over lotteries (whether vNM
or not), even if we are interested in ordinal \ruls,
we should still
make efficiency judgements with regard to these
preferences over lotteries. {In
  other words, a social planner may want an  allocation whose
  efficiency is   robust to the information they cannot access
  (specifically, intensity of preferences between objects).}
  However, the planner cannot distinguish between
  different utility profiles that induce the same ordinal
  preferences. Thus, they can only be sure that an allocation is
  efficient if it is efficient for all the utility profiles that
  induce the ordinal preferences over objects they see. In this sense, we refer to
  it as \emph{\strongefficiency}. The ambiguity, from the planner’s
  perspective, is about agents’ complete preferences over
  lotteries. Rather than a single preference profile, they have to
  make judgements about efficiency for a \emph{set} of preference
  profiles.
   This is the  notion in \cite{GibbardEconometrica1977}
with regard to probabilistic, yet ordinal, voting
schemes. {\cite{GibbardEconometrica1977} refers to ``Pareto
  optimality ex ante'' with regard to vNM utilities. To extend this
  definition to ordinal \ruls, he says that a
  lottery is ``Pareto optimific ex ante'' at  an ordinal preference
  profile $P$ if it is Pareto optimal ex ante for every profile of
  vNM utilities consistent with $P$.\footnote{In this paper, we follow
    the temporal terminology of   \cite{GibbardEconometrica1977}  with regard to the realization of the
  lottery over sure allocations. Thus, ``ex ante''  means ``before the
  lottery is
  realized'' while ``ex post'' means ``after the lottery is realized.''
  Though obvious, we hope to avoid any confusion this may cause for
  readers steeped in the
  literature  where the  temporal frame of reference is the
  realization of individuals' preferences (or types more generally).}
We will eschew the term
  ``optimificity'' and stick with
  efficiency.}

Unlike this notion,
\citeposs{BogomolnaiaMoulinJET2001} efficiency only amounts to
efficiency for \emph{some} consistent
vNM utility
profile \citep{McLennan:JET2002} (see also  \cite{CarrollJET2010} and
  \cite{AzizBrandlBrandtJME2015}, or \cite{ChoDoganEL2016} for lexicographic preferences).    Importantly, an allocation may be
deemed efficient in this sense even if it is not
efficient for the \emph{actual}   utility
profile. For this reason, we will refer to it as
\emph{\sdefficiency}. Again, the ambiguity is from the planner's
perspective: for the set of utility profiles they have in mind, the
allocation is guaranteed to be efficient for some, but not necessarily
all of them.\footnote{In a related context, \cite{DoganDoganYilmazJET2018} propose a domination criterion that compares sets of consistent vNM utility profiles that ensure the efficiency of the corresponding allocations. This criterion is characterized using ex-ante social welfare, defined as the sum of agents’ expected utilities.}

We give a graphical illustration of the difference in
the Marschak triangle for three objects  $a,
b,$ and $c$  (\cref{fig:eff-v-sd}).
 Suppose that agent $i$ has a
vNM utility function $u_i$  such
that $u_i(a) > u_i(b) > u_i(c)$. Given a lottery $\pi_i$, the
lotteries that
stochastically dominate $\pi_i$ are in the cone marked (A). The
lotteries that are dominated by $\pi_i$ are in the cone marked
(B). The lotteries in the cones marked (C) and (D) neither dominate
nor are dominated by $\pi_i$. Since $u_i$ is vNM, $i$'s
indifference curve through $\pi_i$ is a line that passes through the
cones (C) and (D). We depict the  indifference curve through $\pi_i$
with a dashed
line and an arrow showing the direction of increasing utility. For
each allocation  $\pi'$, \sdefficiency merely requires that not all
agents can
receive a lottery in the region  corresponding to cone (A). So, at
least one agent,   perhaps $i$ receives a lottery outside (A), as
depicted. The lottery $\pi'_i$ is not comparable by stochastic
dominance to $\pi_i$, even though $i$ prefers it. On the other hand,
the only way an allocation can be efficient for every ordinally
equivalent  profile of
utility functions is if any change, say to $\pi''$,  moves
some agent into the region corresponding to (B) as shown in the
figure. This yields an alternative definition based on
stochastic dominance that is decidedly stronger than
\sdefficiency and corresponds to \strongefficiency.
An advantage to this definition is that
\emph{even if} agents do not have vNM preferences, as long as actual
preferences are completions of the stochastic dominance relation, it
does not allow any further improvement with regards to their actual preferences.

\begin{figure}
  \centering
  \includegraphics{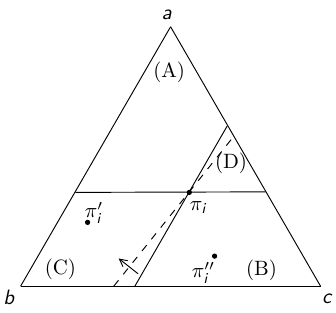}
  \caption{\footnotesize A graphical illustration of how the incompleteness of the
    stochastic dominance relation has consequences for the definition
    of efficiency}
  \label{fig:eff-v-sd}
\end{figure}

Our contributions are as follows. First, we formally
define \strongefficiency as discussed above. We provide a necessary
condition on allocations for them to be \strongefficient. Moreover, we
show that  \strongefficient allocations are \emph{combinatorial}
\citep{AzizMackenzieXiaYeAAMAS2015}, meaning that whether or not an
allocation is efficient is determined by which objects each agent
receives with positive probability regardless of  magnitude.
Second, we argue that  efficiency, even in probabilistic
settings, is hardly compatible with a basic notion of fairness, even
though the purpose of randomization is to typically achieve
equity. Symmetry says that
  agents who have the same preferences ought to receive  the same
  lottery. If the planner only sees ordinal preferences, symmetry is a
  rather strong property: even agents with vastly different
  preferences over lotteries are indistinguishable. We demonstrate that symmetric
  allocations
  are susceptible to  severe efficiency losses. Specifically, we show
  that as
  the number of agents grows, the worst case welfare cost
  of imposing symmetry alongside  ordinality is as high as it could
  possibly   be (\cref{ex cost of symmetry}). Some well-known ordinal
  rules, such as the probabilistic serial rule and random serial
  dictatorship, suffer such severe losses. Our observation about
  welfare loss is  similar to a  result in \cite{Pycia2014}. He shows
  that  welfare loss from  any ordinal, regular, symmetric
  \rul that is   strategy-proof and ambiguously efficient can be
  arbitrarily high.

Third, we show that randomization is still feasible under efficiency, albeit to a limited extent, and contrast the strength of
\strongefficiency to that
of \sdefficiency on the class of  random serial dictatorship
allocations, which are frequently used in
  practice. They involve drawing an ordering of the agents at random
  and assigning each agent, according to the realized ordering, their
  most preferred object among those
  remaining. However, these allocations
  are not necessarily ex ante efficient
  \citep{HyllandZeckhauserJPE1979,BogomolnaiaMoulinJET2001}. It turns
  out that restricting the set of  orders over which one randomizes
  can  correct this flaw at the cost of
  fairness. \cite{HarlessPhanGEB2022} characterize such sets of orders
  that ensure \sdefficiency of random serial dictatorship
  allocations. They call a
  set of orders that vary only in how three consecutive agents are
  ranked ``adjacent-three'' sets and show that these are maximal sets
  over which randomized serial dictatorship is \sdefficient. We show
  that ``adjacent-two'' sets are maximal if one seeks unambiguous
  efficiency, {rather than \sdefficiency (\cref{prop: adjacent-two}).}

Finally, we consider the question of incentives. What \emph{\ruls} can
a social planner use to elicit preferences
truthfully? We characterize the class
of all ordinal  \ruls that are \strongefficient,
strategy-proof, and satisfy some regularity conditions (non-bossiness,
neutrality, and bounded invariance) (\cref{thm: characterization}). To define this
class of \ruls, we provide a novel recursive algorithm. Given a supply
level between zero and one for each object, the algorithm assigns
lotteries to one or two agents at a time. When allocating to a single
agent (``monarchy''), it assigns them their most preferred lottery subject to the
supply vector. When allocating to two agents (``diarchy''), call them $i$ and $j$,
it first  computes $i$'s  most preferred lottery, $\pi_i$, subject to
supply. Then it computes $i$'s most preferred lottery,
$\pi_i'$, subject to supply \emph{after} removing $j$'s most preferred lottery. It assigns to $i$, for a fixed $\alpha$, the lottery
$\alpha\pi_i + (1-\alpha)\pi_i'$. It allocates the analogous
lottery to $j$. It is as though $i$ and $j$ flip
an $\alpha$-weighted coin, the winner picks first, and the other picks
second. Having allocated lotteries to one or two agents as described
above, the
algorithm removes these agents and subtracts from the supply
vector what it has
allocated to them. It  solves the remaining problem recursively
with the new (residual) supply vector and remaining agents. 
The choices of successive agents to allocate to can depend
on allocations that have already been finalized. Specifically, who
comes next can depend on the
identities of the preceding agents and, for each pair of them, the
number of objects of which both have received positive probability.  
We call a \rul induced
by this algorithm a ``hierarchy of monarchies and diarchies''.

The class of these \ruls includes, but is not
limited to, serial dictatorship \ruls characterized by
\cite{SvenssonSCW11999}. In this sense, our  characterization is in
the spirit of
\citeposs{SvenssonSCW11999} characterization, but in the context of probabilistic
allocation. This characterization is not in conflict with
  \citeposs{Zhou:JET1990} incompatibility  between strategy-proofness,
  ex ante efficiency, \emph{and} an equity property in the form of
  symmetry. To the best of our knowledge,
  an analogous characterization with \sdefficiency or ex post
  efficiency
  in place of \robustefficiency does not exist in the literature.
\cite{BasteckEhlersECMA} have recently provided a
counterexample to  the
long standing conjecture that strategy-proofness, ex post efficiency,
and symmetry  characterize random serial dictatorship.

We have organized the remainder of the paper as follows. In
\cref{sec:model-1}, we formally define the random
allocation problem. We discuss the informational restriction to
ordinal preferences in \cref{sec:ordinality}, welfare comparisons in \cref{sec:sd}, and
efficiency in \cref{sec:efficiency}. In \cref{sec: characterization}, we take up an application of our
 efficiency concept to mechanism design.
We offer some concluding remarks in \Cref{sec:conclusion}. All proofs
are in the appendix.

\section{The Model}
\label{sec:model-1}
We are interested in settings where
each agent
  consumes  a unique object and each object is
consumed by a single agent.
Let \defn{$A$} be a finite set of at least three objects and
\defn{$N$} be a set of exactly as many agents, with cardinality $n$.
Bijections from agents to objects describe (deterministic) allocations.
When there are
  fewer than three objects the issues discussed in this paper are
  trivial so we omit that case.

Let \defn{$\Delta(A)$} be the set of lotteries over $A$. As it does not cause
confusion, we identify
degenerate lotteries in $\Delta(A)$ with $A$ itself. So, for each
$a\in A$, we denote the lottery that places probability one on $a$ by
$a$. We denote a generic lottery that agent $i$ might consume by
$\pi_i\equiv (\pi_{ia})_{a\in A}\in \Delta(A)$ where  $\pi_{ia}$ specifies the probability that~$i$ receives $a$. We  denote its support by  \defn{$\supp(\pi_i)$}.
We admit the possibility of
randomizing over deterministic allocations.
By the Birkhoff-von Neumann Theorem, we can think of a (randomized)
allocation  as a bistochastic matrix
\citep{Birkhoff:1946,vonNeumann:1953}.  That
is,  an allocation $\pi$ specifies a lottery over $A$
for each agent $i\in N$ such that
for each $a\in A$, $\sum_{i\in N}\pi_{ia}=
1$. Denote by \defn{$\Pi$} the set of all such allocations.
Each agent has preferences over the set of allocations $\Pi$.
\assum{\label{ass:self-interested}Each agent's preferences depend only upon the lottery that
they receive and are defined by a utility function from
$\Delta(A)$ to $\mathbb R$.}

We denote by $\mathcal U$ the universe of
utility functions considered possible for an agent.
Fixing the set of agents and the
objects to assign to them, an economy is described by a profile of
agents' preferences $u\in \mathcal U^N$.

Given that an economy describes each agent's \emph{complete}
preferences
over allocations, the definition of efficiency is straightforward. An
allocation $\pi\in\Pi$ is \defn{efficient at $u$} if there is no
alternative allocation $\pi'\in \Pi$ that makes at least one agent
{better} off without making any other agent {worse} off. Formally, there
is no $\pi'\in \Pi$ such that for some $i\in N$, $u_i(\pi_i') >
u_i(\pi_i)$ and for no $i\in N$,
$u_i(\pi_i') <  u_i(\pi_i)$.

\subsection{The Ordinal Informational Environment }\label{sec:ordinality}
If the social planner, \rul designer, or whatever other analyst,
is able to successfully elicit from each agent  complete
preferences over lotteries then the definition of efficiency above is
all we need. However, we are interested in an environment where they must make decisions or evaluations based on only \emph{some} aspects of preferences.

A salient constraint is
that the only information available about each agent is
a complete ordering of objects, which we
call the agent's \defn{ordinal preferences}.
Such  a constraint might be driven by
individuals having difficulty in articulating, let alone conveying,
their full preferences over all lotteries (e.g.
\cite{KagelRoth1995}). Alternatively, it could be due to incentive
concerns
(e.g. \cite{EhlersMajumdarMishraSenJME2020,PyciaUnver2024,HeoManjunath2025}).
Our goal is to draw attention to the definition of efficiency in such
an \emph{ordinal informational environment}. We preclude indifference among objects. 
\assum{\label{ass:strict}For any utility function in $\mathcal U$, distinct objects yield distinct utility levels.}
That is, for every utility function $u_i \in \mathcal U$ an agent $i$ might have, there is a unique linear order $P_i$ over $A$, called the \emph{ordering induced by} $u_i$, such that for each pair of distinct objects $a,b \in A$, $u_i(a) > u_i(b)$ if and only if $a \mathrel P_i b$.

The upshot of this assumption is that one must  take  a profile of linear orders
over $A$ as the primitive of an allocation problem.
We denote by
\defn{$\mathcal P$},
the set of linear orders over  $A$. Through
the informationally constrained lens, an economy is thereby
described not by a utility
profile in $\mathcal U^N$ but by an ordinal preference profile in $\mathcal
P^N$.

\subsection{Expected Utility, Stochastic Dominance, and Welfare}
\label{sec:sd}

A natural model of the allocation problem at hand is one where
individuals have preferences over lotteries satisfying
expected utility theory \citep{HyllandZeckhauserJPE1979,Zhou:JET1990}.
Suppose an agent $i\in N$ compares lotteries in this way. So,
$u_i\in\mathcal U$ is a von Neumann-Morgenstern (vNM) utility
function where there is $x\in
\mathbb R^A$ such that for each $\pi_i\in \Delta(A), u_i(\pi_i) = x\cdot
\pi_i$.  We now make a richness assumption on $\mathcal U$.  \assum{\label{ass:rich}If
$u_i\in \mathcal U$ induces the ordering $P_i\in \mathcal P$, then
every von Neumann-Morgenstern utility function $u_i'$ that induces
$P_i$ is also in $\mathcal U$.}

If all the planner knows is that $i$ has a vNM utility function that
induces the ordering $P_i$, then they can still make some unambiguous
comparisons in terms of  $i$'s preferences using stochastic
dominance.
Given a pair of allocations, $\pi, \pi'\in \Pi$, we say \defn{$\pi_i$     stochastically dominates $\pi_i'$
  at $P_i$} if for each $a\in A$, $\sum_{b P_i a}\pi_{ib}\ge
\sum_{b\mathrel  P_i a}\pi'_{ib}$ with at least one strict
inequality. We denote this by  \defn{$\pi_i \pisd \pi_i'$}.
We use the notation  \defn{$\pi_i \risd
  \pi_i'$} if  either $\pi_i \pisd
\pi_i'$ or $\pi_i =\pi_i'$. Since $\risd$ is only a partial order, it may be that
neither $\pi_i \risd \pi'_i$ nor $\pi'_i\risd \pi_i$, in which case we
write \defn{$\pi_i'\nisd \pi_i$}. By definition, $\nisd$ is a symmetric
  relation and if $\pi_i \nisd \pi_i'$ there are
  $a,b\in A$ such that $\sum_{o P_i a}\pi_{io}>
\sum_{o\mathrel  P_i a}\pi'_{io}$ and {$\sum_{o P_i b}\pi'_{io}>
\sum_{o\mathrel  P_i b}\pi_{io}$.}

Our last assumption on $\mathcal U$ is that every utility function in
it respects its own stochastic dominance relation. \assum
{\label{ass:sd}For every $u_i\in\mathcal U$, if $u_i$ induces
  $P_i\in \mathcal P$, then $u_i$ respects $\pisd$. That is, for every
pair $\pi_i, \pi_i'\in \Delta(A)$ such that $\pi_i \pisd \pi_i'$,
$u_i(\pi_i) > u_i(\pi_i')$.}

Under our assumptions,
\[
  \begin{array}{ccc}
    \pi_i \pisd \pi_i' &
   \text{ if and only if }& \left(\begin{array}{c}
    \text{ for each }u_i\in \mathcal
  U\text{ that induces }P_i,\\ u_i(\pi_i) > u_i(\pi_i').\end{array}\right)
  \end{array}
\]
Thus, a planner who sees
preferences only through the ordinal lens is still able to identify
\emph{unambiguous} welfare comparisons between lotteries for each
agent. Moreover, these are exactly the ones that can be made based on
stochastic dominance.

We also make a similar observation about incomparable allocations.
\[
  \begin{array}{ccc}
    \pi_i \nisd \pi_i'  &
    \text{ if and only if }  & \left(\begin{array}{c}
  \text{there are }u_i,u_i'\in
  \mathcal   U\text{ both inducing }P_i \text{ such that } \\u_i(\pi_i) >
  u_i(\pi_i')\text{ but } u_i'(\pi_i) <  u_i'(\pi_i').\end{array}\right)
  \end{array}
\]

  \begin{rmk}{Admissible Domains.}
{\em Our results build on the observation that
  $\pi_i \pisd \pi_i'$ is equivalent to dominance according to
  \emph{every} $u_i\in \mathcal U$ consistent with $P_i$. This holds
  as long as
  $\mathcal U$ is rich enough---in the sense of \cref{ass:rich}. For
  instance, consider a domain  containing all vNM utilities along with
  mixture-averse utilities of \cite{DillenbergerSegal2024} with strict
  orderings of degenerate lotteries. Since this domain satisfies our
  assumptions, despite being a strict superset of vNM utilities, all
  of our analysis applies.

  The only restrictions on admissibility of
  utility functions in the domain alongside vNM utilities are those
  stated in \cref{ass:strict} and \cref{ass:sd}. The former says that
  agents are not indifferent between distinct objects. The latter
  allows for non-vNM utilities so long as they are completions of the
  stochastic dominance relation---in other words, ``monotonic'' with
  respect to it, as assumed by
  \cite{DillenbergerSegal2024}.\footnote{For an instance of such an
    assumption in the  cooperative bargaining context, see
    \cite{RubinsteinSafraThomsonECMA1992}.}
}
  \end{rmk}

\subsection{Efficiency}\label{sec:efficiency}

Now, we may ask: \emph{If  individuals  have
  preferences over lotteries, but the planner/analyst  only knows
  their  ordinal preferences over objects, what should
  they consider efficient?  }

When they see the preference profile
$P\in\mathcal P^N$, they only know that the agents' actual preferences
are some $u\in \mathcal{U}^N$ that induces $ P$. To be  certain that the
allocation is truly efficient it has to be efficient for \emph{every}
$u$ that induces $ P$. We call such an allocation \defn{\strongefficient}.
The counterpart of this notion in the context of ordinal voting \ruls was
studied in \cite{GibbardEconometrica1977}.

On the other hand, the planner or analyst might use the existential
quantifier by  only asking that there is \emph{some} $u$
that induces  $P$ for which their chosen allocation is
efficient. While they could point to this possible utility profile as
justification for their choice, there is no guarantee that it is
efficient with regard to the agents' actual preferences. We call such an allocation \defn{\sdefficient}.

In what follows, we start by discussing
unambiguous efficiency and consider some of its implications. We then compare it to ambiguous efficiency.
Since the only useful information that the planner has is summarized
by the stochastic dominance relation for each agent,
we define them in terms of these relations.\footnote{Conceptual cousins
of unambiguous and ambiguous efficiency have been studied in the literature
on fair deterministic allocation of indivisible objects where agents consume multiple objects
\citep{BramsEdelmanFishburnTD2003,AzizEtAlAAMAS2016,ManjunathWestkamp2024}.
Since agents' preferences are over sets of
  objects, like ordinality in our setting one may wish to
  compute fair allocations using only agents' marginal preferences over
  objects. A notion of dominance analogous to the stochastic
  dominance gives a partial order, and there are again two ways to
  define efficiency with respect to this order.}

\subsubsection{\Strongefficiency}\label{sec:strongefficiency}
As discussed above, an unambiguously efficient allocation is one that
is efficient
no matter what the  agents' fine-grained preferences over lotteries
are.

  An alternative definition of \strongefficiency, in terms of
  stochastic dominance, is the following:
  \begin{quote}
    An allocation $\pi\in \Pi$ is \textbf{{\strongefficient}} at~$P\in
    \mathcal P^N$ if there is no $\pi'\in \Pi\setminus\{\pi\}$ such
    that for each  $i\in N$, either $\pi_i'\risd \pi_i$ or
    $\pi_i'\nisd\pi_i$.
  \end{quote}
In \cref{apx:strongefficiency} we prove that this is equivalent to
the earlier definition.

 It is equivalent to say
that  switching from an \strongefficient allocation to any other
 makes at least one agent \emph{comparably} worse off per their
stochastic dominance relation. So, if  $\pi$ is \strongefficient
at~$P$ then  for each $\pi'\in \Pi\setminus\{\pi\}$, there is some
$i\in N$ such that $\pi_i\pisd \pi_i'$. This formulation guarantees
that we can find  no Pareto-improvement from an \strongefficient
allocation even if we can fully access the  agents'
preferences over lotteries.

We now present several implications of \strongefficiency. First, it
restricts the number of objects that any pair of agents can jointly
receive with positive probabilities.

\begin{lem}\label{lem:limited randomization}
  If $\pi\in \Pi$ is \strongefficient at $P\in \mathcal P^N$, then for
  each pair $i,j\in N$, $|\supp(\pi_i)\cap
  \supp(\pi_j)| \leq 2$.
\end{lem}

This observation seems to
significantly restrict the  size of an agent's support, but we
note that it is not always the case. The following example
demonstrates that the support size can even be maximal for some agent.

\begin{ex}[Support of an \strongefficient allocation]\em{Let
    $N=\{1,\cdots,7,8\}$, $A= \{a_1,\cdots,a_8\}$,  a
  preference profile $P$, and an allocation $\pi$ be given as follows:

\renewcommand{\arraystretch}{0.8}
\[\begin{array}{llllllll}
  P_1 & P_2 & P_3 & P_4 & P_5 & P_6 & P_7 & P_8\\\hline\\[-10pt]
  a_1 & a_1 & a_2 & a_3 & a_4 & a_5 & a_6 & a_7\\
  a_2 & a_8 & a_8 & a_8 & a_8 & a_8 & a_8 & a_8\\
  a_3 & \cdot& \cdot& \cdot& \cdot& \cdot& \cdot& \cdot\\
  a_4 & \cdot& \cdot& \cdot& \cdot& \cdot& \cdot& \cdot\\
  a_5 & \cdot& \cdot& \cdot& \cdot& \cdot& \cdot& \cdot\\
  a_6 & \cdot& \cdot& \cdot& \cdot& \cdot& \cdot& \cdot\\
  a_7 & \cdot& \cdot& \cdot& \cdot& \cdot& \cdot& \cdot\\
  a_8 & \cdot& \cdot& \cdot& \cdot& \cdot& \cdot& \cdot\\\hline
  \end{array}\renewcommand{\arraystretch}{0.9}
\qquad \pi=\left[\begin{array}{c|cccccccc} & a_1 & a_2 & a_3 & a_4 & a_5 & a_6 & a_7 & a_8\\\hline\\[-10pt]
1 & \frac{1}{8} & \frac{1}{8} & \frac{1}{8} &  \frac{1}{8} & \frac{1}{8} & \frac{1}{8} & \frac{1}{8} & \frac{1}{8}\\
2 & \frac{7}{8} & 0 & 0 & 0 & 0 & 0 & 0 & \frac{1}{8}  \\
3 & 0 & \frac{7}{8} & 0 & 0 & 0 & 0 & 0 & \frac{1}{8}  \\
4 & 0 & 0 &\frac{7}{8} & 0 & 0 & 0 & 0 & \frac{1}{8}  \\
5 & 0 & 0 & 0 & \frac{7}{8} & 0 & 0 & 0 & \frac{1}{8}  \\
6 & 0 & 0 & 0 & 0 &\frac{7}{8} & 0 & 0 & \frac{1}{8}  \\
7 & 0 & 0 & 0 & 0 & 0 &\frac{7}{8} & 0 & \frac{1}{8}  \\
8 & 0 & 0 & 0 & 0 & 0 & 0 &\frac{7}{8} & \frac{1}{8}  \\ \end{array}\right]\]
where the preferences over all the unspecified objects in the profile
(denoted by dots) can be anything. Note that $\pi$ is \strongefficient at $P$ and so are any other allocations such that the support of each agent's lottery is exactly the set of objects
specified in the preference profile. The limited intersection of support stated in \cref{lem:limited
  randomization} holds at~$\pi$, yet
the support size of agent~1 is maximal. \hfill $\blacksquare$
}\end{ex}

\renewcommand{\arraystretch}{0.7}

We now turn to the structure of the set of \strongefficient
allocations. An efficiency notion $X$ is \defn{combinatorial} if it is invariant
to shifting mass over the same support. That is, $X$ is combinatorial
if whenever  $\pi\in\Pi$ satisfies $X$, it is also satisfied by every
$\pi'\in \Pi$ such that
 $\{(i,a):\pi_{ia}' > 0\} = \{(i,a):\pi_{ia} > 0\}$
\citep{AzizMackenzieXiaYeAAMAS2015}.  We can strengthen this to
say that $X$ is \defn{strongly combinatorial}  if we replace the
condition relating $\pi$ and $\pi'$ with  $\{(i,a):\pi_{ia}' > 0\} \subseteq \{(i,a):\pi_{ia} > 0\}$.
It turns out that  \strongefficiency
is strongly combinatorial (as is \sdefficiency \citep{AzizMackenzieXiaYeAAMAS2015,EcheniqueRootSandomirskiy2022}).

\begin{prop}\label{prop:smallersupport}
  \Strongefficiency is strongly combinatorial.
\end{prop}

As mentioned above, the Birkhoff-von Neumann Theorem tells us that
every allocation in $\Pi$
is a convex combination of permutation matrices, which correspond to
deterministic allocations.  In other
words, every allocation is a point in the Birkhoff polytope. By \cref{prop:smallersupport}
 if a face of the Birkhoff polytope contains an \strongefficient
allocation in its relative interior, then every allocation on that face is
\strongefficient. So, rearranging or shifting mass between permutation
matrices in the support of an \strongefficient allocation does not
compromise efficiency.

\subsubsection{\Sdefficiency}
We now consider the alternative approach that says  it suffices for the
planner to  just find \emph{some} utility profile to justify their
choice of allocation as efficient.  Seeing the ordinal preference
profile $P\in\mathcal P^N$,  this planner  would say the
allocation $\pi$ is efficient  if there is no other allocation $\pi'$
that gives every agent either the same lottery as $\pi$ or a lottery
that beats it in
stochastic dominance terms. Formally,
 $\pi\in\Pi$ is \sdefficient at $P\in \mathcal{P}^N$, if there is no
 $\pi'\in \Pi\setminus \{\pi\}$ such that for each $i\in N$, $\pi_i' \risd \pi_i$.
An equivalent definition is to say that any alternative to an
\sdefficient allocation gives at least one agent an allocation that is
worse or incomparable in  stochastic dominance terms.
This has been the dominant formulation of efficiency in the literature
on random assignment (often under the moniker ``sd-efficiency'') since
it was first introduced under the name of
``ordinal efficiency'' by \cite{BogomolnaiaMoulinJET2001}. It is a
rather weak property of the allocation in that  there may be
utility profiles where an \sdefficient allocation  is actually far from
efficient.  We
demonstrate this in the following example.

\begin{ex}[\Sdefficiency] \em{Let $N=\{1,2,3\}$ and $A=\{a,b,c\}$. Let
    $\overline P\in \mathcal P^N$ be a preference profile such that
    each $i\in N$ prefers $a$ to $b$ and $b$ to $c$. Let $\pi\in \Pi$ be
    such that $\pi_1=(\frac{1}{2}, 0, \frac{1}{2})$ and
    $\pi_2=\pi_3=(\frac{1}{4}, \frac{1}{2}, \frac{1}{4})$. It is easy
    to verify that this allocation is \sdefficient at~$\overline P$. Suppose
    that each~$i$ has a vNM utility function $u_i\in \mathcal{U}$
    that induces $\overline P_i$:\renewcommand{\arraystretch}{0.7} \[
  \begin{array}{c|ccc}
  & u_1  & u_2 & u_3\\\hline\\[-10pt]
    a & 1 & 1 & 1\\
    b & 1-\varepsilon & \varepsilon & \varepsilon \\
    c & 0 & 0 & 0 \end{array}
\]
where $0<\varepsilon<\frac{1}{2}$. Consider another
allocation $\pi'\in \Pi$ such that $\pi_1'=(0,1,0)$ and
$\pi_2'=\pi_3'=(\frac{1}{2}, 0,\frac{1}{2})$. If we move from  $\pi$ to
$\pi'$, for each $i$ the new lottery~$\pi_i'$
 is incomparable to $\pi_i$, so $\pi'$ does not
 dominate $\pi$ as far as \sdefficiency is concerned. Yet, each agent
 enjoys a
 higher expected utility: for
 agent 1, it increases from $\frac{1}{2}$ to $1-\varepsilon$
and for the others it  increases from
$\frac{1}{4}+\frac{\varepsilon}{2}$ to $\frac{1}{2}$. This
Pareto-improvement  at~$u$
demonstrates how \sdefficiency may be too restrictive in its
definition of domination and therefore too permissive in which
allocations are deemed efficient.  \hfill
$\blacksquare$}\end{ex}

A common  impetus for randomization is the desire for treating agents
fairly. This makes feasible what is, perhaps, the
  oldest form of fairness, going back over two thousand years to
  Aristotle's argument that
  one should  treat equals equally \citep{Aristotle}. For the random
  allocation problem that we study, it first appears in \cite{Zhou:JET1990}.
An allocation   $\pi\in \Pi$ is \textbf{symmetric} (or treats equals
equally) at $P$ if agents who have the same preference receive the same lottery. That is, for each pair of agents
$i,j\in N$, if  $P_i = P_j$ then $\pi_i = \pi_j$.

The main methods of choosing allocations  that have been studied in the literature (the
\emph{Probabilistic Serial} and \emph{Random Serial
  Dictatorship} \ruls described in \cite{BogomolnaiaMoulinJET2001})
choose  symmetric allocations.
However, as we demonstrate in the next example, symmetry  can lead to
substantial losses of welfare. That is, symmetry comes at a
potentially high cost in terms of efficiency.

\begin{ex}[Cost of symmetry]\label{ex cost of symmetry}{\em
  Let $\overline P\in \mathcal P^N$ be such that for each $i\in N,
  a_1\mathrel{ \overline P}_i a_2 \dots\mathrel {\overline P_i}
  a_n$. Let $u^\varepsilon\in \mathcal U^N$ be a vNM utility profile
  such that for each $i\in N$, $u_i^\varepsilon$ induces $\overline
  P_i$. Without loss of generality, we normalize $u^\varepsilon$ in a
  way that for each $i\in N$,  $i$ attaches a utility of 1 to their
  most preferred object and zero  to their least preferred
  object. Then, for each pair of lotteries, $\pi_i$ and $\pi_i'$, the
  difference $(u_i^\varepsilon(\pi_i') - u_i^\varepsilon(\pi_i))$
  measures the change in $i$'s welfare in terms of probability shifted
  from  their least to most preferred object. Let $\varepsilon > 0$
  and $u^\varepsilon\in \mathcal U^N$ be as depicted in the following
  table:\renewcommand{\arraystretch}{0.8}
\[
  \begin{array}{c|cccc}
    &u^\varepsilon_1&u^\varepsilon_2&\dots &u^\varepsilon_n\\
    \hline
    a_1 & 1 & 1 & \dots &1\\
    a_2 & 1-\varepsilon & (n-2)\varepsilon & \dots &(n-2)\varepsilon
    \\
    a_3 & (n-3)\varepsilon & (n-3)\varepsilon & \dots &(n-3)\varepsilon
    \\
    \vdots &\vdots & \vdots & & \vdots\\
    a_{n-1}& \varepsilon & \varepsilon & \dots &\varepsilon\\
    a_n&0 & 0 &\dots &0
  \end{array}
\]
Since all of the agents have the same ordinal preferences, the
symmetric allocation, $\overline \pi$  gives each agent each object with equal
probability of $\frac{1}{n}$, that is
$\vec{\frac{1}{n}}$. Now, consider a
trade between agent~1 and any other agent~$j\neq 1,2$ where
$j$ gives $1$ a portion of $j$'s share of $a_2$ in exchange for some
of $1$'s shares of $a_1$ and, for some $k > 2$, $a_k$.  Denote the rate of
exchange that would leave agent $j$ indifferent  to such a trade (so
that the gains accrue to agent 1) by $\mu^k$. Solving $j$'s
indifference condition, we see that
\[
  \mu^k =  \frac{(k-2)\varepsilon}{1- (n-k)\varepsilon}.
\]
Agent 1 gains from such a trade since, per unit of $a_2$ that they
receive from $j$,  their expected utility changes
by $( -\mu^k + (1-\varepsilon) - (1-\mu^k)(n-k)\varepsilon)$, which is positive for small enough $\varepsilon$.
The volume of this trade is limited by the fact that $j$ can only
transfer $\frac{1}{n}$ units of $a_2$ to agent 1. While agent 1 can
engage in such trades with more than one $j$, they are limited by the
$\frac{1}{n}$ of $a_k$ that they have and the constraint that they end
up with  one unit of total probability. Engaging in $(n-2)$
trades (perhaps evenly distributed among all trading partners in
$N\setminus \{1,2\}$), agent~1 can achieve the following lottery without
harm to any other agent:\renewcommand{\arraystretch}{1}
\[
  \pi^\varepsilon_{1a_l} = \left\{
    \begin{array}{cc}
      \frac{1}{n}\left(1- \sum_{k=3}^n\mu^k\right)& \text{if }l=1\\
      \frac{n-1}{n} & \text{if }l = 2\\
      \frac{\mu^l}{n} & \text{if }l \geq 3.
  \end{array}\right.
\]
Note that $\pi^\varepsilon$  does not exhaust all gains from trade,
but it suffices to make our point.
For any $\delta>0$, we can find $\varepsilon$ such that
$u_1^\varepsilon(\overline \pi_1) - \frac{2}{n} \leq \delta$ and $1-
u_1^\varepsilon(\pi_1^\varepsilon) \leq \delta$. That is,
for small enough $\varepsilon$, $u^\varepsilon_1(\pi_1^\varepsilon) -
u^\varepsilon_1(\overline \pi_1)$ is no smaller than $\frac{n-2}{n}-2\delta$. So, $1$ can
achieve gains from trade close to the equivalent of shifting
$\frac{n-2}{n}$ mass from their least preferred object ($a_n$) to their
most preferred object ($a_1$). As the number of
agents grows, this goes to 1, the equivalent of going from their worst
object with certainty to their best object with certainty.\hfill $\blacksquare$}
\end{ex}

The upshot of \cref{ex cost of symmetry} is similar to
\cite{Pycia2014}. The main difference is that while we consider the
welfare of a single agent, \cite{Pycia2014}
makes a stronger  statement:  welfare loss under a  utilitarian social
welfare function  can be arbitrarily high. In order to draw this
stronger conclusion, he considers
\ruls that are not only ordinal and symmetric, but also regular (as
defined by \cite{LiuPycia2016}) and asymptotically
strategy-proof and ambiguously efficient.

\subsubsection{\StrongEfficiency Versus \SdEfficiency: Randomizing
  Over Serial Dictatorship Allocations}\label{sec:random
  dict}

A simple way of finding an efficient allocation is to let agents pick
their favourite available  object one by one. In fact, we can reach
every efficient deterministic allocation by varying the order in which
agents pick. We call these  \defn{serial dictatorship
  allocations}.
Of course, one can randomize over the
ordering, leading to \defn{random serial dictatorship
allocations}. Selecting an  allocation by  uniformly randomizing
over all possible orderings is a simple way to provide strong
incentives for agents to truthfully convey their ordinal preferences
\citep{LiAER2017,PyciaTroyanEconometrica2023}, but such  allocations
are not necessarily  \strongefficient or even \sdefficient. One way to
get around this
problem is to randomize  over a \emph{subset} of orderings.
In what follows, we compare the extent to which this subset has to be
restricted to achieve \strongefficiency as compared to just
\sdefficiency.
As a point of clarification, in this section we are interested in an
allocation's
``punctual'' property of efficiency at a fixed profile of
preferences. In the next section we will turn to mechanisms, which map
preference profiles to allocations and respond to changes in the
preference profile.

Let $\mathcal O$ be the set of all
orderings over $N$. For each $P\in\mathcal P^N$ and each $\succ\ \in
\mathcal O$, denote the serial
dictatorship allocation with respect to
$\succ$ by  \defn{$\sd(\succ)$}.
Given $P$ and a probability distribution $\lambda$ over $\mathcal O$,
the randomized  serial dictatorship allocation  with respect to
$\lambda$, which we denote
by \defn{$\rp(\lambda)$},  is the $\lambda$-weighted convex combination of $\sd(\succ)$s. So,  $$\rp(\lambda) =
\sum_{\succ\in \mathcal O} \lambda(\succ)\sd(\succ).$$

Given a  subset of orders $\mathcal O' \subseteq \mathcal O$, we say
that $\mathcal O'$ is \defn{compatible} with property $X$  if for each
$\lambda$ with support $\mathcal O'$ and each $P\in \mathcal P^N$,
$\rp(\lambda)$ satisfies $X$. If no strict superset of $\mathcal O'$ in $\mathcal O$
is compatible with $X$, we
say that $\mathcal O'$ is \defn{maximal} for $X$.
A set of orders $\mathcal O'$ is an
\emph{adjacent-$k$} set if there is a set $N'\subseteq N$ of $k$ agents such that
\begin{enumerate}
\item all orders in $\mathcal O'$ agree on $N\setminus N'$;\\[-25pt]
\item for each of $k!$ orderings of $N'$, some ordering of $\mathcal O'$ agrees with it over $N'$;\\[-25pt]
\item agents in $N'$ are consecutive in every order in $\mathcal O'$;\\[-25pt]
\item every order in $\mathcal O'$ ranks the same agents in
  $N\setminus N'$ above those in $N'$.\\[-25pt]

\end{enumerate}

We now contrast \sdefficiency and \strongefficiency in terms of their
maximal sets.

\begin{prop}\citep{HarlessPhanGEB2022}\label{prop: adjacent-three}
  Adjacent-3 sets are maximal for \sdefficiency.
\end{prop}

\begin{prop}\label{prop: adjacent-two}
  A set of orders $\mathcal O' \subseteq \mathcal O$ is maximal for \strongefficiency if 
  and only if $\mathcal O'$ is an adjacent-2 set.
\end{prop}

These results  show that both \sdefficiency and \strongefficiency
substantially limit the flexibility of randomization for serial
dictatorship allocations. However, weakening \strongefficiency to
\sdefficiency only expands the range from adjacent-2 to adjacent-3
sets in terms of adjacent-$k$ sets.
For \sdefficiency, \cite{HarlessPhanGEB2022} identify maximal 
compatible sets outside the adjacent-$k$ family. For \strongefficiency, 
\cref{prop: adjacent-two} leaves no room for such sets.

\section{Incentives and  Efficiency}
\label{sec: characterization}
Until now, we have concerned ourselves only with what a
planner might consider efficient. We thought of  a
planner who had access to ordinal preferences.  \emph{What
  if the planner has also to worry about manipulative agents?}
In this section, we consider this mechanism design question.
Specifically, we describe all  the ways a social
planner in an ordinal informational environment can implement an
efficient allocation subject to incentive constraints and some
regularity conditions. The strength of our efficiency property makes
this novel characterization possible: we are unaware of any
 corresponding characterizations with weaker efficiency requirements.

\subsection{Strategy-proofness and Other Properties}
\label{sec:strat-proofn-other}

Since we are now turning to
incentives, the object of our analysis is a \textbf{\rul}, which
maps the information that the planner receives to an allocation.
Since the planner solicits from agents their ordinal preferences over objects,
a \rul is a mapping \defn{$\varphi: \mathcal P^N \to \Pi$}.
We are interested only in such \ruls that choose an
\strongefficient allocation for every preference profile. We use the
notation  \defn{$\varphi_i(P)$} to denote the lottery
that $i\in N$ receives when the planner sees the ordinal preference
profile
$P$, where $\varphi_i(P)\equiv (\varphi_{ia}(P))_{a\in A}$ with $\varphi_{ia}(P)$ specifying the probability that $i$ receives $a$ under~$\varphi$ at~$P$.

{A standard incentive requirement, originating from
  \cite{BlackJPE1948} and \cite{FarquharsonOEP1956}, is to prevent
  profitable manipulations by any agent. Stated in terms of
  stochastic dominance: a mechanism is \defn{strategy-proof} if for
each $P\in \mathcal{P}^N$, each $i\in N$, and each $P_i'\in \mathcal{P}$,
$\varphi_i(P) \risd
\varphi_i(P_i',P_{-i})$
\citep{GibbardEconometrica1977,BogomolnaiaMoulinJET2001}. It is
equivalent to saying that there is no utility profile $u\in \mathcal
U^N$, agent $i\in N$, and utility function  $u'_i\in \mathcal U$ where
$u$ induces ordinal preference profile $P\in\mathcal P^N$ and  $u_i'$
induces $P_i'$, such that $ u_i(\varphi_i(P_i', P_{-i})) > u_i(\varphi_i(P)).$}
{If the
mechanism satisfies this property, the planner can rest assured that,
regardless of the  true utilities, no agent has  an incentive
to lie about their  preference. This is consistent with our definition
of
\strongefficiency.}

For the version of our  problem  without randomization,
\cite{SvenssonSCW11999} characterizes serial dictatorships as the only
strategy-proof \ruls that satisfy certain
regularity conditions. In what follows, we state versions of these
regularity conditions. They   deliver a characterization of a
class of
\strongefficient and \sdsp
\ruls. As mentioned above, to the best of our knowledge,
this is the first such characterization for random
allocation of objects.

 The first regularity condition is \defn{neutrality}, which says that a \rul ought to
 be robust to  permuting the names  of the
 objects in $A$.
Formally,  $\varphi$ is {neutral} if for
each $P\in \mathcal P^N$, and each permutation of the objects $\rho$,
whenever $P'\in \mathcal P^N$ rearranges $P$ according to $\rho$,
 for each $i\in N$ and $a\in A$,  $\varphi_{i\rho(a)}(P') =
  \varphi_{ia}(P).$ This concept
  first appears in \cite{MayEconometrica1952}. It has subsequently
  featured in the
  literature on allocation of indivisible goods \citep{SvenssonSCW11999}.

Second, $\varphi$  is \defn{non-bossy} if whenever an agent's reported
preferences change without affecting what they get,  this change
does  not affect what any agent gets. Formally, for each $P\in \mathcal
  P^N$, each $i\in N$, and each $P_i'\in \mathcal P$ such that
  $\varphi_i(P) = \varphi_i(P_i',P_{-i})$, we have $\varphi(P) =
  \varphi(P_i',P_{-i})$. Non-bossiness was first defined by
    \cite{SatterthwaiteSonnenscheinRES1981}. 

The final condition,  \textbf{bounded invariance}, is also a
restriction
on how a \rul responds to changes in a single agent's preferences. It
says that if the only changes are for objects they rank below $a\in
A$, then this should not affect how probability shares of  $a$ are
assigned.
Formally, $\varphi$ is boundedly invariant
   if for each $P\in \mathcal P^N$, each $i\in N$, each $a\in A$, and
   each $P_i'\in \mathcal P$, whenever  $\{o\in A:o\mathrel R_i a\}=\{o\in A:o\mathrel R_i' a\}$ and $P_i$ and $P_i'$ coincide on
   $\{o\in A:o\mathrel R_i a\}$,
for each $j\in N$, $\varphi_{ja}(P_i, P_{-i}) = \varphi_{ja}(P_i',
P_{-i})$.\footnote{This requirement is partially
    implied by \emph{strategy-proofness} and reflects a form of
    ``object-wise'' non-bossiness, in the sense that if agent~$i$
    receives the same probability for object~$a$ after changing only rankings
    below $a$, then the allocation of~$a$ to all other agents remains
    unaffected. Nonetheless, \emph{strategy-proofness} and
    \emph{non-bossiness} together do not imply \emph{bounded
      invariance}, and vice versa, demonstrating their logical
    independence.} Originally proposed by
  \cite{Heo2010}, this condition has been studied in the subsequent
  literature on random allocation \citep{BogomolnaiaHeoJET2012,
    BasteckEhlersECMA}. To the best of our knowledge, all mechanisms
  that have been studied in the literature so far satisfy this property.\footnote{An alternative definition of bounded
  invariance is what one might call \emph{lower reshuffling
    object-wise invariance}: Preference $P'_i$ is a lower reshuffling at object $a$ of preference $P_i$ if
for every object $o$, if $o \mathrel R_i a$, then the upper level set at $o$ of $P'_i$ is the same as that of $P_i$.
A preference profile $P'$ is a lower reshuffling at object $a$
of preference profile $P$ if
for every agent $i$, $P'_i$ is a lower reshuffling at $a$ of $P_i$.
A \rul $\varphi$ is lower reshuffling object-wise invariant if
for every object $a$ and every pair of preference profiles $P$ and $P'$,
if $P'$ is a lower reshuffling at $a$ of $P$,
then $\varphi_a(P') = \varphi_a(P)$.
}

\subsection{Hierarchies of Monarchies and Diarchies}
\label{sec:hier-monarch-diarch-1}

We describe a class of \ruls that are characterized by
\strongefficiency, \sdspness and the above three regularity properties.
Informally, these \ruls operate as follows. {The  supply of each object is initially set to~1.}
\begin{enumerate}
\item[Step 1:] Select either a single agent or a pair of agents.

  If we select a single agent, assign to them their most
  preferred object with certainty. If we select a pair of agents $i$
  and $j$,  flip a  weighted coin to
  determine who picks first and
  who picks second to be assigned their most preferred object. Of
  course, in the latter case, the allocations of $i$ and $j$ are
  non-deterministic if and only if they have the same favourite
  object. {Reduce the supply of each object by the
    assignment(s) from this step.}

\item[] $\qquad\qquad\vdots$
\item [Step $k$:]  {If at least one object has non-integral supply (strictly between zero and one) at} the end of the
  previous step, then select a single agent whose allocation we have
  yet to determine. Otherwise, select a
  single agent or a pair of agents.

  If we select a single agent, assign to them their most
  preferred lottery from the remaining supply  of all of
  the objects. If we select a pair of agents $i$
  and $j$,  flip a weighted coin to determine who picks first and who
  picks second to be assigned their most preferred
  object. {Reduce the  supply of each object
    by the assignment(s) from this step.}
\end{enumerate}Stop when all agents' assignments have been determined.

The  selection of single agents or pairs of agents in each step is
governed by a sequencing rule that is fixed a priori. It can
only depend on  some aspects of what has transpired in earlier
steps. Specifically, it can depend on:
{\begin{enumerate}
\item who have been assigned in each of the previous steps, and\\[-25pt]
\item for each pair of agents assigned in the previous steps, the
  number of objects of which both have received positive probability by
  the end of each of those steps.\\[-15pt]
\end{enumerate}}
We call these \ruls, which are parameterized by sequencing rules, \defn{Hierarchies of Monarchies and
  Diarchies (HMD)}. The ``monarchies'' refer to single agents and
``diarchies''
refer to the pairs of agents.  We provide a  rigorous definition of
the above algorithm and restrictions on the selection of agents in
\cref{apx:formal alg} along with some detailed demonstrative examples.
To fix ideas, we present a simple example of how such a \rul  works.
\renewcommand{\arraystretch}{0.9} \begin{ex}[Example of an HMD]\label{hmd}
  \em{Let $N=\{1,2,3,4,5\}$ and $A=\{a,b,c,d,e\}$.
Consider an HMD associated with the following
  simple sequencing rules: (i)~at step~1, a diarchy of 1 and~2  is
  chosen with weight $\frac{1}{3}$ in favour of 1; (ii)~if agents 1
  and 2 receive positive probability of no common object at step~1,
  then a diarchy of 3 and 4 with an equal weight is chosen at step~2,
  followed by a monarchy of 5; (iii)~if they receive positive
  probability of exactly one common object, then monarchies of 5, 4,
  and 3 are chosen in that order; and (iv)~if they receive positive
  probability of exactly two common objects, then a monarchy of 5 is
  chosen, followed by a diarchy of 3 and 4 with an equal weight. We
  illustrate this sequencing rule as a tree in Figure~\ref{hmdpic}.
\begin{figure}[t]
  \centering
  \includegraphics[width=150mm]{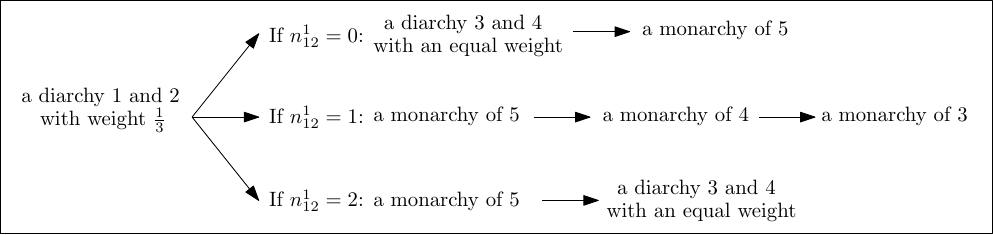}
  \caption{{\footnotesize The sequencing rule defining an HMD in
      Example~\ref{hmd} can be illustrated as a tree: the~selection of a monarchy or a diarchy in each step is contingent
      on the assignments of the earlier steps, where $n^1_{12}$ is the number of objects of which both agents 1 and 2 receive positive probability at step~1.}}
\label{hmdpic}\end{figure}
\renewcommand{\arraystretch}{0.8}
    Let  $P, \bar P\in \mathcal P^N$ be given as follows:
\[
  \begin{array}{ccccc}
 P_1&P_2&P_3&P_4 & P_5\\
    \hline
     a & a & b & a  &b \\
     e&  b & e& e & a  \\
     d&  d& c & b & c \\
     c&  e & a & c & e \\
     b&  c & d & d & d \\
  \end{array}
\qquad\qquad
  \begin{array}{ccccc}
\bar P_1&\bar P_2&\bar P_3&\bar P_4 & \bar P_5\\
    \hline
     a & b & b & a  &b \\
     b&  a & e& e & a  \\
     d&  d& c & b & c \\
     c&  e & a & d & e \\
     e&  c & d & c & d \\
  \end{array}\]
\renewcommand{\arraystretch}{0.9}
{Consider the first profile~$P$ above. In Step~1, we have a
  $\frac{1}{3}$ weighted diarchy of agents~1 and~2. They split $a$
  with 1 receiving
  $\frac{1}{3}$ of it and 2 receiving $\frac{2}{3}$. They  receive the
  remaining $\frac{2}{3}$ and $\frac{1}{3}$ respectively of  their
  second most preferred objects. In
  Step~2, since agents~1 and~2 receive positive probability of exactly
  one common object,~$a$, the selection
  rule says we have a monarchy of
  agent~5 next. Agent~5 gets their most preferred
  lottery from the remaining supply,  $\frac{2}{3}$
  of $b$ and $\frac{1}{3}$ of  $c$.
  Next, in Step~3, following the same branch of the
  sequencing rule, we have a monarchy of agent~4. Agent~4
  gets $\frac{1}{3}$ of~$e$ and $\frac{2}{3}$ of $c$.
  Finally,
  agent~3 gets  1 of the remaining object~$d$. The resulting
  allocation is $\pi$ below.

  To highlight the possibility of different execution paths, consider
  $\bar   P$. In Step~1, we again have a $\frac{1}{3}$ weighted
  diarchy of  agents~1 and~2, but   they
  receive probability~1 of their respective favourite objects. Since
  they receive positive probability of no common object, in Step~2 we have a
  $\frac{1}{2}$ weighted diarchy of agents~3 and~4. They split their
  common favourite object~$e$
  equally and get  the remaining half of their next preferred
  objects. In Step~3, agent~5 gets the remaining
  probabilities of~$c$ and~$d$. The resulting allocation is $\bar
  \pi$ shown  below.}
\[
\pi = \begin{array}{c|ccccc}
& a&b&c&d & e\\\hline
 1 & \frac{1}{3} & 0 & 0  & 0 & \frac{2}{3}\\
 2 & \frac{2}{3} & \frac{1}{3} & 0 & 0 & 0 \\
 3 & 0  & 0 & 0 & 1 & 0 \\
 4 & 0 & 0 & \frac{2}{3} & 0 & \frac{1}{3}\\
 5 & 0 & \frac{2}{3} & \frac{1}{3} & 0& 0\end{array}
\qquad\qquad
\bar \pi =
\begin{array}{c|ccccc}
& a&b&c&d & e\\\hline
 1 & 1 & 0 & 0  & 0 & 0\\
 2 & 0 & 1 & 0 & 0 & 0 \\
 3 & 0  & 0 & \frac{1}{2}  &  0 & \frac{1}{2} \\
 4 & 0 & 0 & 0 & \frac{1}{2} & \frac{1}{2}\\
 5 & 0 & 0 & \frac{1}{2} & \frac{1}{2} & 0\end{array}
\]\hfill $\blacksquare$}

\end{ex}
\renewcommand{\arraystretch}{1.1}
\begin{rmk}
  \label{remark:trees}
\emph{Each HMD is parameterized by an MD-sequencing rule on reachable
histories. Two MD-sequencing rules that differ only at non-reachable
histories therefore induce indistinguishable HMDs. Identifying
MD-sequencing rules that agree on reachable histories, we can enumerate
the entire class of HMDs by means of a tree-based representation, akin
to what is shown in \cref{hmdpic}. While there are only $n!$ serial
dictatorships, that enumeration shows the number of HMDs, without even
accounting for the freedom in specifying diarchy weights, to be doubly
exponential in $n$ and thus vastly larger.\footnote{We give the formal tree-based representation and the count in \cref{apx:trees}.}}
\end{rmk}

HMDs are reminiscent of the ``inheritance trees'' in
\cite{papai2000} where allocations are   deterministic. An important
difference is that an inheritance tree governs the rights to a
\emph{single object} while in an HMD, \emph{all} objects are
``controlled'' by either a monarch or diarchs at each step. Moreover,
the constraints that govern branching of these
structures  are different because our setting allows partial
allocation and that of \cite{papai2000} is not subject to
neutrality.

The combination of properties we impose  restricts the
degree to which a \rul can randomize.
For the problem of choosing a common
lottery by a voting scheme, non-bossiness and bounded invariance are vacuous. For
such problems, \cite{GibbardEconometrica1977} shows that every
ordinal, \strongefficient, and strategy-proof \rul is dictatorial and
therefore involves no randomization.
For our problem of allocating objects, if we do not allow
randomization at all, efficiency is implied by strategy-proofness,
non-bossiness, and neutrality, so these three  properties characterize
serial dictatorships \citep{SvenssonSCW11999}.
While our \ruls   do involve
some randomization, they still have the hierarchical nature of serial
dictatorships.

\begin{thm}\label{thm: characterization}
  A \rul  is \strongefficient,
  \sdsp,  non-bossy, neutral, and boundedly invariant if and only if
  it is   a hierarchy   of   monarchies and diarchies.
\end{thm}

In the appendix, we give examples of \ruls demonstrating that the
four properties other than  bounded invariance are independent of the
others.    

If there are only three agents, then every hierarchy of
monarchies and diarchies is a random serial dictatorship over an
adjacent-two set. Moreover, as we show in \cref{sec:axiom-relations}, bounded
  invariance is implied by the other axioms, making it redundant in our
list of properties and yielding the following corollary. Whether this is
true for more than three agents and objects remains an open question.

\begin{cor}
If $|A| = |N| = 3$, then a \rul is \strongefficient,
\sdsp, neutral, and non-bossy if and only if it is a random
serial dictatorship over an adjacent-two set.
\end{cor}

\begin{rmk}{``Triarchies''}: {\em A straightforward extension of the class of HMD \ruls 
  is to include the randomizations over three agents, or \emph{triarchies}, at histories
  with integral remaining probability vectors. The proofs that HMDs are
strategy-proof, neutral, non-bossy, and boundedly invariant adapt easily to show
that \emph{hierarchies of monarchies, diarchies,
  and triarchies} satisfy these properties. Though not necessarily
\strongefficient, these are \sdefficient.   }
\end{rmk}

Finally, we consider the unambiguous version of group strategy-proofness.
A \rul $\varphi$ is \defn{unambiguously group strategy-proof} if,
for every $P \in \mathcal P^N$, every nonempty coalition $C \subseteq N$, every
$P' \in \mathcal P^N$ with $P'_{-C} = P_{-C}$,
if there exists $j \in C$ such that $\varphi_j(P') \neq \varphi_j(P)$,
then there exists $i \in C$ such that $\varphi_i(P) \pisd
\varphi_i(P')$. 
This property implies both (unambiguous) \sdspness and non-bossiness.
While the converse holds for deterministic \ruls, it does not hold for
random \ruls in general.
Nevertheless, the converse does hold on the class of HMD \ruls.

\begin{cor}\label{cor:gsp}
  A \rul  is \strongefficient,
  unambiguously group strategy-proof,  neutral, and boundedly invariant if and only if
  it is   a hierarchy   of   monarchies and diarchies.
 \end{cor}

\section{Concluding Remarks}
\label{sec:conclusion}
For over a century, economists have considered efficient an
allocation that cannot be improved from the perspective of every
agent. Equivalently, an efficient allocation is one where any change
would necessarily make someone worse off. This equivalence breaks down
when  welfare comparisons are incomplete due to sparseness of
information, in which case  the first
formulation is strictly weaker than the second.

The allocation of  objects via lotteries on the basis of ordinal
rankings is a well-studied setting
where such comparisons are incomplete. The literature has focused on
\sdefficiency, which is the weaker  form. We have argued
that the stronger form has a more natural normative interpretation.
  \cite{EhlersMajumdarMishraSenJME2020,PyciaUnver2024,HeoManjunath2025}
  suggest that
      robustness in incentives requires ordinal mechanisms. One might
      view our results as saying that
      robustness in both incentives \emph{and}
      efficiency severely curtails the scope for randomization.

As we have demonstrated, an \sdefficient allocation may be quite far
from maximizing each agent's welfare without making others worse
off. This property deserves greater
scrutiny. A more systematic analysis of the degree to which it does
or does not ensure that welfare gains are realized is an important
avenue for future research.

There remain important  open questions about normative criteria that
 \strongefficiency is compatible with.
 Though \strongefficiency and symmetry are incompatible,
it is worth developing a better sense  of how much randomization, and
 fairness more generally, one can achieve with
 \strongefficient allocations in ordinal settings.

Finally, we have leveraged the strength of \strongefficiency in our
characterization result. Aside from the question of whether bounded
invariance is redundant, counterparts of this theorem for
weaker efficiency properties like ex post efficiency or even
\sdefficiency remain open.

\appendixpage
\appendix
\setstretch{1.2}

\section{Proofs from \Cref{sec:strongefficiency}}
\label{apx:strongefficiency}

  Our assumptions ensure that $\mathcal U$ contains all vNM utility
  functions with no ties.  \citet{McLennan:JET2002} shows that
  ambiguous efficiency at $P\in \mathcal P^N$ is equivalent to
  efficiency for \emph{some} vNM utility profile consistent with
  $P$. We show below that unambiguous efficiency at $P$ (as defined in
  terms of stochastic dominance in \cref{sec:strongefficiency})  is
  equivalent to
  efficiency for \emph{all} vNM utility profiles consistent with
  $P$.
  \begin{prop}
    Let $P\in \mathcal P^N$ and $\pi\in \Pi$. The following are
    equivalent:
    \begin{enumerate}
    \item \label{sd-unamb} There is no $\pi'\in \Pi\setminus \{\pi\}$ such that for
      each $i\in N$, either $\pi_i' \risd  \pi_i$ or $\pi_i' \nisd
      \pi_i$.
    \item \label{vnm-eff} For each vNM utility profile $u$ that induces $P$, $\pi$ is
      efficient at $u$.
    \end{enumerate}
  \end{prop}
  \begin{proof}
If \cref{vnm-eff} is not true, there is $u\in
\mathcal U^N$ such that $\pi$ is not efficient at $u$ while $u$ is
consistent with $P$. Then, there is $\pi'\in \Pi$ such that for each
$j\in N$, $u_j(\pi_j')\geq u_j(\pi_j)$, so $\pi_j = \pi_j'$ or $\pi_j$
does not stochastically dominate $\pi_j'$ at $P_j$. Moreover,  for
some $i\in N$, $u_i(\pi_i')>u_i(\pi_i)$ meaning that $\pi\neq
\pi'$. Since $\pi' \neq \pi$ and for each $i \in N$, either $\pi'_i =
\pi_i$ or not $\pi_i \mathrel P^{sd}_i \pi_i'$, so \cref{sd-unamb} is not true.

Suppose \cref{sd-unamb} is not true. Then
there is some $\pi'\neq \pi$ such that for each $i\in N$
either $\pi_i' = \pi_i$ or it is not
the case that $\pi_i \mathrel P^{sd}_i \pi_i'$, with at least one
instance of the latter.
In what follows, we construct a vNM utility profile $u$ such that for
each $i\in N$,  $u_i(\pi_i') \geq u_i(\pi_i)$ with at least one strict
inequality. For $i\in N$ such that  $\pi_i' = \pi_i$, we can select
any vNM utility $u_i$ that is consistent with $P_i$. So,
to complete the proof
that \cref{vnm-eff} is not true,
we only need to construct, for each $i\in N$ such that $\pi_i'\neq \pi_i$,
a vNM
utility $u_i\in \mathcal U$ such that $u_i(\pi_i') > u_i(\pi_i)$. In
what follows, we use the fact that $\pi_i$ does not stochastically
dominate $\pi_i'$ to  do so.

Since it is not the case that $\pi_i \mathrel P^{sd}_i \pi_i'$,
there is some $a\in A$ such that $\sum_{b\mathrel P_i a} \pi'_{ib}
>\sum_{b\mathrel P_i a} \pi_{ib}$. Let $a_1$ be the best
(according to $P_i$) such object.
Since $\pi_i$ sums to one, $a_1$ cannot be the least preferred object
according to $P_i$.  So, there is  $b_1\in A$  such that $a_1 \mathrel
P_i b_1$ and for no $b\in A$ is it the case that $a_1\mathrel P_i b
\mathrel P_i b_1$---that is, $b_1$ is ranked just below $a_1$.
Let $m = \sum_{a\mathrel R_i a_1} \pi_{ia}$ and $m' =
\sum_{a\mathrel R_i a_1} \pi'_{ia}$. As noted above, $m' > m$.

For $\epsilon, \Delta > 0$, such that $\Delta > \epsilon$, consider a
vNM utility
$u_i^{\epsilon\Delta}$ such that:
\begin{itemize}
\item $u_i^{\epsilon\Delta}(a_1) = \Delta$.
\item $u_i^{\epsilon\Delta}(b_1) = \epsilon$.
\item For each $a\in A$ such that
$a\mathrel P_i a_1$, $u^{\epsilon\Delta}_i(a) -u^{\epsilon\Delta}_i(a_1) <\epsilon$.
\item For each $b\in A$ such that
$b_1\mathrel P_i b$, $u^{\epsilon\Delta}_i(b_1) -u^{\epsilon\Delta}_i(b) <\epsilon$.
\end{itemize}

Then,     $u_i^{\epsilon\Delta}(\pi_i') - u_i^{\epsilon\Delta}(\pi_i) =$
\[
  \begin{array}{cl}
    &\sum_{a\in A} \pi_{ia}' u_i^{\epsilon\Delta}(a) - \sum_{a\in A} \pi_{ia}u_i^{\epsilon\Delta}(a) \\
    = &\left[\sum_{a\mathrel R_i a_1} \pi'_{ia}u_i^{\epsilon\Delta}(a)
        +\sum_{b_1 \mathrel R_i b} \pi'_{ib}u_i^{\epsilon\Delta}(b)
        \right]
        -
        \left[
        \sum_{a\mathrel R_i a_1} \pi_{ia}u_i^{\epsilon\Delta}(a)
        +\sum_{b_1 \mathrel R_i b} \pi_{ib}u_i^{\epsilon\Delta}(b)
        \right]\\
    \geq & \left[\sum_{a\mathrel R_i a_1} \pi'_{ia}\Delta
           +\sum_{b_1 \mathrel R_i b} \pi'_{ib}0
           \right]
           -
           \left[
           \sum_{a\mathrel R_i a_1} \pi_{ia}(\Delta+\epsilon)
           +\sum_{b_1 \mathrel R_i b} \pi_{ib}u_i^{\epsilon\Delta}(b_1)
           \right]\\
    \geq & \Delta\sum_{a\mathrel R_i a_1} \pi'_{ia}
           -
           (\Delta+\epsilon)\sum_{a\mathrel R_i
           a_1} \pi_{ia}
           -
           \epsilon\sum_{b_1 \mathrel R_i b} \pi_{ib}
    \\
    =& m'\Delta - m(\Delta+\epsilon) - (1-m)\epsilon\\
    = & (m' - m)\Delta - \epsilon.
  \end{array}
\]
Since $m' > m$, for any $\Delta$ if $\epsilon < (m' - m)\Delta$, we
have  $u_i^{\epsilon\Delta}(\pi_i') >u_i^{\epsilon\Delta}(\pi_i)$.
  \end{proof}

\begin{proof} [Proof of \Cref{lem:limited randomization}]Suppose by
  contradiction that there are $P\in \mathcal{P}^N$, $\pi\in \Pi$ that
  is \strongefficient at $P$, and distinct $i,j\in N$ such that $|\supp(\pi_i)\cap
  \supp(\pi_j)| \ge 3$. For each $a,b\in \supp(\pi_i)\cap
  \supp(\pi_j)$, $a\mathrel P_i b$ if and only if $a\mathrel
  P_j b$. Otherwise,
  say if $a\mathrel P_i b$ but  $b\mathrel P_j a$, then agent $i$ can transfer a positive probability
  of $b$ to agent $j$, receiving the same amount of $a$ from $j$,
  yielding a Pareto-improvement. This contradicts
  \robustefficiency. Therefore, there are distinct $a,b,c\in
  \supp(\pi_i)\cap  \supp(\pi_j)$ such that
  $a\mathrel P_i b\mathrel P_i c$ and   $a \mathrel P_j b\mathrel P_j
  c$. Agent $i$ can transfer a   positive probability
  of $a$ and $c$ to agent $j$, receiving the same total amount of $b$
  from $j$, yielding a possible Pareto-improvement (the result of the
  trade would not be comparable by $i$ or $j$ to their lottery at
  $\pi$ in terms of stochastic dominance). This again
  contradicts  \robustefficiency.\end{proof}

\begin{proof}[Proof of \cref{prop:smallersupport}] Before proving
  \cref{prop:smallersupport}, we start with some notation.
For an agent $i\in N$, define a
\textbf{probability shift} as $\eta_i\in [-1,1]^A$  such that
$\sum_{a\in A} \eta_{ia} = 0$. Given $\pi_i\in \Delta(A)$, we say that
the probability shift $\eta_i$ is a \textbf{feasible shift from \boldmath $\pi_i$}
if $\pi_i+ \eta_i\in \Delta(A)$. This definition implies that if
$\eta_i$ is a feasible shift from $\pi_i$, then for each $a\in A$ such
that  $\eta_{ia} < 0$, $-\eta_{ia} \leq \pi_{ia}$.
Given a preference relation $P_i\in \mathcal P$ for $i\in N$, the
probability shift $\eta_i$ is \textbf{not sd-worsening \boldmath at $P_i$} if 
it is trivial or if there is
$\underline a\in A$ such that $\sum_{a: a\mathrel P_i \underline a}
\eta_{ia} > 0$. By this definition, if $\eta_i$ is a feasible shift
from $\pi_i$ and is not sd-worsening at $P_i$, then it is not the case
that $\pi_i \pisd (\pi_i + \eta_i)$. This is because for some
$\underline a\in A$, $\sum_{a: a\mathrel P_i \underline a} \pi_{ia} <   \sum_{a: a\mathrel
    P_i \underline a} (\pi_{ia}  + \eta_{ia}) = \sum_{a: a\mathrel P_i \underline a}\pi_{ia} + \sum_{a: a\mathrel P_i \underline a}\eta_{ia}.$

For any pair $\pi_i, \pi_i' \in \Delta(A)$, $\eta_i = \pi_i' - \pi_i$
is, by definition, a feasible shift from $\pi_i$. Moreover, unless
$\pi_i \pisd \pi_i'$,  $\eta_i$ is not
sd-worsening at $P_i$. In fact, for any $\varepsilon\in(0,1)$,
$\varepsilon \eta_i$ is a feasible shift from $\pi_i$ that is not
sd-worsening.

\begin{cl}\label{cl: shift exists}
  Let $\pi_i, \pi_i'\in \Delta(A)$ be such that $\supp(\pi_i) \subseteq
  \supp(\pi_i')$. If $\eta_i$ is a non-trivial feasible shift from $\pi_i$ then
  there is $\varepsilon > 0$ such that $\varepsilon\eta_i$
  is a feasible shift from $\pi_i'$.
\end{cl}
\begin{proof}
  Let $\varepsilon = \min\{\frac{\pi_{ia}'}{-\eta_{ia}}:
        a\in A,\eta_{ia} < 0\}$. Since $\supp(\pi_i) \subseteq \supp(\pi_i')$ and $\eta_i$ is a
  feasible shift from $\pi_i$, for each $a\in A$ such that $\eta_{ia}
  < 0$, we know that $\pi_{ia} > 0$ and therefore $\pi_{ia}' >
  0$. Therefore, $\varepsilon > 0$. Moreover, $-\varepsilon \eta_{ia}
  \leq \pi_{ia}'$ meaning that $\pi'_{ia} + \varepsilon \eta_{ia} \geq
  0$. Thus, for all $a\in A, \pi'_{ia} + \varepsilon \eta_{ia} \geq
  0$.
  By definition of a feasible shift, $\sum_{a\in A} (\pi_{ia}' +
  \varepsilon\eta_{ia}) = \sum_{a\in A} \pi_{ia}' +\sum_{a\in A}
  \varepsilon\eta_{ia}  = 1 + 0 =1$. Thus, for each $a\in A, \pi_{ia}' +
  \varepsilon\eta_{ia} \leq 1$, so we conclude that $\pi_{i}' +
  \varepsilon\eta_{i}\in \Delta(A)$ and therefore $\varepsilon\eta_i$
  is a feasible shift from $\pi_i'$.
\end{proof}

\begin{cl}\label{cl: shift together}
For each pair $\pi, \pi'\in \Pi$, if $\eta = \pi -
\pi'$, then for any $\varepsilon \in [0,1]$, $\pi' + \varepsilon
\eta \in \Pi$.
\end{cl}
\begin{proof}
  Let $\overline \pi = \pi' + \varepsilon \eta$.
  As argued above,  for each $i\in N$, $\sum_{a\in A}\overline
  \pi_{ia} = 1$. It remains to show that for each $a\in A$,
  $\sum_{i\in N}\overline \pi_{ia} = 1$. For each $a\in A$,

  \[
    \begin{array}{rcl}
      \sum_{i\in N}\overline \pi_{ia} & = &\sum_{i\in N} (\pi'_{ia} +
                                            \varepsilon\eta_{ia}) =\sum_{i\in N} \pi'_{ia} +
                                          \varepsilon\sum_{i\in
                                          N}\eta_{ia} = 1 +  \varepsilon\sum_{i\in
                                            N}(\pi_{ia} - \pi_{ia}')\\& =&  1 +  \varepsilon(\sum_{i\in
                                            N}\pi_{ia} - \sum_{i\in
                                            N}\pi_{ia}')=1 +  \varepsilon(1-1)=1.
    \end{array}\]$\quad$\\[-35pt]\end{proof}
The following observation is a direct consequence of \cref{cl: shift
  exists,cl: shift together}.
\begin{cl}\label{cl: shift transfer}
  For each pair $\pi, \pi'\in \Pi$, such that for each $i\in N$,
  $\supp(\pi_i') \subseteq \supp(\pi_i)$, if $\eta = (\eta_i)_{i\in N}$ is a
  profile of feasible shifts from $\pi_i'$ such that $\pi' + \eta \in
  \Pi$, then there is $\varepsilon > 0$ such that $\pi+
  \varepsilon\eta\in \Pi$.
\end{cl}

We now prove \cref{prop:smallersupport}.
 Let $P\in\mathcal P^N$ and let $\pi\in\Pi$ be an \strongefficient
 allocation at~$P$. Let $\pi'\in \Pi$ be such that for each $i\in N$,
  $\supp(\pi_i') \subseteq \supp(\pi_i)$.  If $\pi'$ is not \strongefficient at~$P$, then there is $\overline \pi\in   \Pi\setminus \{\pi'\}$
  such that for each $i\in N$, it is not the case that  $\pi_i' \pisd
  \overline \pi_i$. For each $i\in N$, let $\eta_i = \overline \pi_i -
  \pi_i'$. Therefore, for each $i\in N$, $\eta_i$ is not sd-worsening
  at $P_i$.  By \cref{cl: shift transfer}, there is $\varepsilon > 0$
  such that $\hat \pi = \pi + \varepsilon \eta\in \Pi$. Since, for each
  $i\in N$, $\eta_i$  is not sd-worsening at $P_i$, it
  is not the case that   $\pi_i\pisd \hat\pi_i$. This means $\pi$ is
  not \strongefficient, a contradiction.
\end{proof}

\begin{proof}[Proof of \Cref{prop: adjacent-two}]
  Suppose $\rp^\lambda$ is unambiguously efficient. 
  We  prove that $\lambda$ randomizes over an adjacent-2 set or places
  full probability on a single ordering. Consider
  a pair $\succ,\succ'\in \supp(\lambda)$ that coincide up to the
  $(t-1)$-th dictator but differ in the $t$-th dictator (where $t\in
  \{1,\cdots,n-1\}$). If there is no such pair, $\lambda$ is a
  degenerate lottery over $\mathcal O$ and  we are done. If there are
  multiple pairs of such $\succ$ and $\succ'$, choose a pair with the
  smallest $t$.
  Let $i$ and $j$ be the $t$-th dictator at $\succ$ and $\succ'$,
  respectively. Suppose that agent $k\neq i$ is the $(t+1)$-th
  dictator at $\succ'$ as follows:\renewcommand{\arraystretch}{0.8}
  \begin{center}
    ~~~~~~~~$t^\text{th}$ dictator\\
    $\left(
      \begin{array}{llll} & & & \downarrow \\
        \succ& :&\cdots,&\underline{i},
                          \cdots\\
        \succ'&:& \cdots,&\underline{j},~k,\cdots,i,\cdots
      \end{array}
    \right)$
  \end{center}

  Let $P\in \mathcal P^N$ be a profile of identical
  preferences. Let $a$ and $b$  be the best two remaining objects, in
  that order, after
  the first $(t-1)$ agents have been allocated lotteries. Then $i$
  receives
  $a$ with positive probability and some $c$ that every agent ranks
  below $b$. Moreover, $k$ receives $b$ with positive
  probability. This contradicts \strongefficiency, because, with all other agents' lotteries fixed, agents~$i$ and $k$ may exchange probabilities of~$a$, $c$ and $b$ between them to obtain lotteries that are incomparable to what they initially had. Therefore, $i$ must be the
$(t+1)$-th dictator at $\succ'$ and by a symmetric argument, $j$
must be the $(t+1)$-th dictator at $\succ$.
If $\succ$ and $\succ'$ coincide from the $(t + 2)$-th dictator, we are done. 
If not, choose the smallest $s \geq t + 2$ at which $\succ$ and $\succ'$ do not coincide. 
The same argument as above applies to show that $s$-th dictator and $(s + 1)$-th dictator 
are adjacent at $\succ$ and $\succ'$. By continuing this process 
along $\succ$ and $\succ'$, we conclude that $\succ$ and $\succ'$ can differ only 
for adjacent pairs of agents.

Next, we prove that there can be only one such pair. Suppose otherwise that there 
are two disjoint pairs of adjacent dictators who are switched, so that for some $s < t$,
\begin{center}~~~~~~~~$s^\text{th}$ dictator~~~~~$t^\text{th}$
  dictator\\
  $\left(
    \begin{array}{lllll}
                     & & & \downarrow  & \downarrow \\
      \succ&: &\cdots,&\underline{i},~j
                        \cdots,&~\underline{k},~l,\cdots\\
      \succ'&: & \cdots,&\underline{j},~i,\cdots,&~\underline{l},~k,\cdots
    \end{array}
  \right)$.
\end{center}
Let $P\in\mathcal P^{N}$ be such
  that $i$ ranks $a$ first and $b$ second, $j$ ranks $a$ first and
  $c$ second,  both $k$ and $l$ rank $b$, $c$, and $d$ in that order,
  and every other agent ranks $a,b,c,$ and $d$ below all other objects. 
  Then $\rp(\lambda)$ assigns $k$ positive probabilities of $b$ and $d$
  and $l$ a positive probability of $c$, contradicting \cref{lem: no
    gaps with common prefs} for the triple $b$, $c$, $d$, which $k$ and $l$ rank alike. 
  Hence, the support of $\lambda$ is a singleton or an adjacent-2 set.

  The converse follows from \cref{thm: characterization} because if $\lambda$'s support is contained in an
  adjacent-2 set, then it is an HMD.
\end{proof}

\section{Relationships Between Axioms}
\label{sec:axiom-relations}
\begin{lem}
Suppose $|N| = |A| = 3$. If a mechanism $\varphi: \mathcal P^N \to \Pi$
is \sdsp, non-bossy, neutral, and \sdefficient, 
then it is boundedly invariant.
\end{lem}

\begin{proof}
    
    The desired conclusion obtains from the following two Claims.\\

    \noindent\emph{Claim 1:} Let $P \in \mathcal P^N$. If there exist $i \in N$, $a \in A$,
    and $P'_i \in \mathcal P$ such that $\varphi_a(P) \neq \varphi_a(P'_i, P_{-i})$
    and $P'_i$ is a lower-set reshuffling at $a$ of $P_i$, then all agents
    top-rank $a$ under $P$.
    \begin{proof}
        Assume the hypothesis.

        If $a$ is not top-ranked under $P_i$, then, since there are only 
        three objects in $A$, $P_i = P'_i$, contradicting the hypothesis.
        Thus, $a$ is top-ranked under $P_i$. 
        Denote the other two objects by $b$ and $c$ so that 
        $a \mathrel P_i b \mathrel  P_i c$ and $a \mathrel  P'_i c \mathrel  P'_i b$.

        Let $\pi = \varphi(P)$ and $\pi' = \varphi(P'_i, P_{-i})$.

        By \sdspness, $\pi_{ia} = \pi'_{ia} = \alpha$ for some 
        $\alpha \in [0,1]$. 
        If $\alpha = 1$, feasibility implies $\pi_{ha} = \pi'_{ha} = 0$ 
        for every $h \ne i$, contradicting the hypothesis. 
        Thus, $\alpha < 1$.

        If $\pi_{ib} = 0$, then $P'_i$ is a support monotonic transformation
        of $P_i$ at $\pi_i$. By \sdspness, non-bossiness, and
        Lemma~\ref{lem: monotonicity}, $\pi' = \pi$, again contradicting the
        hypothesis.
        A symmetric argument applies if $\pi'_{ic} = 0$.
        Thus, $\pi_{ib} > 0$ and $\pi'_{ic} > 0$.
        For any agent $h \ne i$, if $\pi_{ha} > 0$, ambiguous efficiency implies
        $a  \mathrel P_h b$; otherwise, $h$ could give up some probability of $a$
        for $b$, while $i$ makes the reverse trade, yielding an
        sd-improvement. 
        By a symmetric argument, $\pi'_{ha} > 0$ implies $a  \mathrel P_h c$.

        Since $i$'s probability of receiving $a$ is unchanged, feasibility
        requires that both the other agents' probabilities of receiving $a$
        change. 
        Name these agents $j$ and $k$ so that $\pi_{ja} > \pi'_{ja}$ and 
        $\pi_{ka} < \pi'_{ka}$.
        In particular, $\pi_{ja} > 0$ and $\pi'_{ka} > 0$, so by the
        argument above, $a  \mathrel P_j b$ and $a \mathrel  P_k c$.

        Since $\pi_{ja} > 0$, feasibility implies $\pi_{jc} < 1$.
        Thus, some agent $h \in \{i, k\}$ has $\pi_{hc} > 0$.
        Both $i$ and $k$ prefer $a$ to $c$ at $P$, so ambiguous efficiency
        requires $a  \mathrel P_j c$; otherwise, $j$ could exchange some probability
        of $a$ for $c$ with $h$, yielding an sd-improvement at $\pi$.

        Similarly, $\pi'_{ka} > 0$ implies $\pi'_{kb} < 1$, so some agent
        $h' \in \{i,j\}$ has $\pi'_{h'b} > 0$.
        Both $i$ and $j$ prefer $a$ to $b$ at $(P'_i,P_{-i})$, so
        ambiguous efficiency requires $a \mathrel  P_k b$; otherwise $k$ could 
        exchange some probability of $a$ for $b$ with $h'$,
        yielding an sd-improvement at $\pi'$.

        Thus, all three agents rank $a$ first at $P$.
    \end{proof}

    For arbitrary object $a$, define $\mathcal P^a$ to be the subset of all 
    preference orders in $\mathcal P$ with $a$ top-ranked.\\

    \noindent\emph{Claim 2:} The allocation of $a$ is invariant on $(\mathcal P^a)^N$.
    \begin{proof}
    Since there are only two objects besides $a$, $\mathcal P^a$ contains
    exactly two preference orders.
    
    For any agent $h$, profile $P \in (\mathcal P^a)^N$ and preference 
    $P'_h \in \mathcal P^a$, $\varphi_{ha}(P) = \varphi_{ha}(P'_h, P_{-h}) = \alpha_h$ 
    by \sdspness, where $\alpha_h \in [0,1]$ could only vary 
    with $h$ and $P_{-h}$ but not $P_h$.
    By neutrality, the only feature of $P_{-h}$ with which $\alpha_h$ 
    could vary is whether the preferences $P_j$ and $P_k$ of 
    distinct $j, k \in N \setminus \{h\}$ are identical or not.
    Denote the value of $\alpha_h$ by $x_h$ in the identical-preferences 
    case and $y_h$ in the other case.

    Feasibility at a profile $P \in (\mathcal P^a)^N$ yields
    \[
        \sum_{h \in N} x_h = 1
    \]
    when all agents have the same preference order at $P$,
    and yields
    \[
        x_i + \sum_{h \neq i} y_h = 1
    \]
    when an arbitrary agent $i$ has a preference order that is
    distinct from the (necessarily identical) preference orders of 
    the other two agents.
    
    Subtracting the first displayed equation from the second, we obtain
    \[
    \sum_{h \neq i} (y_h - x_h) = 0,
    \]
    for every $i \in N$.
    But this can only hold if, for every agent $i$, $y_i - x_i = 0$.
    \end{proof}

    If bounded invariance were violated by a lower-set reshuffling
    at an arbitrary object $a$, Claim 1 would imply that all agents rank 
    $a$ first. Both profiles would therefore belong to $(\mathcal P^a)^N$, 
    contradicting Claim 2.
\end{proof}

\section{Formal Definition of HMD}\label{apx:formal alg}
\subsection{A Recursive Algorithm}
\label{sec: recursive}

Before defining the general algorithm,  we give some preliminary
definitions.
A \textbf{supply vector} is  $s \in [0,1]^A$
such that $\sum_{a\in A} s_a$ is a non-negative integer, where each 
coordinate $s_a$ is the probability of object $a$ to be allocated.
Let $S$ denote the set of all such vectors.
Denote by $A(s)$ the set of objects available in positive
amounts at $s$ and by $\underline A(s)$ those that are only
fractionally available. {That is, $A(s) = \{a\in
  A: s_a>0\}$ and $\underline A(s) = \{a\in
  A: 0 < s_a  < 1\}$.}
We write $\vec 0$ and $\vec 1$ for the zero and unit vectors in $S$.
A \textbf{partial allocation} is an $|N|\times|A|$ substochastic
matrix $\pi$ such that for each $i \in N$, either $\pi_i = \vec 0$ or
$\pi_i \in \Delta(A)$.
Let {\boldmath$\overline \Pi$} be the set of all
partial allocations. A partial allocation $\pi\in \overline \Pi$  is
\textbf{feasible} at   $s$ if $\sum_{i\in N}\pi_i \leq s$.
A \textbf{generalized  \rul} is a mapping $f:S\times
\mathcal P^N \to \overline\Pi$ such that for each $s\in S$ and
  $P\in\mathcal
  P^N$, $f(s,P)\in \overline \Pi$ is feasible at $s$ and allocates
  a lottery in $\Delta(A)$ to some agent when $s \neq \vec 0$.
 Let  {\boldmath$\mathcal F$} be the set of all generalized \ruls.

If we start with a supply vector $s\in S$, the residual supply after
allocating $\pi\in\overline \Pi$ is $r(s,\pi)\in S$ such that for
each $a\in A$,  {\boldmath$r_a(s,\pi)$}$\ = s_a - \sum_{i\in N}\pi_{ia}$.

The algorithm uses generalized \ruls to compute a
complete allocation part by part.
It  keeps track of these parts via a \textbf{history}, which is a
(possibly empty) list of
partial allocations $\eta = (\pi^1, \dots, \pi^K)$ such that for each
$k = 1,\dots K$, $\pi^k \in \overline \Pi$ and
$\sum_{k=1}^K \pi^k \in \overline \Pi$. Let  {\boldmath$N(\eta)$}$\ =\{i\in N: \text{
  there is }k\text{ such that }\pi^k_i \neq  \vec 0\}$ be the set of agents
allocated to along the history $\eta$.
Having collected the history
$\eta$, the next partial allocation can be any $\pi^{K+1}\in \overline
\Pi$ such that for each $i\in N(\eta), \pi_i^{K+1} = \vec 0$. That is, only partial allocations that
allocate to agents not allocated to by any part of $\eta$ can occur
after $\eta$. We append $\pi^{K+1}$ to form a new history $(\pi^1,
\dots, \pi^K, \pi^{K+1})$ and denote it by $\eta$
{\boldmath$\doubleplus$} $\pi^{K+1}$.
Let  {\boldmath$\mathcal H$} be the set of all histories.
The set of  \textbf{terminal histories} is  {\boldmath$\mathcal H^T$}$\ = \{\eta\in
\mathcal H: N(\eta) = N\}$. {Note that for each $(\pi^1,
  \dots, \pi^K)\in \mathcal   H^T, \sum_{k=1}^K \pi^k \in \Pi$ is an
  allocation, and not just a   partial one.} The set of
\textbf{non-terminal histories} is
 {\boldmath$\mathcal H^{NT}$}$\ =
\mathcal H\setminus \mathcal H^T$.
A \textbf{sequencing rule}, $\sigma:\mathcal H^{NT} \to \mathcal F$,
maps each non-terminal history to a partial   \rul. We require that
for each  $\eta \in \mathcal H^{NT}$, each $i\in N(\eta)$, and each $\pi\in
\text{range}(\sigma(\eta))$, $\pi_i = \vec 0$. That is, $\sigma(\eta)$
does not  allocate to agents who have already been allocated to along
$\eta$. To summarize:
\begin{enumerate}
\item Each $f\in \mathcal F$ solves a part of the allocation problem.\\[-25pt]
\item A history $\eta\in \mathcal H$ keeps track of accumulated  parts
  of the solution.\\[-25pt]
\item The sequencing rule $\sigma$  tells us what partial \rul  solves
  the next part.\\[-15pt]
\end{enumerate}

We now put these components together to recursively  compute an
allocation for any preference profile.
Define the function $\Psi^\sigma: \mathcal H\times S \times\mathcal P^N
\to \overline\Pi$ by setting, for each $(\eta, s, P) \in \mathcal
H\times S\times \mathcal P^N$, \renewcommand{\arraystretch}{0.8}
\[
  \Psi^\sigma(\eta, s, P) = \left\{
    \begin{array}{ll}
      (\vec 0, \dots, \vec 0) &~~\text{if }\eta\in \mathcal H^T;\\

      \pi + \Psi^\sigma(\eta\doubleplus \pi,  r(s, \pi), P) & ~~\text{if }\eta\in \mathcal H^{NT}, \text{ where } \pi = \sigma(\eta)(s, P).\\
    \end{array}
    \right.
\]

If we start with the empty history and the full supply vector
then we have a \rul for each  sequencing rule $\sigma$. We denote this
\rul by  $\varphi^\sigma$, which we define by setting  for each $P\in
\mathcal P^N$, $\varphi^{\sigma}(P) = \Psi^\sigma((), \vec 1, P)$.

To fix ideas, we demonstrate how one can describe a serial dictatorship
in the recursive language we have laid out above.
\begin{ex}[Recursive formulation of serial dictatorship]{\em
In \cref{sec:random
  dict}, we described serial dictatorship
allocations where agents select
their best objects according to a priority order. We now articulate
these as \ruls in the language of this section.
Consider  $\succ\in \mathcal O$ and let  $1, \dots, n$ be a labeling
of the agents in decreasing order of their priority according to
$\succ$. For each history $\eta\in \mathcal H$, let $\len(\eta)$ be the length
of $\eta$. {That is,  $\len(()) = 0$ and  $\len(\pi^1,\dots,
  \pi^K) = K$}.
Let $\sigma:\mathcal H^{NT} \to \mathcal F$ be such that for each
$\eta\in \mathcal H^{NT}$, $\sigma(\eta) = f$ where  $f(s, P)$ allocates to
$i= \len(\eta)+1$ the best possible lottery according to $P_i$ given
the  supply $s$.

As a concrete example, take  $N = \{1,2,3,4\}$ and $A=
\{a,b,c,d\}$ and the following preference profile $P\in \mathcal P^N$:
\[
  \begin{array}{cccc}
    P_1&P_2&P_3&P_4\\
    \hline\\[-10pt]
    a & a & b & a   \\
    d &  b & c& b  \\
    c &  d& a & c  \\
    b &  c & d & d \\
  \end{array}
\]
We write down the input  history and residual supply as well as the
calculated partial allocation and return value of each recursive call:
\begin{footnotesize}
{\renewcommand{\arraystretch}{1.3}
  \begin{longtable}{|c||c|c|c|}
    \hline
  Call&Input ($\eta, s$)& {\renewcommand{\arraystretch}{0.9}  \begin{tabular}{c}Calculated partial\\ allocation ($\pi$)\end{tabular}} &
                                                                  {\renewcommand{\arraystretch}{0.9}  \begin{tabular}{c}Return
                                                                    value\\
                                                                  ($\pi
    + \Psi^\sigma(\eta \doubleplus \pi,  r(s, \pi), P)$ or $(\vec{0},\dots,\vec{0})$) \end{tabular}}\\\hline
  \hline
  1 &  $(), \vec{1}$ & $(a, \vec 0,\vec 0,\vec 0)$ & $(a,\vec 0,\vec
                                                     0,\vec 0) +
                                                     \Psi^\sigma((a,\vec
                                                     0,\vec 0,\vec 0),
                                                     (0,1,1,1), P)$\\
  \hline
  2 & $((a,\vec 0,\vec 0,\vec 0)), (0,1,1,1)$ & $(\vec 0,b,\vec 0,\vec
                                                0) $ &
                                                       {\renewcommand{\arraystretch}{0.9} $(\vec 0,b,\vec 0,\vec 0) + \Psi^\sigma\left(\begin{array}{c}((a,\vec 0,\vec 0,\vec 0),\\(\vec 0,b,\vec 0,\vec 0)),\end{array} (0,0,1,1), P\right)$}\\
  \hline
  3 & {\renewcommand{\arraystretch}{0.9} $\begin{array}{c}((a,\vec 0,\vec 0,\vec 0),\\(\vec 0,b,\vec 0,\vec 0)),\end{array} (0,0,1,1)$}
      & $(\vec 0,\vec 0,c,\vec 0) $ &
                                      {\renewcommand{\arraystretch}{1}
                                      $(\vec 0,\vec 0,c,\vec 0)
                                      +\Psi^\sigma\left(\begin{array}{c}((a,\vec
                                                   0,\vec 0,\vec
                                                   0),\\(\vec 0,b,\vec
                                                   0,\vec 0),\\(\vec
                                                   0,\vec 0,c,\vec
                                                   0)),\end{array}
    (0,0,0,1), P\right)$}\\
  \hline
  4 & {\renewcommand{\arraystretch}{0.9} $\begin{array}{c}((a,\vec 0,\vec 0,\vec 0),\\(\vec 0,b,\vec 0,\vec 0),\\(\vec 0,\vec 0,c,\vec 0)),\end{array}
    (0,0,0,1)$} &$(\vec 0,\vec 0,\vec 0,d) $ &  {\renewcommand{\arraystretch}{1}$(\vec 0,\vec 0,\vec 0,d) +\Psi^\sigma\left(\begin{array}{c}((a,\vec 0,\vec 0,\vec 0),\\(\vec 0,b,\vec 0,\vec 0),\\(\vec 0,\vec 0,c,\vec 0),\\ (\vec 0,\vec 0,\vec 0,d)),\end{array} \vec 0,P\right)$}\\
  \hline
  5 & {\renewcommand{\arraystretch}{1} $\begin{array}{c}((a,\vec 0,\vec 0,\vec 0),\\(\vec 0,b,\vec 0,\vec 0),\\(\vec 0,\vec 0,c,\vec 0),\\
         (\vec 0,\vec 0,\vec 0,d)),\end{array} \vec 0$} & &$(\vec{0},\dots,\vec{0})$ \\
    \hline
  \end{longtable}}\end{footnotesize}
\noindent The initial call of $\Psi^\sigma((),\vec{1},P)$ returns the final
 allocation  $\varphi^\sigma(P)=\Psi^\sigma((),\vec 1, P)=(a,b,c,d)$. \hfill$\blacksquare$
  }
\end{ex}

\subsection{Monarchical and Diarchical Partial Allocation \Ruls}

It is not for  every  $\sigma$ that $\varphi^\sigma$
satisfies all of the properties on our list. We now define conditions  on
sequencing rules that are necessary and sufficient for all of the
properties to be satisfied. We consider restrictions on the types of
generalized \ruls that can ever be chosen as well as the degree
to which the history can influence the choice of generalized \rul.

\paragraph{Permissible  partial allocation \ruls}
We first restrict the range of the sequencing rule to partial
allocation \ruls that only allocate to  one or two
agents at a time.  We  define these formally here.

We start by introducing  notation to depict an agent's most
preferred lottery among those feasible at a given supply vector. For each $i\in N$, each
$P\in \mathcal P^N$, and each $s\in S$, the \textbf{top allocation for
  \boldmath$i$ subject to
  $s$} is $\tau^i(s, P) = \pi_i$
such that  listing the objects as $a_1$ through $a_{|A|}$ in decreasing
  order of $i$'s preference $P_i$,  $\pi_{ia_1} =
 s_{a_1}, \pi_{ia_2} = \min\{1 - \pi_{ia_1},
  s_{a_2}\}, \dots,  \pi_{ia_{|A|}} = \min\{1 - (\pi_{ia_1} +
  \dots + \pi_{ia_{|A|-1}}), s_{a_{|A|}}\}$. In words, we compute
  $\pi_i$ by greedily allocating as much of each object as possible to
  $i$ in
  decreasing order of $i$'s preference. This is the
  \emph{unambiguously best} lottery that we can give $i$, given the
  supply vector $s$ as it stochastically dominates any other lottery
  feasible at $s$.

The first kind of partial allocation \ruls that our list of properties
allows allocate to a single agent.  For each $i\in N$, the
\textbf{monarchy of \boldmath$i$} is the
partial allocation \rul { \boldmath$m^i$} such that for each $(s, P) \in
S\times \mathcal P^N$ and each $j\in N$,
\[
  m^i_j(s,P) = \left\{
    \begin{array}{cl}
      \tau^i(s,P)& \text{if }j=i\\
      \vec 0&\text{otherwise.}
    \end{array}
    \right.
\]
Let {\boldmath$\mathcal M$} be the set of all monarchical partial  allocation \ruls.

The next type of partial allocation \rul allocates to a pair of
agents. For each distinct
pair  $i,j\in N$, each $P\in \mathcal P^N$, and  each $s\in
S$,  the \textbf{top allocation for \boldmath$i$
after $j$ subject to  $s$} is $\tau^{ji}(s, P) =
\tau^i(s-\tau^j(s, P), P)$. This represents the best lottery for one agent after the
other agent has already chosen from a given supply vector. Now, given a pair $i,j\in N$ and $\alpha\in (0,1)$, the
\textbf{\boldmath$\alpha$-diarchy of $i$ and $j$} is the partial
allocation \rul { \boldmath$d^{(i,j,\alpha)}$} such that for each $(s, P) \in
S\times \mathcal P^N$ and each $k\in N$,
\[
  d^{(i,j,\alpha)}_k(s,P) = \left\{
    \begin{array}{cl}
      \alpha\tau^i(s, P) + (1-\alpha)\tau^{ji}(s, P)& \text{if }k=i\\
      \alpha\tau^{ij}(s, P) + (1-\alpha)\tau^{j}(s, P)& \text{if }k=j\\
      \vec 0&\text{otherwise.}\footnotemark
    \end{array}
    \right.
  \]
  \footnotetext{If $\sum_{a \in A}s_a < 2$, then we set 
      $d_i^{(i,j,\alpha)}(s,P) = d_j^{(i,j,\alpha)}(s,P ) = \vec 0$ as well. However, 
      in what follows, we never apply diarchies to such inputs.}
  The order of the parameters matters and $d^{(i,j, \alpha)} = d^{(j,i, 1-\alpha)}$.
  Let {\boldmath$\mathcal D$} be the set of all diarchical  generalized  \ruls.

\paragraph{Restricting responsiveness to histories}
The set of all histories is larger than  the set of all  partial
allocations
in $\overline \Pi$:  a history $\eta\in \mathcal H$ specifies what,
and in what order, each agent in $N(\eta)$ gets. However, to satisfy the properties we
characterize, a sequencing rule cannot respond to all of this
information.  Below we define an equivalence between histories so that
a permissible (per our properties) sequencing rule would not distinguish
between such histories.

For each history $\eta = (\pi^k)_{k=1}^K\in \mathcal H$, let
{\boldmath$r(\eta)$}$\ = r(\vec 1, \sum_{k=1}^K \pi^k)$ be the
residual supply after allocating each part of $\eta$ from the full
supply vector. Given a sequencing rule $\sigma$ and $P\in \mathcal
P^N$, the \textbf{execution of \boldmath$\sigma$ at $P$} is the sequence of
histories $\eta^\sigma_0(P) = (\,)$ and, for each $t\geq 1$,
\[
\eta^\sigma_{t}(P) = \eta^\sigma_{t-1}(P)\doubleplus \sigma(\eta^\sigma_{t-1}(P))(r(\eta^\sigma_{t-1}(P)),
P),
\]
continued until a terminal history. A history is \textbf{reachable
under \boldmath$\sigma$} if it equals $\eta^\sigma_t(P)$ for some $P
\in \mathcal P^N$ and some $t$. The $t^{\text{th}}$ recursive call to
$\Psi^\sigma$ in the computation of $\varphi^\sigma(P)$ takes
$\eta^\sigma_{t-1}(P)$ and $r(\eta^\sigma_{t-1}(P))$ as its input and
produces $\eta^\sigma_t(P)$.

For each history $\eta = (\pi^k)_{k=1}^K$, each $t \leq K$, and each pair of distinct agents $i, j \in N$, let \[
n^t_{ij}(\eta) = \Bigl| \Bigl\{a \in A : \textstyle\sum_{k \leq t}
\pi^k_{ia} > 0 \ \text{ and } \ \sum_{k \leq t} \pi^k_{ja} > 0 \Bigr\}
\Bigr| \]
be the \emph{number} of objects of which both $i$ and $j$  receive positive
probability along the first $t$ parts of  $\eta$. Only the number of
these objects matters, not their identities.
Consider a pair of histories  $\eta =
(\pi^k)_{k=1}^K$ and $\tilde\eta = (\tilde\pi^k)_{k=1}^{\tilde K}\in
\mathcal H$. We say that they are \textbf{identity-overlap}
equivalent if the following conditions hold:
\begin{enumerate}
\item they are of the same length and allocate to the same agents at
  each part, so that $K = \tilde K$ and 
  for each $t\in \{1,\cdots,K\}$,  $ N(\pi^t) = N(\tilde\pi^t)$, and \\[-25pt]
\item each pair of agents shares positive probability of the same
  number of objects, cumulatively, along each initial part of these
  histories, so that for each $t\in \{1,\cdots,K\}$ and each pair of
  distinct agents $i, j \in N$, $n^t_{ij}(\eta) =
  n^t_{ij}(\tilde{\eta})$.\\[-15pt]
\end{enumerate}

\paragraph{Putting the restrictions together:}
We are  now ready to  define the sequencing rules that respect the
properties  we characterize.
A sequencing rule $\sigma$ is a \textbf{sequence of monarchies and
  (integral)  diarchies} if the followings hold:
\begin{enumerate}
\item it only ever involves monarchies or diarchies, so that
  $\text{range}({\sigma})\subseteq \mathcal M \cup \mathcal  D$,\\[-25pt]
\item it  maps to a monarchy if the residual supply is not integral, so
  that for each $\eta = (\pi^k)_{k=1}^K$ where $\vec{1}-\sum_{k=1}^K\sum_{i=1}^n \pi_i^k$  is not integral, $\sigma(\eta)\in \mathcal
  M$, \\[-25pt]
\item  it maps to the same partial allocation \rul for any pair of
  identity-overlap equivalent non-terminal histories that are
  reachable under $\sigma$, so that if $\eta$ and
  $\eta'\in \mathcal H^{NT}$ are reachable under $\sigma$ and identity-overlap equivalent then
  $\sigma(\eta) = \sigma(\eta')$.\\[-15pt]
\end{enumerate}
We call a sequencing rule that satisfies the first two conditions an
\textbf{MD-sequencing rule}, and we call an MD-sequencing rule that
also satisfies the third condition \textbf{identity-overlap invariant}.
If $\sigma$ is a sequence of monarchies and diarchies, then we call
$\varphi^\sigma$ a \textbf{Hierarchy of Monarchies and Diarchies}
and denote it by {\boldmath$MD^\sigma$}. Below is an example of a
sequence of monarchies and diarchies and the HMD that it defines.

\begin{ex}[A hierarchy of monarchies and diarchies]{\em
Let $N =\{1,2,3,4,5\}$ and $A = \{a,b,c,d,e\}$. We first specify the
values of  $\sigma$, which  is a sequence of monarchies and diarchies.
Notice that we have not considered all possible histories in
$\mathcal H$ but only those that could appear on an execution path. How we specify $\sigma$ at any other history than those
we have listed in the table here does not
affect the output of the algorithm.

\begin{footnotesize}
\[{\renewcommand{\arraystretch}{1.3}
  \begin{array}{|l|l|}
    \hline
    ~\eta & \sigma(\eta)\\\hline
    \hline
    ~() & d^{(1,2,\sfrac{1}{3})}\\
    \hline
    \begin{array}{l}
      (\pi^1) \text{ where }\pi^1_3 = \pi^1_4 = \vec 0
      \text{ and } \pi^1_1\text{ and }\pi^1_2 > 0  \text{ are }
      \text{integral}
    \end{array}
         &    d^{(3,4,\sfrac{1}{2})}
    \\\hline
    \begin{array}{l}
      (\pi^1) \text{ where }\pi^1_3 = \pi^1_4 = \vec 0    \text{ and } \pi^1_1\text{ and }\pi^1_2 > 0  \text{ are }
      \text{not integral}
    \end{array}
         &    m^3

    \\\hline
    \begin{array}{l}
      (\pi^1, \pi^2) \text{ where } N(\pi^1, \pi^2) =
      \{1,2,3\}\text{ and }\vec{1}-\pi^1 - \pi^2 \text{ is integral}
    \end{array}
         &    d^{(4,5,\sfrac{1}{2})}
    \\\hline    \begin{array}{l}
      (\pi^1, \pi^2) \text{ where } N(\pi^1, \pi^2) =
      \{1,2,3\}\text{ and }\vec{1}-\pi^1 - \pi^2 \text{ is not integral}
    \end{array}
         &    m^4
    \\\hline
    \begin{array}{l}
      (\pi^1, \pi^2) \text{ where } N(\pi^1, \pi^2) = \{1,2,3, 4\}
      \text{ or }       (\pi^1, \pi^2, \pi^3) \text{ where } N(\pi^1, \pi^2,\pi^3) = \{1,2,3, 4\}
    \end{array}
         &    m^5
    \\\hline

  \end{array}}
\]\end{footnotesize}

We compute the allocation chosen by $\md^\sigma$ for the following $P\in \mathcal P^N$:\renewcommand{\arraystretch}{0.7}
\[
  \begin{array}{ccccc}
 P_1&P_2&P_3&P_4 & P_5\\
    \hline\\[-10pt]
     a & a & b & a  &b \\
     e&  b & e& e & a  \\
     d&  d& c & b & c \\
     c&  e & a & c & e \\
     b&  c & d & d & d \\
  \end{array}
\]
The algorithm proceeds as follows:
\begin{enumerate}
\item Starting with the empty history, it allocates from
  the initial supply vector of $\vec 1$ according to the
  $\sfrac{1}{3}$-diarchy  of agents 1 and 2. So, agent 1 receives
  $\sfrac{1}{3}$ of $a$ and $\sfrac{2}{3}$ of $e$ and 2 receives
  $\sfrac{2}{3}$ of $a$ and $\sfrac{1}{3}$ of $b$. We denote the
  partial allocation
  where agents 1 and 2 get these lotteries and all others get $\vec 0$ by
  $\pi^1$.
\item The history after the first step is $(\pi^1)$. The residual
  supply is not integral, so  the table above tells us that $\sigma$
  maps this history to the monarchy of agent 3. All of $a$ has already
  been allocated to agents 1 and 2 and only $\sfrac{2}{3}$ of agent
  3's most preferred object $b$ remains. So, agent 3 receives
  $\sfrac{2}{3}$ of $b$ and $\sfrac{1}{3}$ of their next most
  preferred object, which is $e$. We denote the partial allocation
  where agent 3 gets this lottery and all others get $\vec 0$ by $\pi^2$.
\item  The residual supply at the history $(\pi^1, \pi^2)$ is
  integral. So, the table tells us the next partial allocation \rul is
  the $\sfrac{1}{2}$-diarchy of agents 4 and 5.  The only remaining
  objects are $c$ and $d$. Since agents 4 and 5 rank these two objects
  in the same order, each of them receives $\frac{1}{2}$ of each of
  these objects.
\item Since all agents have been allocated to, the algorithm stops.
\end{enumerate}

We write down the input  history,  residual supply, partial allocation
\rul,
calculated partial allocation, and return value of each recursive
call in the table below. In it, {we denote a  lottery that places $\alpha$ probability
  on $a$, $\beta$ probability on $b$, $\gamma$ probability on $c$,
  $\delta$ probability on $d$, and $\epsilon$ probability on $e$ by
  $\alpha a + \beta b + \gamma c +  \delta d + \epsilon e$.}\\

\begin{footnotesize}\hspace{-0.3in}  {\renewcommand{\arraystretch}{1.2}\begin{tabular}{|c||c|c|c|c|}
    \hline
  Call&Input ($\eta, s$)&$\sigma(\eta)$ &\begin{tabular}{c}Calculated partial\\[-4pt] allocation ($\pi$)\end{tabular} &
                                                                  \begin{tabular}{c}Return
                                                                    value\\[-4pt]
                                                                  ($\pi
    + \Psi^\sigma(\eta \doubleplus \pi, r(s, \pi),P)$ or $(\vec{0},\dots,\vec{0})$) \end{tabular}\\\hline
  \hline
                    1 &  $(), \vec{1}$ & $d^{(1,2,\sfrac{1}{3})}$&
                                                               \begin{tabular}{c}
                                                                 $(\sfrac{1}{3}a+\sfrac{2}{3}e, \sfrac{2}{3}a + \sfrac{1}{3}b,\vec 0, \vec 0, \vec 0)$\\
                                                                 ($\pi^1$)
                                                               \end{tabular} &
                                                                               \begin{tabular}{ll}
                                                                                 &$ (\sfrac{1}{3}a+\sfrac{2}{3}e, \sfrac{2}{3}a+\sfrac{1}{3}b,\vec 0,\vec 0,\vec 0)$\\
&$+$ $\Psi^\sigma((\pi^1), (0,\sfrac{2}{3},1,1,\sfrac{1}{3}),P)$
                                                                               \end{tabular}\\
  \hline
  2 & $(\pi^1), (0,\sfrac{2}{3},1,1,\sfrac{1}{3})$ & $m^3$&  \begin{tabular}{c}
                                                                 $(\vec0,\vec{0}, \sfrac{2}{3}b+\sfrac{1}{3}e,\vec0,\vec 0) $\\
                                                                 ($\pi^2$)
                                                               \end{tabular}
                    &   \begin{tabular}{ll}
                          &$(\vec0,\vec{0}, \sfrac{2}{3}b+\sfrac{1}{3}e,\vec0,\vec 0) $\\
&$+$ $\Psi^\sigma((\pi^1, \pi^2), (0,0,1,1,0),P)$
                        \end{tabular}\\
  \hline
  3 & $(\pi^1, \pi^2), (0,0,1,1,0)$ & $d^{(4,5,\sfrac{1}{2})}$
                                        &     \begin{tabular}{c}
                                                $(\vec 0, \vec 0, \vec 0, \sfrac{1}{2}c+\sfrac{1}{2}d, \sfrac{1}{2}c+\sfrac{1}{2}d)$\\
                                               ($\pi^3$)
                                             \end{tabular}
                    &   \begin{tabular}{ll}
                          &$(\vec 0, \vec 0, \vec 0, \sfrac{1}{2}c+\sfrac{1}{2}d, \sfrac{1}{2}c+\sfrac{1}{2}d)$\\
&$+$ $\Psi^\sigma((\pi^1, \pi^2, \pi^3), \vec 0,P)$
                        \end{tabular}\\
  \hline
  4 & $(\pi^1, \pi^2, \pi^3), \vec 0$ & & &$(\vec{0},\vec 0,\vec 0,\vec 0,\vec 0) $ \\
                 \hline
  \end{tabular}}\end{footnotesize}

$\quad$\\
The final allocation is calculated as \begin{center}~~~$\md^\sigma(P) =
  \Psi^\sigma((),\vec 1,P)= (\sfrac{1}{3}a+\sfrac{2}{3}e,
  \sfrac{2}{3}a+\sfrac{1}{3}b,\sfrac{2}{3}b+\sfrac{1}{3}e,
  \sfrac{1}{2}c+\sfrac{1}{2}d,
  \sfrac{1}{2}c+\sfrac{1}{2}d)$. \hfill$\blacksquare$ \end{center}}\end{ex}

\section{Proofs from \cref{sec: characterization}}\label{ap:main-proofs}
Before proving \cref{thm: characterization}, we prove a few lemmata.

\begin{lem}\label{lem: no gaps with common prefs}
  Let $P\in \mathcal P^N$ and  $\pi\in\Pi$ be \strongefficient
  at $P$. Consider a distinct pair  $i,j\in N$ and a
  triple $a,b,c\in
  A$ such that $a\mathrel P_i b \mathrel P_i c$ and  $a\mathrel P_j b
  \mathrel P_j c$. If $\pi_{jb} > 0$ then either $\pi_{ia} = 0$ or
  $\pi_{ic} = 0$.
\end{lem}
\begin{proof}
  Suppose, to the contrary, that $\pi_{ia}, \pi_{ic}> 0$. Let
  $\varepsilon >0$ be such that $\varepsilon < \min\{\pi_{ia},
  \pi_{jb}, \pi_{ic}\}$. Consider $\pi'\in \Pi$ such that for all
  $(k,d)\in (N\times A)\setminus (\{i,j\}\times\{a,b,c\}), \pi_{kd}' =
  \pi_{kd}$, $\pi_{ia}' = \pi_{ia} - \frac{\varepsilon}{2}$,
  $\pi_{ib}' = \pi_{ib} + \varepsilon$,  $\pi_{ic}' = \pi_{ic} -
  \frac{\varepsilon}{2}$,   $\pi_{ja}' = \pi_{ja} + \frac{\varepsilon}{2}$,
  $\pi_{jb}' = \pi_{jb} - \varepsilon$, and  $\pi_{jc}' = \pi_{jc} +
  \frac{\varepsilon}{2}$. Then, for all $k\in N\setminus \{i,j\}$,
  $\pi_k' = \pi_k$ and for all $k\in \{i,j\}$, $\pi_k' \mathrel {N_k^{sd}} \pi_k$, a
  contradiction to \strongefficiency.
\end{proof}

For each supply vector $s\in S$, let $\integ(s)=1$ if for each $a\in A$,
$s_a\in \{0,1\}$ and $\integ(s)=0$ if there is $a\in A$ such that $s_a\in (0,1)$.

\begin{lem}\label{lem: int supply means there exists mon or di}
  Let $P\in \mathcal P^N$ and  $\pi\in\Pi$ be \strongefficient at $P$.
  Let nonempty $N'\subseteq N$ and  $s\in S$ be such that
  \begin{enumerate}

  \item   $|N'| = \sum_{a\in A}s_a$ and $\sum_{i\in N'} \pi_{i} = s$ and
  \item for each pair   $i,j\in N', P_i|_{A(s)} = P_j|_{A(s)}$.
  \end{enumerate}
  If $\integ(s) = 1$, then either there is $i\in N'$ such that $\pi_i =
  m_i^i(s,P)$  or a pair $i,j\in N'$ and an $\alpha\in (0,1)$
  such that for each $k\in \{i,j\}$,  $\pi_k = d^{(i,j,\alpha)}_k(s, P)$.
\end{lem}
\begin{proof}
  If $|A(s)|= 1$, then we are done. Otherwise, let $a_1$ and $a_2$ be the best and second best objects in $A(s)$ according to,
  for each $k\in N'$,
  $P_k|_{A(s)}$.
  Let $i\in N'$ be such that $\pi_{ia_1}  > 0$. If $\pi_{ia_1} = 1$,
  we are done. So, suppose $\pi_{ia_1} = \alpha \in (0,1)$. By
  feasibility, there is $j\in N'$ such that $\pi_{ja_1} > 0$. Suppose
  $\pi_{ja_1} < 1-\alpha$. Then there is $k\in N'\setminus \{i,j\}$
  such that $ \pi_{ka_1} > 0$. By feasibility, for at least
  one member  $l\in \{i,j,k\}$, $\pi_{la_1} + \pi_{la_2} < 1$ meaning
  there is  $a_3\in A(s)$ such that $\pi_{la_3} > 0$. Moreover,
  feasibility ensures that for some $m\in N'\setminus \{l\}$,
  $\pi_{ma_2} > 0$. However, since $P_{m}|_{A(s)} =
  P_{l}|_{A(s)} $, this contradicts \cref{lem: no gaps with common prefs}. Thus, $\pi_{ja_1} = 1-\alpha$.
  Now, suppose $\pi_{ia_2} < 1-\alpha$. Then, by feasibility, there are
  $k\in N'\setminus \{i\}$ and $a_3\in A(s)\setminus \{a_1, a_2\}$
  such that $\pi_{ka_2},\pi_{ia_3} > 0$. Again, since $P_{k}|_{A(s)} =
  P_{i}|_{A(s)} $, this contradicts \cref{lem: no gaps with common prefs}. Thus, $\pi_{ia_2} = 1-\alpha$. By the same  reasoning, we
  conclude that $\pi_{ja_2} = \alpha$.
\end{proof}

\cref{lem: int supply means there exists mon or di} enables us to back out a monarchy or a diarchy as
a partial allocation \rul from a given \strongefficient \rul when the residual
supplies are integral. We consider the case where they are not
integral in the proof of \cref{lem: step to recurse}  directly. Next, we show implications of \sdspness and non-bossiness.

\begin{lem}\label{lem: individual monotonicity}
  Let $\varphi$ be a \sdsp \rul,  $P\in
  \mathcal P^N$, $i\in
  N$, and $P_i'\in\mathcal P$. If for
  each $a\in A$ such that $\varphi_{ia}(P) > 0$,  $\{b\in A: b\mathrel {P_i'} a\}\subseteq     \{b\in A: b \mathrel P_i a\},$ then $\varphi_i(P_i', P_{-i}) = \varphi_i(P)$.
\end{lem}
\begin{proof}
  Suppose $\varphi_i(P_i', P_{-i}) \neq \varphi_i(P)$. By \sdspness, $\varphi_i(P_i, P_{-i}) \pisd \varphi_i(P_i', P_{-i})$. But the
  definition of $P_i'$ in the statement of the lemma  ensures that this
  implies $\varphi_i(P_i, P_{-i}) \mathrel {P_i'}^{sd}
  \varphi_i(P_i', P_{-i})$. This contradicts  \sdspness of
  $\varphi$ since at true preference $P_i'$, $i$ gains by misreporting
  $P_i$ when others' preferences are $P_{-i}$.\end{proof}

\begin{lem}\label{lem: monotonicity}
  Let $\varphi$ be a \sdsp and non-bossy \rul and
  $P, P'\in
  \mathcal
  P^N$. If for
  each $i\in N$ and each  $a\in A$ such that $\varphi_{ia}(P) > 0$, $
    \{b\in A: b\mathrel {P_i'} a\}\subseteq     \{b\in A: b \mathrel P_i a\},$
  then $\varphi(P') = \varphi(P)$.
\end{lem}
\begin{proof}
  By \cref{lem: individual monotonicity}, $\varphi_1(P_1', P_{-1}) =
  \varphi_1(P)$. By non-bossiness, $\varphi(P_1', P_{-1}) =
  \varphi(P)$. Repeating this argument $(n-1)$ more times, we conclude
  that $\varphi(P') = \varphi(P)$.
\end{proof}

The property of $\varphi$  described in \cref{lem: monotonicity} is
\textbf{invariance to support monotonic transformations} and for each $i\in N$,
 $P'_i$ is a \textbf{support monotonic transformation} of
$P_i$ at $\varphi_i(P)$ \citep{HeoManjunathSCW2017}.\medskip

Next, given nonempty $N'\subseteq N$, let $\Pi^{N'} =\{\pi\in\overline \Pi: \text{ for each }
i\in N\setminus N', \pi_i = \vec 0\}$  be the set of allocations for
$N'$.
For each $\kappa = 1,\dots, n$, let $S(\kappa) = \{s\in S:\sum_{a\in A}s_a
= \kappa\}$. 
We say that $f:S(|N'|)\times \mathcal P^{N'}\to \Pi^{N'}$ is a
\textbf{generalized \rul for $N'$} if for each $s\in S(|N'|)$ and each $P\in \mathcal{P}^{N'}$, $\sum_{i\in N'}f_i(s,P)=s$.
We extend each  axiom from \cref{sec:strat-proofn-other}, other than neutrality, to  $f$ by requiring that, for each $s \in S(|N'|)$, the axiom holds for $f(s, \cdot)$.
Unambiguous efficiency is defined with respect to the set of all $\pi' \in \Pi^{N'}$ such that $\sum_{i \in N'} \pi_i' = s$.
We say that $f$ is \emph{neutral} if for each $(s, P) \in S(|N'|) \times \mathcal P^{N'}$, each permutation $\rho: A \to A$ such that $\rho(a) = a$
for each $a$ with $s_a < 1$, and each $i \in N'$, we have $f_i(s, \rho(P)) = \rho(f_i(s, P))$.
That is, only fully available objects can be relabelled.
At $s = \vec 1$, this coincides with neutrality as defined in \cref{sec:strat-proofn-other}.
All of the lemmata we have proved so far apply to generalized \ruls as well.

We first show that for such generalized \ruls, \sdspness and
non-bossiness ensure that the mechanism only considers agents'
preferences over objects with positive supply. In the
arguments that follow, this allows us to move objects with zero
availability arbitrarily in agents' preferences
without changing the allocation, by applying  bounded invariance and
invariance to support monotonic transformations.

\begin{lem}
    \label{lem:ignore-exhausted}
Let $N'\subseteq N$ and consider a generalized  \rul $f:S(|N'|) \times
\mathcal P^{N'}\to \Pi^{N'}$  for $N'$  that is  \sdsp and non-bossy. Fix some $s\in S(|N'|)$.
If $P,P'\in \mathcal P^{N'}$ are such that $P|_{A(s)} = P'|_{A(s)}$, then $f(s,P) = f(s,P')$.
\end{lem}
\begin{proof}
    Let $\pi = f(s, P)$ and suppose that for some  $i\in N'$, $f(s,P) \neq f(s, (P_i', P_{-i}))$. By non-bossiness, this means $\pi_i' = f_i(s, (P_i', P_{-i}))\neq \pi_i$, contradicting \sdspness since $P_i$ and $P_i'$ coincide over all feasible lotteries for $i$ given $s$. Repeating this argument $|N'|-1$ more times, we conclude that $f(s,P) = f(s,P')$.
\end{proof}

The next lemma establishes that, independently
of the profile, the initial mechanism must be  a monarchy or a
diarchy in a way that is consistent with the conditions defining an
MD-sequencing rule.

\begin{lem}
  \label{lem: step to recurse}
For any nonempty $N'\subseteq N$,  if a generalized  \rul $f:S(|N'|)\times
\mathcal P^{N'}\to \Pi^{N'}$ is  \strongefficient, \sdsp, neutral, non-bossy, and boundedly
  invariant, then for each $s\in S(|N'|)$ such that $|\underline A(s)| \leq 2$ :
  \begin{enumerate}
  \item If $\integ(s) = 1$,  there is either $i\in N'$ such that for
    each $P\in \mathcal P^{N'}$, $f_i(s,P) = m^i_i(s,P)$ or  a pair
    $i,j\in
    N'$ and $\alpha\in(0,1)$ such that for each $P\in \mathcal P^{N'}$ and
    each $k\in \{i,j\}, f_k(s,P) = d_k^{(i,j,\alpha)}(s,P)$.
  \item If $\integ(s) = 0$, there is $i\in N'$ such that for each $P\in
    \mathcal P^{N'}$, $f_i(s,P) = m^i_i(s,P)$.
  \end{enumerate}
\end{lem}
\begin{proof} Consider a preference profile $P\in\mathcal P^{N'}$. Let
  $\overline P\in\mathcal P^{N'}$ be a profile of identical
  preferences for agents in $N'$. By \cref{lem:ignore-exhausted}, we
  assume that every agent in $N'$   ranks the objects of $A(s)$ above
  those of $A\setminus A(s)$.

\paragraph{Integral \boldmath$s$:}   
  Let $\overline \pi = f(s,\overline P)$. By \cref{lem: int supply means there exists mon or di}, there
  is either $i\in N'$ such that $\overline \pi_i = m_i^i(s,\overline
  P)$ or a
  pair $i,j\in N'$ and $\alpha\in (0,1)$ such that for each $k\in
  \{i,j\}, \overline \pi_k = d^{(i,j,\alpha)}_k(s,\overline P)$. We
  consider each of
  these possibilities separately:\\

\noindent {\bf \boldmath 1. There is $i\in N'$ such that  $\overline \pi_i =
    m_i^i(s,\overline P)$.} By neutrality, it is without
    loss of
    generality to suppose that for each $j\in N'$, $\overline P_j =
    P_i$. Since $\integ(s) = 1$, there is $a\in A$ such that
    $\overline \pi_{ia} = 1$. Let $\tilde P\in \mathcal P^{N'}$ be
    such that for each $j\neq i$ and each $b\in A\setminus \{a\}$, $a\mathrel {\tilde P_j} b$ and
    $\tilde     P_j|_{A\setminus \{a\}} =     P_j|_{A\setminus \{a\}} $.
    Since $f_{ia}(s,(P_i, \overline P_{N'\setminus \{i\}})) = 1$, by
  bounded invariance, $f_{ia}(s,(P_i, \tilde P_{N'\setminus \{i\}})) =
  1$. Since for each $j\in N'\setminus\{i\}$, $f_{ja}(s,(P_i, \tilde
  P_{N'\setminus \{i\}})) = 0$, $P_j$ is a support monotonic
  transformation of $\tilde P_j$ at $f_{j}(s,(P_i, \tilde
  P_{N'\setminus \{i\}}))$. Thus, by invariance to support monotonic
  transformations (\cref{lem: monotonicity}), $f(s,P) = f(s,(P_i,
  \tilde P_{N'\setminus
    \{i\}}))$, so $f_i(s,P) = m_i^i(s,P)$.\medskip

\noindent{\bf \boldmath 2. There are $i,j\in N'$ and $\alpha\in (0,1)$ such that for
  each $k\in \{i,j\}, \overline \pi_k = d_k^{(i,j,\alpha)}(s,\overline
  P)$.} For each $k\in \{i,j\}$, let $a_k$ and $b_k$ be the best and
  second best objects  in $A(s)$ respectively
  according to $P_k$. We distinguish three cases:\smallskip

 \textbf{(a) No Conflict \boldmath ($a_i \neq a_j$):} By neutrality,    it is
    without loss of generality to suppose that for each $k\in N'$,
    $\overline P_k$ ranks $a_i$ first and $a_j$ second. Let $\hat P
    \in \mathcal P^{N'}$ be such that for each $k\in N'$,  $a_j\mathrel {\hat P_k} a_i$,
    for each $b\in A\setminus \{a_i,a_j\}$, $a_i \mathrel {\hat
      P_k}  b$, and $\hat P_k|_{A\setminus\{a_i,a_j\}} = \overline
    P_k|_{A\setminus\{a_i,a_j\}} $.

    Let $\hat\pi = f(s, (\hat P_j, \overline P_{N'\setminus
      \{j\}}))$. We show below that $\hat\pi_{ja_j} = 1$ and
    $\hat\pi_{ia_i} = 1$.  By \sdspness, $\hat\pi_{ja_i} + \hat\pi_{ja_j} =
    1$. If $\hat \pi_{ja_i} > 0$, then $\hat
    \pi_{ja_j} < 1$. By feasibility, there is some $k\in N'\setminus
    \{j\}$ such that $\hat \pi_{ka_j} > 0$. This violates
    \strongefficiency. Thus, $\hat \pi_{ja_j} = 1$.
    By     invariance to support monotonic transformations,
    $f(s, (\overline P_i, \overline P_j, \hat P_{N'\setminus
      \{i,j\}})) = \overline \pi$.
    Now, consider $\doublehat\pi =f(s, (\overline P_i,
    \hat P_{N'\setminus \{i\}}))$. \Sdspness implies    $\doublehat  \pi_{ja_i} + \doublehat \pi_{ja_j} = 1$.
    By neutrality, $f_{ia_j}(s, \hat P)+f_{ia_i}(s,
    \hat P) = 1$ since $\hat P$ swaps the names of $a_i$ and
    $a_j$ compared to $\overline P$. So, by \sdspness,
    $\doublehat  \pi_{ia_i} + \doublehat \pi_{ia_j} = 1$. Then, since
    $a_i \mathrel {\overline P_i} a_j$ and $a_j \mathrel {\hat P_j}
    a_i$, \strongefficiency ensures that $\doublehat \pi_{ia_i} = 1$
    and $\doublehat \pi_{ja_j} = 1$. Now, by invariance to support
    monotonic transformations, $\hat \pi = \doublehat \pi$, so $\hat
    \pi_{ia_i} = 1$. Thus, for each $l\in\{i,j\}, f_l(s, (\hat P_j,
    \overline P_{N'\setminus \{j\}})) = d_l^{(i,j,\alpha)}(s,P)$.

      Next, consider  $\tilde P \in \mathcal P^{N'}$ such that
    $\tilde P_j = \hat P_j$, and for
    each $k\in N'\setminus\{j\}$, $\tilde P_k$ ranks $a_i$ first,
    $a_j$ second, and
    agrees with $P_k$ on the remaining objects in $A(s)$. By bounded
    invariance, for   each $l\in\{i,j\}, f_l(s, \tilde P) =
    d_l^{(i,j,\alpha)}(s,P)$. In turn,  by invariance to support
      monotonic transformations, we conclude that
     for each $l\in\{i,j\}, f_l(s,P) = d_l^{(i,j,\alpha)}(s,P)$.\smallskip

\textbf{(b) Full Conflict \boldmath ($a_i = a_j = a$ and $b_i =
      b_j=b$):}
    By neutrality, it is without loss of generality to
    suppose that for
    each $k\in N'$, $\overline P_k = P_i$. Since $\integ(s) = 1$,
    $\overline \pi_{ia} + \overline \pi_{ib} = 1$ and     $\overline
    \pi_{ja} + \overline \pi_{jb} = 1$.
    We proceed as above  by considering $\tilde P\in \mathcal
 P^{N'}$ where, for each $k\in N'$, $\tilde P_k$  ranks  $a$ first,
    $b$ second, and
    agrees with $P_k$ on the remaining objects in $A(s)$. By bounded invariance, for each $l\in \{i,j\}$, $f_l(s,\tilde P)=f_l(s,\overline P)$.
    We conclude,
    by invariance to support monotonic transformations (from $\tilde
    P$ to $P$),
    that for each $l\in\{i,j\}, f_l(s,P) = d_l^{(i,j,\alpha)}(s,P)$.\smallskip

\textbf{(c) Partial Conflict \boldmath ($a_i = a_j$ and $b_i \neq b_j$):}
    Write $a$ for the common best object in $A(s)$ and $b$ and $c$ for the second
    best at $P_i$ and at $P_j$ respectively. By neutrality, suppose $\overline P_k = P_i$ for
    each $k\in N'$, so that $\overline\pi_{ia}=\alpha$ and
    $\overline\pi_{ja}=1-\alpha$. Let $\widehat P$ raise $a$ to the top of every
    preference in $P$. Note that $\widehat P_i=P_i$ and $\widehat P_j=P_j$. As $a$
    is first at both $\overline P$ and $\widehat P$, bounded invariance at $a$,
    applied one agent at a time, gives $f_{ia}(s,\widehat P)=\alpha$ and
    $f_{ja}(s,\widehat P)=1-\alpha$. Thus, no agent other than  $i$ or
    $j$ receives positive probability of $a$ at
    $\widehat  P$ and for each $o$ with
    $f_{ko}(s,\widehat P) > 0$ we have  $U(P_k,o)\subseteq U(\widehat
    P_k,o)$, where for each $P_0\in\mathcal P$ and $o\in A$, $U(P_0,o) = \{b\in A: b\mathrel {P_0} o\}$. So $f(s,P)=f(s,\widehat P)$.
    If  $f_{ib}(s,\widehat P)<1-\alpha$, then $i$  receives
  positive probability of some $o$ such that
    $b\mathrel{P_i}o$ while $b$ is  fully available. Some $k\neq i$
    receives positive probability of  $b$. If $o\mathrel{\widehat P_k}b$, then $i$ and $k$ can
    trade,
    contradicting \strongefficiency. Otherwise, since all agents rank
    $a$ first, $k$ ranks $a$, $b$, and $o$ in the same order as $i$,
    contradicting
    \cref{lem: no gaps with common prefs}. So
    $f_{ib}(s,P)=1-\alpha$. Analogously,  $f_{jc}(s,P)=\alpha$.
    We conclude then that  for each $l\in\{i,j\}, f_l(s,P) = d_l^{(i,j,\alpha)}(s,P)$.\smallskip

\paragraph{Non-integral \boldmath$s$:}    Since $\sum_{o\in A}s_o=|N'|$ and
$|\underline A(s)|\leq 2$, we have $\underline A(s)=\{a,b\}$ with $s_b=1-s_a$. The
remaining probabilities distinguish $a$ and $b$ from the other objects, so the
common ranking can no longer be chosen freely by appeal to
neutrality. If $|N'|=1$, the sole agent, say $i$, receives
$s=\tau^i(s,P)$. Otherwise fix
$c\in A(s)\setminus\{a,b\}$, so that $s_c=1$.

Let $\overline P$ be the identical profile ranking $c$, $a$, $b$ on top and
$\overline\pi=f(s,\overline P)$. By \cref{lem: no gaps with common prefs} and
feasibility, there is $i\in N'$ with $\overline\pi_i=m^i_i(s,\overline P)$, so
$\overline\pi_{ic}=1$. By \cref{lem: monotonicity}, this
persists when the other agents lower $c$ without altering their ranking of the
other objects $(\dag)$. We show that $i$ is the monarch at every identical
profile.

Let $\overline P'$ rank $a$, $c$, $b$ on top and
$\widehat\pi=f(s,(\overline P_i',\overline P_{N'\setminus\{i\}}))$. By
\sdspness, $\widehat\pi_{ib}=0$ and
$\widehat\pi_{ia}+\widehat\pi_{ic}=1$. Were $\widehat\pi_{ia}<s_a$, then
$\widehat\pi_{ic}>0$ and, by feasibility, $\widehat\pi_{ja}>0$ for some $j\neq
i$, so $i$ and $j$ could trade $a$ against $c$, contradicting
\strongefficiency. Thus $\widehat\pi_{ia}=s_a$ and $\widehat\pi_{ic}=1-s_a$,
then, by \cref{lem: no gaps with common prefs} and feasibility,
$\widehat\pi_{kc}=s_a$ and $\widehat\pi_{kb}=1-s_a$ for some $k\in N'$. Let
$\widehat\pi'=f(s,(\overline P'_{ik},\overline P_{N'\setminus\{i,k\}}))$. By
\sdspness, $\widehat\pi'_{kb}=1-s_a$,
$\widehat\pi'_{ka}+\widehat\pi'_{kc}=s_a$, and
$\widehat\pi'_{ia}+\widehat\pi'_{ic}=1$, since $i$ would otherwise gain by
reporting $\overline P_i$ and receiving all of $c$ by $(\dag)$. If
$\widehat\pi'_{ka}>0$, then $i$ and $k$ rank $a$, $c$, $b$ alike while
$\widehat\pi'_{ka},\widehat\pi'_{kb},\widehat\pi'_{ic}>0$, contradicting
\cref{lem: no gaps with common prefs}. So $\widehat\pi'_{ia}=s_a$ and
$\widehat\pi'_{ic}=1-s_a$. Every other agent is assigned none of $a$, $b$, and
$c$, so $\overline P'_l$ is a support monotonic transformation of
$\overline P_l$ at $\widehat\pi'_l$, giving $f(s,\overline P')=\widehat\pi'$ and
$f_i(s,\overline P')=m^i_i(s,\overline P')$.

Let $\overline P''$ rank $a$, $b$, $c$ on top and
$\tilde\pi=f(s,(\overline P_i'',\overline P'_{N'\setminus\{i\}}))$. By
bounded invariance and \strongefficiency, $\tilde\pi_{ia}=s_a$ and
$\tilde\pi_{ib}=1-s_a$. Take $j\neq i$ and let
$\tilde\pi'=f(s,(\overline P''_{ij},\overline P'_{N'\setminus\{i,j\}}))$. By
$(\dag)$ and \sdspness,
$\tilde\pi'_{ia}+\tilde\pi'_{ib}+\tilde\pi'_{ic}=1$, and by bounded
invariance $\tilde\pi'_{ia}=s_a$. Were $\tilde\pi'_{ic}>0$, then
$\tilde\pi'_{ib}<s_b$ and $\tilde\pi'_{kb}>0$ for some $k\neq i$: for
$k=j$, agents $i$ and $j$ rank $a$, $b$, $c$ alike while
$\tilde\pi'_{ia},\tilde\pi'_{ic},\tilde\pi'_{jb}>0$, contradicting
\cref{lem: no gaps with common prefs}. For $k\neq j$, agent $i$ prefers $b$ to $c$
while $k$ prefers $c$ to $b$, contradicting \strongefficiency. So
$\tilde\pi'_{ib}=1-s_a$. Switching the remaining agents from $\overline P'$
to $\overline P''$ one at a time by the same argument, we obtain
$f_i(s,\overline P'')=m^i_i(s,\overline P'')$.

At the identical profile ranking $c$, $b$, $a$ on top every agent again ranks $c$
first, so bounded invariance at $c$ leaves $i$ with all of it. Symmetric
arguments  cover the orderings $b$, $c$, $a$ and  then $b$,
$a$, $c$, so $i$ is the monarch whenever $a$, $b$, and $c$ occupy the top three
positions of an identical profile, in any order. By neutrality, $c$
may be any fully available object, and bounded invariance then makes $i$ the monarch at every
identical profile.

Finally, fix $P$, let $\overline P$ be the identical profile with
$\overline P_k=P_i$ for each $k$, let $\overline\pi=f(s,\overline P)$, and let
$\overline A=\{o\in\supp(\overline\pi_i):\overline\pi_{io}=s_o\}$. Let
$\tilde P$ be such that, for each $k\neq i$, $\tilde P_k$ ranks every
object of $\overline A$ above every object of $A(s)\setminus\overline A$,
agreeing with $P_i$ on $\overline A$  and  with $P_k$ on
$A(s)\setminus\overline A$. Let $\tilde\pi=f(s,(P_i,\tilde P_{N'\setminus\{i\}}))$. Since
$f_i(s,(P_i,\overline P_{N'\setminus\{i\}}))=\overline\pi_i$, bounded invariance
gives $\tilde\pi_{io}=\overline\pi_{io}$ for each $o\in\overline A$. If
$\supp(\overline\pi_i)=\overline A$, then $\tilde\pi_{jo}=0$ for each
$o\in\overline A$ and each $j\neq i$, so $P_j$ is a support monotonic
transformation of $\tilde P_j$ at $\tilde\pi_j$. Otherwise, the best and
second best objects of $A(s)$ at $P_i$ are some $o\in\overline A$ and some
$o'\in A(s)\setminus\overline A$ with $s_{o'}=1$. Were
$\tilde\pi_{io'}\neq\overline\pi_{io'}$, feasibility would give
$\tilde\pi_{io''}>0$ for some $o''$ with $o'\mathrel{P_i}o''$, and
$\tilde\pi_{jo'}>0$ for some $j$, contradicting
\cref{lem: no gaps with common prefs} if $o'\mathrel{\tilde P_j}o''$ and
\strongefficiency otherwise. So $\tilde\pi_i=\overline\pi_i$, and since
$\tilde P_j$ differs from $P_j$ only in the position of $o$, where
$\tilde\pi_{jo}=0$, again $P_j$ is a support monotonic transformation of
$\tilde P_j$ at $\tilde\pi_j$. In either case, by \cref{lem: monotonicity}, $f(s,P)=\tilde\pi$, so
$f_i(s,P)=m^i_i(s,P)$.
\end{proof}

The next Lemma allows us to repeatedly invoke
\cref{lem: step to recurse}.

\begin{lem}\label{lem:shape}
  Let $\sigma$ be an MD-sequencing rule and
  let $\eta = (\pi^1, \dots, \pi^K)$ be reachable under $\sigma$.
Either $\integ(r(\eta)) = 1$ or there are distinct $i,j\in N(\eta)$,
distinct $a_i,a_j\in A(r(\eta))$, and $t_i,t_j\leq K$ such that
$\underline A(r(\eta)) = \{a_i,a_j\}$, $r_{a_i}(\eta) + r_{a_j}(\eta) = 1$,
$\pi_{ia_i}^{t_i} = 1-r_{a_i}(\eta)$, and $\pi_{ja_j}^{t_j} = 1-r_{a_j}(\eta)$.
\end{lem}

\begin{proof}
  We proceed by induction on $\eta$. The base case of $\eta = ()$
  holds since $\integ(r(())) =   1$. As an induction hypothesis,
  suppose the conclusion of the Lemma holds for $\eta =  (\pi^1,
  \dots, \pi^K)$. Let $\eta' = (\eta, \pi^{K+1})$ where $\pi^{K+1} =
  \sigma(\eta)(r(\eta), P)$ for some $P$ at which $\eta'$ is
  reachable. For the remainder of the proof, write $s = r(\eta)$ and
  $s' = r(\eta')$. If $\integ(s') = 1$, we are done. So, suppose
  $\integ(s') = 0$. If $\integ(s) = 1$ then since $\integ(s') = 0$,
  there is a distinct pair $i, j\notin N(\eta)$ and $\alpha\in (0,1)$
  such that $\sigma(\eta) = d^{(i,j,\alpha)}$ with $i$ and $j$ both
  ranking the same object $o$ as their most preferred among $A(s)$ at
  $P$, but distinct objects, $a$ and $b$ respectively, as their second
  most preferred in $A(s)$. Then, $\underline A(s') =\{a, b\}$, $s'_a+s'_b=1$,
  $\pi^{K+1}_{ia} = 1-s'_a$, and $\pi^{K+1}_{jb} = 1-s'_b$, so we are
  done. The remaining case is where $\integ(s) = \integ(s') = 0$. By
  definition of an MD-sequencing rule, $\sigma(\eta) = m^i$ for some
  $i\notin N(\eta)$ and $\pi_i^{K+1} = \tau^i(s, P)$ while for all
  $j\neq i, \pi_j^{K+1} = \vec 0$. The induction hypothesis ensures that
  $\underline A(s) = \{a, b\}$, and there are $t_j, t_k \leq K$ and  $j,
  k\in N(\eta)$ such that $j\neq k$ where $s_a + s_b = 1$, $\pi^{t_j}_{ja}
  = 1-s_a$, and $\pi^{t_k}_{kb} = 1-s_b$. If $i$'s most preferred object
  in $A(s)$ is $c\in A(s)\setminus \{a,b\}$ then $s_c = 1$ so
  $\pi_{ic}^{K+1} = 1$ and the conclusion follows from the induction
  hypothesis. Otherwise,  $i$'s
  most preferred object in $A(s)$ is one of $a$ or $b$. Without loss
  of generality, suppose it is $a$. 
  If $b$ were $i$'s second most
  preferred object in $A(s)$, then $\pi_{ia}^{K+1}=s_a$ and
  $\pi_{ib}^{K+1}=s_b$, so $s'$ would be integral. 
  Thus, $i$'s
  second most preferred object in $A(s)$ is some $c\in A(s)
  \setminus\{a,b\}$. Then $s'_a = 0, s'_b = s_b,$ and $s'_c =
  1-\pi_{ic}^{K+1}=s_a$. Therefore, we have $\underline A(s') = \{b, c\}$,
  $s_b' + s_c'  = s_b + s_a = 1$, $\pi^{K+1}_{ic} = 1-s_c'$, and
  $\pi_{kb}^{t_k} = 1- s_b'$.
\end{proof}

In wrapping up the proof of \cref{thm: characterization} we will show that the
MD-sequencing rule that we back out from a mechanism satisfying our
axioms is  identity-overlap invariant. We will, in fact, show that it
satisfies  an equivalent property.
Given a permutation $\rho:A\to A$ and  $\eta\in \mathcal H$, apply
$\rho$ to $\eta$ coordinatewise: $\rho(\eta)
=(\rho(\pi^t))_{t=1}^K$. Say that  an MD-sequencing
rule $\sigma$ is \textbf{neutral} if $\sigma(\rho(\eta))=\sigma(\eta)$ for each permutation
$\rho:A\to A$ and each reachable non-terminal history $\eta$.

\begin{lem}\label{lem:measurable-iff-neutral}
An MD-sequencing rule is identity-overlap invariant if and only if it is
neutral.\end{lem}

\begin{proof}
By definition of permutation of histories, for each $\rho:A\to A$ and
history $\eta$,  $\eta$ and $\rho(\eta)$ are identity-overlap
equivalent.

We first show that if $\sigma$ is  identity-overlap invariant, then it
is neutral. Let $\eta = (\pi^1, \dots, \pi^K)$ be a non-terminal history reachable
under $\sigma$ at $P\in \mathcal P^N$ and let $\rho:A\to A$ be a
bijection.
Define $\eta^k = (\pi^1, \dots, \pi^k)$.
We show, by induction on $k$, that $\sigma(\rho(\eta^k)) =
\sigma(\eta^k)$ and that $\rho(\eta^k)$ is reachable at $\rho(P)$.
The base case of the empty history, $k = 0$, is trivial. As an
induction hypothesis,
assume that it holds for some $k < K$.

Consider $\eta^{k+1} = (\pi^1, \dots, \pi^{k+1})$, which is reachable
under $\sigma$ at $P$. By the induction hypothesis, $\rho(\eta^k)$
is reachable at $\rho(P)$. By definition of an MD-sequencing rule, and
since $\eta^k$ and $\rho(\eta^k)$ are identity-overlap equivalent, we
have $\sigma(\eta^k) = \sigma(\rho(\eta^k)) = f\in \mathcal M \cup
\mathcal D$. Moreover, by definition of $\eta^{k+1}$ being reachable
under $\sigma$ at $P$, $\pi^{k+1} = f(r(\eta^k), P).$ Under $\rho(P)$,
the $(k+1)^{\text{th}}$ recursive call produces
\[
\tilde \pi^{k+1} = f(r(\rho(\eta^k)), \rho(P)) = f(\rho(r(\eta^k)),
\rho(P)) = \rho(f(r(\eta^k), P)) = \rho(\pi^{k+1}),
\]
where the second equality follows from the definition of permutation and the third follows from $f\in \mathcal{M}\cup \mathcal{D}$. Thus, the subsequent  history under $\rho(P)$ is $(\rho(\eta^k),
\rho(\pi^{k+1})) = \rho(\eta^{k+1})$. Moreover, since $\eta^{k+1}$ and
$\rho(\eta^{k+1})$ are identity-overlap equivalent, we have
$\sigma(\rho(\eta^{k+1})) = \sigma(\eta^{k+1})$, completing the first
direction of the proof.

Now we show that if $\sigma$ is neutral, then it is identity-overlap
invariant. It suffices to show that for any pair of
identity-overlap equivalent histories, $\eta = (\pi^1, \dots, \pi^K)$
and $\tilde\eta= (\tilde\pi^1, \dots, \tilde\pi^K)$ reachable under $\sigma$, there is a
bijection  $\rho:A\to A$ such that for each $k \leq K$,
$\rho(\pi^k) = \tilde \pi^k$.

Define, for each $k \leq K$,
$\eta^{k} = (\pi^1, \dots, \pi^k)$ and $\tilde\eta^{k} = (\tilde\pi^1,
\dots, \tilde\pi^k)$. We show, by induction on $k$, that there is a
bijection $\rho:A\to A$ such that  $\rho(\eta^k) = \tilde \eta^k$.
This is trivially true for the base case of the empty history, $k=0$.
As an induction hypothesis, suppose such $\rho$ exists such that,
$\rho(\eta^{k-1}) = \tilde \eta^{k-1}$.
By neutrality of $\sigma$, this means $\sigma(\tilde \eta^{k-1}) =
\sigma(\rho(\eta^{k-1})) = \sigma(\eta^{k-1})  = f \in \mathcal M\cup
\mathcal D$.
 Let $r = r(\eta^{k-1})$ and $\tilde r = r(\tilde\eta^{k-1})$. By the
induction hypothesis, $\rho(r) = \tilde r$.

Letting $B =  \{a\in A:r_a < 1\}$, any modification to $\rho$ that
preserves it over $B$ preserves $\rho(\eta^{k-1}) = \tilde
\eta^{k-1}$. Letting $C = \{a\in A: r_a = 1, \pi_{ia}^k > 0\text{ for some
}i\in N(\pi^k) \}$, we only need to
show that we can select the image of $C$ under $\rho$ to complete the
proof.  Specifically, we need to adjust $\rho$ so that it maps $C$ to
$\tilde C =  \{a\in A: \tilde r_a = 1, \tilde\pi_{ia}^k > 0\text{ for some
}i\in N(\tilde \pi^k) \}$ in a
way that $\rho(\pi^k) = \tilde \pi^k$.

If $\integ(r)=1$ and $f=m^i$, then $\pi^k_{io} = 1$ and
$\tilde\pi^k_{i\tilde o} = 1$ for some $o\in C$ and $\tilde o\in
\tilde C$. So we set $\rho(o)=\tilde o$.
If $\integ(r)=1$ and $f=d^{(i,j,\alpha)}$,
the three case distinctions in the proof of \cref{lem: step to recurse}
apply the same way for both $\pi^k$ and $\tilde\pi^k$ since $\eta$ and
$\tilde \eta$ are
identity-overlap equivalent. In the case
of no conflict,
 $\pi^k_{io} = 1$,  $\pi^k_{jp} = 1$, $\tilde\pi^k_{i\tilde o} = 1$, and
$\tilde\pi^k_{j\tilde p} = 1$ for distinct pairs $o,p\in C$ and
$\tilde o, \tilde p\in \tilde C$. So we set $\rho(o) = \tilde o$ and
$\rho(p) = \tilde p$. In the case of full conflict, there are two
pairs of distinct objects $o, p\in C$ and $\tilde o, \tilde p\in\tilde C$
such that
$\pi^k_{io} = \pi^k_{jp} = \alpha$, $\pi^k_{ip} = \pi^k_{jo} = 1-
\alpha$, $\tilde \pi^k_{i\tilde o} = \tilde\pi^k_{j\tilde p} =
\alpha$, and  $\tilde\pi^k_{i\tilde p} = \tilde\pi^k_{j\tilde o} = 1-
\alpha$. So we can set  $\rho(o) = \tilde o$ and
$\rho(p) = \tilde p$. Finally, in the case of partial conflict, there
are two  distinct triples of  objects $o, p, q\in C$ and $\tilde o,
\tilde p, \tilde q\in \tilde C$
such that
$\pi^k_{io} = \pi^k_{jp} = \alpha$, $\pi^k_{iq} = \pi^k_{jo} = 1-
\alpha$, $\tilde \pi^k_{i\tilde o} = \tilde\pi^k_{j\tilde p} =
\alpha$, and  $\tilde\pi^k_{i\tilde q} = \tilde\pi^k_{j\tilde o} = 1-
\alpha$. So we can set  $\rho(o) = \tilde o$,
$\rho(p) = \tilde p$, and $\rho(q) = \tilde q$.

If $\integ(r)=0$, then $f=m^i$ and, by \cref{lem:shape},
$\underline A(r)=\{a,b\}$, where $a$ and $b$ were previously allocated
to distinct agents $j$ and $j'$, and by the induction hypothesis $\underline A(\tilde r) =
\{\rho(a), \rho(b)\}$, with $\rho(a)$ allocated  to $j$ and $\rho(b)$
to  $j'$. Since $i$ overlaps
with $j$ only through $a$ and with $j'$ only through $b$,
identity-overlap equivalence gives $\pi^k_{ia}>0$ if and only if
$\tilde\pi^k_{i\rho(a)}>0$, and similarly for $b$. Every fractional
object received by the monarch is exhausted. Since $\tilde r =
\rho(r)$,  $\tilde\pi^k_{i\rho(a)} = \pi^k_{ia}$
and $\tilde\pi^k_{i\rho(b)} = \pi^k_{ib}$. There is at most one $o\in
C$ with $\pi^k_{io}>0$, and then at most one $\tilde o\in\tilde C$ with
$\tilde\pi^k_{i\tilde o}>0$, where $\pi^k_{io} = \tilde\pi^k_{i\tilde
o}$ since $i$ is allocated a total of $1$ in each. So we set $\rho(o) =
\tilde o$.
In each case, the specified images extend to a bijection between
the objects with remaining probability one. Together with the
unchanged images on $B$, this completes the induction.
\end{proof}

We next establish the independence from earlier reports needed
for the recursive construction.

\begin{lem}\label{lem:history-independence}
Let $\varphi$ be a \sdsp, non-bossy, and neutral \rul,
and let $\sigma$ be an MD-sequencing rule. Let
$\eta = (\pi^1, \dots, \pi^K)$ be reachable under $\sigma$, let
$\eta^t = (\pi^1, \dots, \pi^t)$, with $\eta^0 = (\,)$, and let
$f^t = \sigma(\eta^{t-1})$. Suppose that for
each $t = 1, \dots, K$, each profile $\overline P$ at which $\eta^{t-1}$ is
reachable, and each $i \in N(\pi^t)$,
\[
  \varphi_i(\overline P)
  = f^t_i(r(\eta^{t-1}), \overline P).
\]
If $\eta$ is reachable at $P$ and $P'$, and
$P_{N \setminus N(\eta)} = P'_{N \setminus N(\eta)}$, then
$\varphi(P) = \varphi(P')$.
\end{lem}

\begin{proof}
For each $t = 0, \dots, K$, let
\[
  \widetilde P^t
  =\bigl(P'_{N(\eta^t)}, P_{N \setminus N(\eta^t)}\bigr).
\]
Each $\widetilde P^t$ reaches $\eta$. To see this, proceed through
the calls in order. At each call, all agents served in that call
report either as at $P$ or as at $P'$. Given the same preceding
history, either choice therefore produces the same partial
allocation.
Fix $t\in \{1, \dots, K\}$ and let $s^t = r(\eta^{t-1})$ be the
residual supply at $\eta^{t-1}$. Since both profiles reach $\eta$,
\[
  f^t(s^t, \widetilde P^{t-1}) = f^t(s^t, \widetilde  P^t) = \pi^t.
\]

If $f^t = m^i$ for some $i\in N(\pi^t)$, then by
the hypothesis,
$\varphi_i(\widetilde  P^{t-1}) = \pi_i^t = \varphi_i(\widetilde
P^t)$, so non-bossiness yields $\varphi(\widetilde P^{t-1})
= \varphi(\widetilde P^t)$.

If  $f^t = d^{(i,j,\alpha)}$, by definition of an MD-sequencing
rule, $s^t$ is integral, so $\pi^t$ determines whether $i$ and $j$ share
a best object in $A(s^t)$. If they do not, $\pi^t$ determines each of
$i$ and $j$'s best object. Suppose  they share a best object and either
\mbox{(A)}~$\alpha\neq \frac12$, \mbox{(B)}~they have different second best objects,
or \mbox{(C)}~they rank the same best two objects in the same order at $\widetilde P^t$
as $\widetilde P^{t-1}$. In \mbox{(A)} and \mbox{(B)},
$\pi^t$ determines the shared best
object and each agent's second best object. So in each of these cases,
the best and, where relevant, the second best
objects in $A(s^t)$ agree at $\widetilde P^{t-1}$ and $\widetilde  P^t$. Changing the two
preference reports one at a time then preserves that $f^t$ yields $\pi^t$,
so the hypothesis and non-bossiness give
$\varphi(\widetilde P^{t-1}) = \varphi(\widetilde P^t)$.

It remains to consider the case where $\alpha = \frac12$ and there are
distinct $a, b \in A(s^t)$  such that
both $i$ and $j$ rank $a$ first and $b$ second in $A(s^t)$ at $\widetilde
P^{t-1}$, and  $b$ first and  $a$ second  at $\widetilde P^t$. Then
$\pi^t_{ia} = \pi^t_{ib} = \pi^t_{ja} = \pi^t_{jb} = \frac12$, so
at neither $\widetilde P^{t-1}$ nor  $\widetilde P^t$ does any other agent receive any
positive probability of $a$ or $b$. Let  $\widehat
P_{-\{i,j\}}^{t-1},\widehat
P_{-\{i,j\}}^{t} \in \mathcal P^{N\setminus \{i,j\}}$ be such that they
both coincide with $\widetilde P^{t-1}_{-\{i,j\}} = \widetilde
P^{t}_{-\{i,j\}}$ on $A\setminus \{a,b\}$ and every agent $k\in
N\setminus\{i,j\}$ ranks $a$ and $b$ in the last two positions, with
$a\mathrel{\widehat P_k^{t-1}} b$ and $b\mathrel{\widehat P_k^{t}} a$.
\cref{lem: monotonicity} says that
$\varphi(\widehat P^{t-1}_{-\{i,j\}}, \widetilde P_{\{i,j\}}^{t-1}) =
\varphi(\widetilde P^{t-1})$ and that
$\varphi(\widehat P^t_{-\{i,j\}}, \widetilde P_{\{i,j\}}^t) =
\varphi(\widetilde P^t)$. Next, let
 $\widehat P_{\{i,j\}}^{t-1},\widehat P_{\{i,j\}}^t\in \mathcal P^{\{i,j\}}$ be
such that they agree with each other on $A\setminus \{a,b\}$ with both agents ranking $a$
and $b$  in the first two positions, with  $a$ first in
$\widehat P_{\{i,j\}}^{t-1}$ and $b$ first in $\widehat
P_{\{i,j\}}^{t}$. Appealing to \cref{lem: monotonicity} again, we have
$\varphi(\widehat P^{t-1}_{-\{i,j\}}, \widehat P_{\{i,j\}}^{t-1}) =
\varphi(\widehat P^{t-1}_{-\{i,j\}}, \widetilde P_{\{i,j\}}^{t-1}) $ and
$\varphi(\widehat P^t_{-\{i,j\}}, \widehat P_{\{i,j\}}^{t}) =
\varphi(\widehat P^t_{-\{i,j\}}, \widetilde P_{\{i,j\}}^{t})
$. Thus, $\varphi(\widetilde P^{t-1}) = \varphi(\widehat P^{t-1})$ and
 $\varphi(\widetilde P^{t}) = \varphi(\widehat P^{t})$.

Observe that $\widehat P^t = \rho(\widehat P^{t-1})$ where
$\rho$ interchanges $a$ and $b$ but leaves all other objects
fixed.  So, by neutrality of $\varphi$, $\varphi(\widetilde P^{t}) =
\varphi(\widehat P^t) = \rho(\varphi(\widehat P^{t-1})) =
\rho(\varphi(\widetilde P^{t-1}))$.
Moreover, since $\pi^t_{ia} = \pi^t_{ib} = \pi^t_{ja} =
\pi^t_{jb}=\frac12$, $\rho(\varphi(\widetilde P^{t-1})) =
\varphi(\widetilde P^{t-1})$ and therefore  $\varphi(\widetilde P^{t}) =\varphi(\widetilde P^{t-1})$.

Since $\widetilde P^0 = P$ and $\widetilde P^K = P'$, repeating the
argument for $t = 1, \dots, K$ yields $\varphi(P) = \varphi(P')$.
\end{proof}

We now use
\cref{lem: step to recurse,lem:shape,lem:measurable-iff-neutral,lem:history-independence}
to complete the proof of \cref{thm: characterization}.

\begin{proof}[Proof of \cref{thm: characterization}]
Suppose that $\varphi$ is an \strongefficient, \sdsp, neutral,
non-bossy, and boundedly invariant \rul. We construct $\sigma$
recursively, maintaining the agreement condition in
\cref{lem:history-independence}.
At the empty history, \cref{lem: step to recurse}
gives a monarchy or a diarchy that agrees with $\varphi$ on the
agents it serves at every profile. This is $\sigma((\,))$.

Let $\eta = (\pi^1, \dots, \pi^K)$ be a reachable non-terminal history,
and suppose that the selections at its proper subhistories of $\eta$ satisfy
the agreement condition. Write
\[
  S = N(\eta),\qquad S' = N \setminus S, \qquad s = r(\eta),
\]
and choose a profile $P'$ at which $\eta$ is reached. The first
$K$ calls depend only on the reports of $S$, so $\eta$ is reached
at $(P'_S, P_{S'})$ for every $P_{S'} \in \mathcal P^{S'}$. By
\cref{lem:history-independence}, $\varphi(P'_S, P_{S'})$ is independent
of the choice of reports $P'_S$ producing $\eta$.
For every $P_{S'}$, the agents in $S$ receive their allocations in
$\eta$. Consequently,
\[
  \sum_{i \in {S'}} \varphi_i(P'_S, P_{S'}) = s.
\]
With $P'_S$ fixed, the generalized \rul that allocates $s$ to
members of $S'$ inherits
\sdspness, non-bossiness, and bounded invariance from
$\varphi$. It also inherits \strongefficiency: any different
allocation for ${S'}$ feasible at $s$ can be combined with the
unchanged allocation of $S$, and the agent whose assignment is
strictly worsened must belong to ${S'}$.

Let $\rho$ fix every object $a$ with $s_a < 1$, so $\rho(\eta) = \eta$. By
neutrality of each monarchy and diarchy, induction along the first
$K$ calls shows that $\eta$ is also reached at
$(\rho(P'_S), \rho(P_{S'}))$. Applying
\cref{lem:history-independence} and then neutrality of $\varphi$
gives, for each $i \in {S'}$,
\[
  \varphi_i(P'_S,\rho(P_{S'}))
  =\varphi_i(\rho(P'_S),\rho(P_{S'}))
  =\rho\bigl(\varphi_i(P'_S,P_{S'})\bigr).
\]
This establishes neutrality of the generalized \rul for $S'$.

By \cref{lem:shape}, $|\underline A(s)| \leq 2$.
\cref{lem: step to recurse}, applied at $s$ with $P'_S$ fixed,
therefore gives
$f \in \mathcal M \cup \mathcal D$ serving agents in ${S'}$ and agreeing
with $\varphi(P'_S, P_{S'})$ on their allocations for every $P_{S'}$.
If $s$ is non-integral, $f$ is a monarchy. Set $\sigma(\eta) = f$.
Since the remaining allocations are independent of the choice
of $P'_S$ producing $\eta$, this agreement holds at every profile
at which $\eta$ is reached. Thus, the agreement condition of
\cref{lem:history-independence} is
preserved by the construction.

Repeating this argument specifies $\sigma$ at every reachable
non-terminal history. Each call serves at least one previously
unserved agent, so the construction terminates after at most $n$
calls at each profile, and
\[
  \varphi=\Psi^\sigma((\,),\vec 1,\cdot).
\]
At unreachable histories, we complete $\sigma$ arbitrarily in
accordance with conditions (i) and (ii).

By \cref{lem:measurable-iff-neutral}, it remains to show that $\sigma$ is
neutral. Fix $P$ and $\rho$, write $P'=\rho(P)$, let $\eta^k=(\pi^1,\dots,\pi^k)$
be the history after the first $k$ calls at $P$, and set $\overline s=r(\eta^k)$
and $\overline N=N\setminus N(\eta^k)$. We induct on $k$ to show that
\begin{center}(i) the first $k$ calls at $P'$ produce $\rho(\eta^k)$, and (ii) $\sigma(\rho(\eta^k))=\sigma(\eta^k)$ whenever $\eta^k$ is non-terminal.
\end{center}
Both hold at $k=0$ for the two empty histories. Assume (i) and (ii) at
$k-1$. By (ii) at $k-1$, both executions apply the same $f=
\sigma(\eta^{k-1})$. Since $f\in \mathcal{M}\cup \mathcal{D}$, the
$k^{\text{th}}$ allocation at $P'$ is
$f(\rho(r(\eta^{k-1})),\rho(P))=\rho(f(r(\eta^{k-1}),P))=\rho(\pi^k)$, so by (i)
at $k-1$, the first $k$ calls produce
$\rho(\eta^k)$, leaving $\overline s'=\rho(\overline s)$.
Since $\rho(\pi_i) = \vec 0$ if and only if $\pi_i = \vec 0$, $N(\rho(\eta^k)) = N(\eta^k)$, so $N \setminus N(\rho(\eta^k)) = \overline N$.
If $\eta^k$ is terminal, nothing remains to show. Otherwise,
it remains to establish (ii) at $k$.

Let $\overline P_0$ rank the objects in decreasing order of $\overline s$ (fix one among such preferences). Let
$\overline P$ coincide with $P$ for those in $N\setminus\overline N$ and coincide with $\overline P_0$ for those in $\overline N$. The agents in $N\setminus\overline N$ have the same preferences at
$P$ and $\overline P$, so the first $k$ calls at $\overline P$ also produce
$\eta^k$, leaving $\overline s$. Setting $\overline P_0'=\rho(\overline P_0)$ and
constructing $\overline P'$ from $P'$ likewise, we get
$\overline P'=\rho(\overline P)$, and the first $k$ calls at $\overline P'$
produce $\rho(\eta^k)$, leaving $\overline s'$.

Let $a$ be the best object of $A(\overline s)$ at $\overline P_0$, so that
$\rho(a)$ is the best object of $A(\overline s')$ at $\overline P_0'$. Since all
agents of $\overline N$ share the preference $\overline P_0$, the
$(k+1)^{\text{th}}$ call exhausts $a$: writing $g=\sigma(\eta^k)$, either $g=m^i$
and $\varphi_{ia}(\overline P)=\overline s_a$, or $g=d^{(i,j,\alpha)}$ and
$\varphi_{ia}(\overline P)=\alpha\overline s_a$ and
$\varphi_{ja}(\overline P)=(1-\alpha)\overline s_a$. In either case no other
agent of $\overline N$ is assigned any of $a$. The same holds at $\overline P'$
for $g'=\sigma(\rho(\eta^k))$ and $\rho(a)$. By neutrality of $\varphi$ and $\overline P'=\rho(\overline P)$, we have
$\varphi_l(\overline P')=\rho(\varphi_l(\overline P))$, that is,
$\varphi_{l\rho(a)}(\overline P')=\varphi_{la}(\overline P)$, for each
$l\in\overline N$. So $g$ and $g'$ assign positive probability of $a$ and of
$\rho(a)$ to the same agents in the same amounts, which identifies them:
$g'=m^i$ if $g=m^i$, and $g'=d^{(i,j,\alpha)}$ if $g=d^{(i,j,\alpha)}$, up to the
convention $d^{(i,j,\alpha)}=d^{(j,i,1-\alpha)}$. Thus
$\sigma(\rho(\eta^k))=\sigma(\eta^k)$.

We now show that a hierarchy of monarchies and diarchies with respect to {an identity-overlap invariant MD-sequencing rule} $\sigma$, $MD^{\sigma}$, satisfies the properties listed in the statement of \cref{thm: characterization}.\medskip

\noindent \textbf{\StrongEfficiency:}  Given an input  $P\in
  \mathcal P^N$, let $\eta = (\pi^t)_{t=1}^K\in \mathcal H^T$ be the history accumulated
  along the execution path of a hierarchy of monarchies and diarchies,
  $MD^{\sigma}$. Let $\pi
  = MD^{\sigma}(P)$.   For each $t\in \{1, \dots,K\}$ and each $i\in N(\pi^t)$, $\pi_i$
  beats  every lottery in
  \[
    \{
    \overline\pi_i \in \Delta(A): \text{ for each }a\in A,
    \overline \pi_{ia} \leq  1 - \displaystyle\sum_{\substack{k\in N(\pi^{s})\setminus\{i\}\\s\le t}} \pi_{ka}
    \}
  \]
  according to $\risd$.

Fix an allocation $\pi'\neq\pi$, and let $t$ be the first call
such that $\pi'_i\neq\pi_i$ for some $i\in N(\pi^t)$.
The allocations of agents served before this call are unchanged.
Thus, the remaining agents' allocations at $\pi'$ are feasible
at
\[
  r=\vec 1-\sum_{\substack{k\in N(\pi^s)\\s<t}}\pi_k.
\]
If $N(\pi^t)=\{i\}$, then $\pi'_i$ belongs to the displayed option
set and differs from $\pi_i$, so $\pi_i\mathrel{P_i^{sd}}\pi'_i$.

If there is $j\in N(\pi^t)\setminus\{i\}$, then, by definition of
$\sigma$, $r$ is integral. Either $P_i$ and $P_j$ have different
best objects in $A(r)$, or they have the same best object $a_1$.
By definition of a diarchy, in the former case $i$ and $j$ receive
probability one of their respective best objects. Hence, any
agent in this call whose assignment differs at $\pi'$ is
strictly worse off according to their stochastic-dominance
preference.

In the latter case, any allocation feasible at $r$ that assigns
more of $a_1$ to $i$ necessarily assigns less of it to $j$.
Also, by definition of a diarchy, if $a_2$ is the second best
object in $A(r)$ according to $P_i$, then
$\pi_{ia_2}=1-\pi_{ia_1}$, and the analogous statement holds for
$j$. Thus, if either agent receives more of $a_1$ at $\pi'$,
the other is strictly worse off according to their
stochastic-dominance preference. Otherwise, neither receives
more of $a_1$, and any agent in this call whose assignment differs
receives either less of $a_1$ or, with the same share of $a_1$,
less of their second best object. That agent is strictly worse
off. Therefore, for every $\pi'\neq\pi$, there is an agent $k$
such that $\pi_k\mathrel{P_k^{sd}}\pi'_k$.\medskip

\noindent \textbf{Neutrality:} { The definitions of  monarchies
  and diarchies make no reference to the names of the
  objects. By \cref{lem:measurable-iff-neutral}, the neutrality of the MD-sequencing rule ensures that the selection of successive monarchies and diarchies is also invariant to relabelling of objects. Thus, a hierarchy of monarchies and diarchies is  neutral.}\smallskip

\noindent\textbf{\Sdspness:}
  Let $P \in\mathcal P^N$ and $\pi = MD^{\sigma}(P)$.
Let $i\in N$. The calls preceding $i$'s assignment do not depend on $i$'s report. At a history $\eta = (\pi^{s})_{s=1}^{t-1}$, where
$\sigma(\eta) = m^i$ or $\sigma(\eta) = d^{(i,j,\alpha)}$ for some $j\in
N$ and $\alpha\in (0,1)$, $i$'s assignment is chosen to be the
 best allocation, according to $P_i^{sd}$, in
  \[
    \{
    \overline\pi_i \in \Delta(A): \text{ for each }a\in A,
    \overline \pi_{ia} \leq  1 - \displaystyle\sum_{\substack{k\in N(\pi^{s})\setminus\{i\}\\s\le t}} \pi_{ka}
    \}.
  \]
If $\sigma(\eta) = m^i$,  this option set is unresponsive to $i$'s
reported preferences.
If $\sigma(\eta) = d^{(i,j,\alpha)}$, this option set only depends on
whether $i$'s reported top object in
 $A( \vec 1 -
  \sum_{\substack{k\in   N(\pi^{s})\\s<t}}
  \pi_{k})$ coincides with that of $j$. Let $a_1$ be the top object
  in this set at $P_j$. If $a_1$ is also the top object in
  this set according to $i$'s report, $P_i$, then $\pi_{ja_1} =
  1-\alpha$. If  $a_1$ is not the top object in
  this set according to $P_i$, then $\pi_{ja_1} = 1$. Thus, if $a_1$
  is the true top object in this set for $i$, $i$ receives it with
  probability $\alpha$ if they report it as such, but receives it
  with probability zero if they do not. Since the truthful lottery
places its remaining probability on $i$'s true second best available
object, it weakly dominates the allocation induced by every report
with a different top object. Reports keeping $a_1$ first face the same option set.
If $a_1$ is not $i$'s true top object, truthful reporting gives $i$
their best available object with probability one.\medskip

  \noindent \textbf{Non-bossiness:}
  Given an input preference profile $P\in
  \mathcal P^N$, let $\eta = (\pi^{t})_{t=1}^K\in \mathcal H^T$ be the history accumulated
  along the execution path of a hierarchy of monarchies and diarchies,
  $MD^{\sigma}$, at $P$. Let $\pi
  = MD^{\sigma}(P)$.
Let $i\in N$ and $\overline P_i\in \mathcal P$ be such that $MD_i^{\sigma}(\overline P_i, P_{-i}) = MD_i^{\sigma}(P)$. Let $\overline \eta = (\overline \pi^t)_{t=1}^L
\in\mathcal H^T$ be the history accumulated along the execution path
 of the hierarchy at $(\overline P_i, P_{-i})$.

Suppose $i\in N(\pi^t)$ for some $t\in \{1,\cdots,K\}$. By definition of $MD^{\sigma}$, for each
$s < t$, $\pi^{s} = \overline \pi^{s}$, and for each $j\in N(\pi^{s}),
MD^{\sigma}_j(\overline P_i, P_{-i}) = MD^{\sigma}_j(P)$. Since
$MD^{\sigma}_i(\overline P_i, P_{-i}) = MD^{\sigma}_i(P)$,
$i$'s lottery at this call is unchanged. If the call is a diarchy,
the other agent's report is unchanged, so the definition of a
diarchy also preserves that agent's lottery. Thus,
$\pi^t=\overline\pi^t$. If this history is terminal, the two full
allocations coincide. Otherwise, $\sigma(\pi^1,\dots, \pi^t) = \sigma(\overline\pi^1, \dots,
\overline\pi^t)$. Since the remaining probability vector is the same, $\pi^{t+1} =
\overline \pi^{t+1} $ and for each $j\in N(\pi^{t+1})$, $MD^{\sigma}_j(\overline P_i,
P_{-i}) = MD^{\sigma}_j(P)$.  Repeating this argument, we
conclude that $MD^{\sigma}(\overline P_i, P_{-i}) = MD^{\sigma}(P)$.\medskip

\noindent \textbf{Bounded invariance:}
  Let $P\in \mathcal P^N$, $i\in N$, $a\in A$ and $\overline P_i\in\mathcal
  P$ such
  that  $U(P_i, a)=U(\overline P_i,a)$ and $\overline P_i|_{U(P_i, a)} = P_i|_{U(P_i, a)}$, where
  $U(P_i,a) = \{b\in A: b\mathrel P_i a\}$ and $U(\overline P_i,a) = \{b\in A: b~\overline P_i~ a\}$.
  Let $\pi=MD^{\sigma}(P)$.  Let the history along the execution path
  at   $P$ be $\eta = (\pi^{t})_{t=1}^K\in \mathcal H^T$ and at $(\overline P_i, P_{-i})$ be $\overline \eta=
 (\overline \pi^t)_{t=1}^L
\in\mathcal H^T$. If $i\in N(\pi^t)$ for some $t$, then, as argued above, for
 all $s < t$, $\pi^{s} = \overline \pi^{s}$ and for each $j\in N(\pi^{s})$, $MD_j^{\sigma}(\overline P_i, P_{-i}) = MD_j^{\sigma}(P)$.

Under the profile $P$, at the $t^{\text{th}}$ call to $\Psi^\sigma$, if the remaining
probability of $a$ is zero, then it is zero at $(\overline P_i, P_{-i})$ as well and
in both cases $a$ is allocated only to agents in $N(\pi^1, \dots,
\pi^{t-1})$. As we have shown above, for each such $j$,  $MD_j^{\sigma}(\overline P_i, P_{-i}) = MD_j^{\sigma}(P)$.

If $MD_{ia}^{\sigma}(P) = 1$, then $a$ is the best object that
has positive availability at the $t^{\text{th}}$ call  according to $P_i$. Since the
remaining probability is the same with input $(\overline P_i, P_{-i})$, and
since   $\overline P_i|_{U(P_i, a)} = P_i|_{U(P_i, a)}$, $a$ is also the best
object that has positive availability according to $\overline P_i$, so
$MD_{ia}^{\sigma}(\overline P_i, P_{-i}) = 1$.  By non-bossiness,
$MD^{\sigma}(\overline P_i, P_{-i}) = MD^{\sigma}(P)$.

If the remaining probability of $a$ at the  $t^{\text{th}}$ call at $P$ is
positive, but   $MD_{ia}^{\sigma}(P) {= 0}$, then
$\supp(MD_i^{\sigma}(P)) \subseteq U(P_i, a)$. Since
$\overline P_i|_{U(P_i, a)} = P_i|_{U(P_i, a)}$, $MD_i^{\sigma}(\overline P_i,
P_{-i}) {= MD_i^{\sigma}(P)}$ and we conclude by non-bossiness.

Finally, if the remaining probability of $a$ at the  $t^{\text{th}}$ call at $P$ is
positive, and $0<MD_{ia}^{\sigma}(P) < 1$, then $a$ is either the best object or the second best object according to $P_i$ among those with positive remaining probabilities. If $a$ is the best object at $P_i$, then by definition
of $MD^{\sigma}$, $a$ is allocated entirely to the agents in
$N(\pi^1, \dots, \pi^t)$, regardless of whether $a$ is allocated to a
monarchy or to a diarchy at the $t^{\text{th}}$ call at $P$. Since $a$
is still $i$'s best object at the $t^{\text{th}}$ call
at $\overline{P}_i$, the allocation of $a$ is the same at both $P$ and
$(\overline P_i,P_{-i})$. If $a$ is the second best object at $P_i$,
then there are two cases. If $a$'s supply is exhausted at
the $t^{\text{th}}$ call at $P$,  then the same argument as above
applies. So, suppose that $a$'s supply is not exhausted at the
$t^\text{th}$ call at $P$. Since $i$'s preference remains the same
down to $a$ at $\overline{P}_i$ as in $P_i$, the allocation
of the $t^{\text{th}}$ call coincides at $P$ and $(\overline P_i,
P_{-i})$. By definition of $MD^\sigma$, the subsequent selections
by $\sigma$ do not change, keeping the allocation of $a$ unchanged.
\end{proof}

We now prove that every hierarchy of monarchies and diarchies 
with respect to an MD-sequencing rule $\sigma$, $\varphi^{\sigma}$, 
is unambiguously group strategy-proof.

\begin{proof}[Proof of \cref{cor:gsp}]
Let $P$ and $P'$ be preference profiles and $C$ a nonempty coalition
such that $P'_{-C} = P_{-C}$.
Let $\pi = \varphi^\sigma(P)$ and $\pi' = \varphi^\sigma(P')$,
and let $\eta = (\pi^t)_{t=1}^K$ and $\eta' = (\pi'^t)_{t=1}^{K'}$ 
be the terminal histories reached under $P$ and $P'$, respectively.

Suppose $\pi \neq \pi'$.
Let $\hat t$ be the least $t$ with $\pi^{t} \neq \pi'^{t}$.
This is well-defined since final allocations differ.
The histories under the two preference profiles coincide 
for every step preceding $\hat t$.
Therefore, $\sigma$ selects the same partial allocation mechanism at
$\hat t$ under both profiles.
Let $r$ denote the common remaining probability vector and $J$ denote
the common set of agents active in the selected partial allocation mechanism.
Let $x = \pi^{\hat t}$ and $y = \pi'^{\hat t}$ be the partial allocations
determined at step $\hat t$ under $P$ and $P'$, where $x \neq y$.

Since $\hat t$ is the first step at which the two processes diverge,
there must be an agent in both the coalition $C$ and the active set $J$,
given that the selected partial mechanism depends only upon the preferences
reported by the agents in $J$.
We will show that there exists $i \in C \cap J$ such that
$x_i \pisd  y_i$, which establishes unambiguous group-strategy-proofness,
since $x_i = \varphi^\sigma_i(P)$ and $y_i = \varphi^\sigma_i(P')$.

\noindent\emph{Case: Monarchy at step $\hat t$}

Let $i$ be the monarch, necessarily a member of $C$. 
Then $x_i = \tau^i(r, P)$, which by construction stochastically dominates
every other lottery feasible at $r$, including $y_i$.
Since $x \neq y$ and $J = \{i\}$, $x_i \neq y_i$.
We conclude that $x_i \pisd y_i$.

\noindent\emph{Case: Diarchy at step $\hat t$}

A necessary condition for an MD-sequencing rule to select a diarchy
is integrality of remaining probability vector.
In the proof of \cref{thm: characterization}, it is established that
diarchies are unambiguously efficient amongst the diarchs.
Hence $x$ is unambiguously efficient amongst partial allocations
to the set $J$ feasible at $r$, which includes $y$.
Since $x \neq y$, this implies there exists $i \in J$ such that
$x_i \pisd y_i$.
This settles the subcase where both diarchs are in coalition $C$.

Now suppose only one of the diarchs, denoted $i$, is in $C$.
Denote the other diarch by $j$.
Recall that a diarchy is a weighted random serial dictatorship amongst
the diarchs and so is strategy-proof and non-bossy.
Given that $P_j = P'_j$, \sdspness of diarchies, 
together with the invariance of a diarchy's outcome to the 
preferences of the non-diarchs, implies $x_i \risd y_i$.
If $x_i = y_i$, non-bossiness yields $x = y$, a contradiction.
So we conclude that $x_i \pisd y_i$, completing the case.
\end{proof}

\section{Independence of all but bounded invariance}
First, a \rul selecting the random serial dictatorship allocation for
each profile satisfies all the properties but \strongefficiency. Second,
suppose that each object is endowed with a common strict priority ordering over
agents. The (deterministic) Boston \rul (also known as the ``immediate acceptance
\rul'') under the profile of priorities satisfies all the properties but
\sdspness. Third, consider a \rul that (i)~selects the serial
dictatorship allocation with respect to 1-2-3-4-$\cdots$-$n$ if
agent~1's favourite object is~$a$ and (ii)~selects the serial
dictatorship allocation with respect to 1-3-2-4-$\cdots$-$n$,
otherwise. This \rul satisfies all the properties but neutrality. Fourth,
consider a \rul that assigns degenerate lotteries as follows: for each
profile, it allocates agent~1 their favourite object, say~$a_1$, with
probability~1; among the remaining objects, (i)~if agent~2's favourite
object is not~$a_1$, it selects the serial dictatorship with
respect to 2-3-4-$\cdots$-$n$; (ii)~if agent~2's favourite object is
$a_1$, it selects the serial dictatorship with respect to
2-4-3-$\cdots$-$n$. This \rul satisfies all the properties but
non-bossiness.

\section{Tree Representations of Hierarchies of Monarchies and Diarchies}
\label{apx:trees}

An HMD is parameterized by a sequencing rule that attaches a monarchy or a
diarchy to each non-terminal history. This parameterization is redundant in
two ways. First, the domain of a sequencing rule is a continuum of
histories, while the histories that it actually generates are finite in
number, so its values elsewhere are irrelevant. Second, and consequently, the
map from sequencing rules to \ruls is far from injective, so a sequencing
rule is not a name for an HMD.

In this appendix we supply such a name. We define a
\textbf{monarchy--diarchy tree}, or \textbf{MD-tree}: a finite rooted tree
whose vertices are labelled by monarchies and diarchies and whose arcs are
labelled by one of eight \emph{branch conditions}. The branch conditions
record only what an HMD is permitted to notice: whether the two diarchs at a
stage compete for zero, one, or two objects, and, once a pair of objects has
been left in fractional supply, whether the next monarch consumes neither,
one, or both of them. We show that MD-trees and HMDs describe the same \ruls
(\cref{tr:prop:equivalence}). Together with \cref{thm: characterization},
this gives a purely combinatorial description of the \ruls satisfying the
five axioms. We then show that the representation is essentially unique
(\cref{tr:prop:rigidity}): the only freedom left is that an
$\alpha$-diarchy of $i$ and $j$ is the same partial allocation \rul as a
$(1-\alpha)$-diarchy of $j$ and $i$. Uniqueness turns the representation into
a counting device. We derive a recursion for the number of MD-trees on $n$
agents up to this freedom, and hence for the number of HMDs up to the weights
attached to diarchies (\cref{tr:prop:count,tr:prop:growth}). The count is
doubly exponential in $n$, whereas there are only $n!$ serial dictatorships.

\subsection{Preliminaries}
\label{tr:sec:prelim}

We use the definitions of \cref{apx:formal alg}. Throughout this appendix, a
diarchy $d^{(i,j,\alpha)}$ has distinct agents $i\neq j$ and weight
$\alpha\in(0,1)$. This is without loss of generality: every HMD satisfies the
properties in \cref{thm: characterization}, and the proof of that theorem
(\cref{lem: step to recurse}) shows that every \rul satisfying these
properties is $\varphi^\sigma$ for an identity-overlap invariant MD-sequencing
rule $\sigma$ whose diarchies are of this form. Let $\mathcal M$ be the set of
monarchies and write $\mathcal D$ for the set of such diarchies. Define
\[
\operatorname{agents}(m^i) = \{i\}, \qquad \operatorname{agents}(d^{(i,j,\alpha)}) = \{i,j\}.
\]
For a diarchy $d^{(i,j,\alpha)}$ we call $i$ the \textbf{left diarch} and $j$
the \textbf{right diarch}. The order of the arguments is a matter of
description only:
\begin{equation}
\label{tr:eq:reversal}
d^{(i,j,\alpha)} = d^{(j,i,1-\alpha)} .
\end{equation}
Diarchies are used only at integral supply vectors. There, if the diarchs'
favourite available objects differ, each simply receives their own favourite
and $\alpha$ is immaterial; if they coincide, $\alpha$ is the share of the
common favourite that goes to the left diarch.

For a history $\eta = (\pi^k)_{k=1}^K$, let $N(\pi^k) = \{i\in N: \pi^k_i \neq
\vec 0\}$ and let $X(\eta) = \sum_{k \le K} \pi^k$ be the \textbf{cumulative
allocation} at $\eta$. We refer to the two requirements on an MD-sequencing
rule $\sigma$ as follows: for each $\eta\in \mathcal H^{NT}$,
\begin{enumerate}
\item[(\emph{i})] $\operatorname{agents}(\sigma(\eta)) \cap N(\eta) = \emptyset$, and
\item[(\emph{ii})] if $\integ(r(\eta)) = 0$, then $\sigma(\eta) \in \mathcal M$.
\end{enumerate}
We write $MD^\sigma = \varphi^\sigma$. By \cref{lem:measurable-iff-neutral},
a \rul is an HMD if and only if it is $MD^\sigma$ for some neutral
MD-sequencing rule $\sigma$, and we use this description from here on.

We record two facts about executions that are used repeatedly. The first
refines \cref{lem:shape}: a non-integral residual supply always takes a very
particular shape.

\begin{lem}
\label{tr:lem:shape}
Let $\sigma$ be an MD-sequencing rule and let $\eta = (\pi^1,\dots,\pi^K)$ be
reachable under $\sigma$. Either $\integ(r(\eta)) = 1$, or there are two
distinct objects $\ell$ and $\varrho$ and some $\alpha \in (0,1)$ such that
\[
r_{\ell}(\eta) = \alpha, \qquad r_{\varrho}(\eta) = 1-\alpha, \qquad
\underline A(r(\eta)) = \{\ell,\varrho\},
\]
and there are distinct agents $i_\ell, i_\varrho \in N(\eta)$ with
\[
X_{i_\ell \ell}(\eta) = 1-\alpha, \quad X_{k\ell}(\eta) = 0 \text{ for }
k \neq i_\ell,
\qquad
X_{i_\varrho \varrho}(\eta) = \alpha, \quad X_{k\varrho}(\eta) = 0 \text{ for }
k \neq i_\varrho .
\]
\end{lem}

We call $\ell$ and $\varrho$ the \textbf{left} and \textbf{right fractional
object}, respectively, $(\ell,\varrho)$ the \textbf{active pair}, and $i_\ell$
and $i_\varrho$ their \textbf{anchors}. We call a history \textbf{integral} or
\textbf{fractional} according to whether $\integ(r(\eta)) = 1$ or not.

\begin{proof}
We induct on $K$. The base case $\eta = ()$ holds since $r(()) = \vec 1$.
Suppose the conclusion holds at $\eta$ and let $\eta' = \eta\doubleplus\pi^{K+1}$ be
reachable, with $s = r(\eta)$ and $s' = r(\eta')$. If $\integ(s') = 1$ we are
done, so suppose not.

Suppose first that $\integ(s) = 1$. Since $\integ(s') = 0$, condition
(\emph{ii}) and the definition of a monarchy imply that
$\sigma(\eta) = d^{(i,j,\alpha)}$ for some $i,j \notin N(\eta)$ and some
$\alpha \in (0,1)$, and that $i$ and $j$ rank the same object $a$ first among
$A(s)$ but distinct objects $\ell$ and $\varrho$ second. Then
$\pi^{K+1}_i = \alpha a + (1-\alpha)\ell$ and
$\pi^{K+1}_j = (1-\alpha)a + \alpha\varrho$, so $s'_\ell = \alpha$,
$s'_\varrho = 1-\alpha$, and $\underline A(s') = \{\ell,\varrho\}$. As $a$
had unit supply and $\ell$ and $\varrho$ had zero columns in $X(\eta)$, the
anchors are $i$ and $j$, which are distinct.

Now suppose $\integ(s) = 0$, with active pair $(\ell,\varrho)$, weight
$\alpha$, and anchors $i_\ell \neq i_\varrho$. By condition (\emph{ii}),
$\sigma(\eta) = m^k$ for some $k \notin N(\eta)$, and
$\pi^{K+1}_k = \tau^k(s,P)$ for the profile $P$ at which $\eta'$ is reached.
Every object of $A(s) \setminus \{\ell,\varrho\}$ has unit supply. If $k$
ranks such an object $c$ first among $A(s)$, then $\pi^{K+1}_{kc} = 1$ and
neither the active pair nor its anchors change. If $k$ ranks $\ell$ first,
then, since $\integ(s') = 0$, the second-ranked object of $k$ in $A(s)$ is
some $c$ of unit supply, and $\pi^{K+1}_k = \alpha\ell + (1-\alpha)c$. Hence
$s'_\ell = 0$, $s'_c = \alpha$, and $s'_\varrho = 1-\alpha$, so the new
active pair is $(c,\varrho)$; the anchor of $c$ is $k$, with
$X_{kc}(\eta') = 1-\alpha$, and $k \neq i_\varrho$ since
$k \notin N(\eta)$. The case where $k$ ranks $\varrho$ first is symmetric.
\end{proof}

The second fact says that the preferences of agents not yet served are free.
We state it for MD-sequencing rules; the identical statement and proof apply
to the executions of MD-trees defined in \cref{tr:sec:trees}, since there too the
partial allocation \rul applied at a stage depends on the profile only
through the history so far.

\begin{lem}
\label{tr:lem:free-tail}
Let $\sigma$ be an MD-sequencing rule and let $\eta$ be reachable under
$\sigma$ at $P$. If $P'$ agrees with $P$ on $N(\eta)$, then $\eta$ is
reachable under $\sigma$ at $P'$, at the same stage.
\end{lem}

\begin{proof}
Induct through the partial allocations that make up $\eta$. At each preceding
history the rule selects the same partial allocation \rul, because the
history is the same, and the agents it selects lie in $N(\eta)$, whose
preferences are unchanged. So the same partial allocation is produced.
\end{proof}

\subsection{Monarchy--diarchy trees}
\label{tr:sec:trees}

\subsubsection{The definition}

Let
\[
\mathcal B = \{\mathsf I,\ \mathsf Z,\ \mathsf P,\ \mathsf F,\ \mathsf N,\ \mathsf L,\ \mathsf R,\ \mathsf B\}
\]
be the set of \emph{branch conditions}. The label $\mathsf I$ stands for
\emph{integral}. The labels $\mathsf Z$, $\mathsf P$, and $\mathsf F$ stand for \emph{zero},
\emph{partial}, and \emph{full support overlap}. The labels $\mathsf N$, $\mathsf L$,
$\mathsf R$, and $\mathsf B$ stand for \emph{neither}, \emph{left}, \emph{right}, and
\emph{both fractional objects taken}. \cref{tr:subsec:branches} explains each in
turn.

\begin{df}[Monarchy--diarchy tree]
\label{tr:def:MD-tree}
A \emph{monarchy--diarchy tree}, or \emph{MD-tree}, is a tuple
$\mathcal T = (V,Q,v^0,\Lambda,H)$ in which $V$ is a finite set of vertices,
$Q \subseteq V\times V$ is a set of arcs, $v^0 \in V$, the directed graph
$(V,Q)$ is a rooted tree with root $v^0$, and the maps
\[
\Lambda: V \to \mathcal M \cup \mathcal D, \qquad H: Q \to \mathcal B
\]
satisfy Properties L and S below.
\end{df}

For a nonroot vertex $v$ with incoming arc $(u,v)$, write $h(v) = H(u,v)$.
Call $v$ \emph{integral} if $h(v) \in \{\mathsf I,\mathsf Z,\mathsf F,\mathsf B\}$ and
\emph{fractional} if $h(v) \in \{\mathsf P,\mathsf N,\mathsf L,\mathsf R\}$. The root is integral.

\paragraph{Properties L: labelling conditions.}
\begin{enumerate}[L.1]
\item Each agent appears in exactly one vertex label along any root-to-leaf
path. That is, for every maximal path $(v^0,v^1,\dots,v^K)$, the family
$\{\operatorname{agents}(\Lambda(v^t))\}_{t=0}^K$ is a partition of $N$.
\item Distinct arcs leaving the same vertex carry distinct labels: if
$(v,w),(v,w') \in Q$ with $w \neq w'$, then $H(v,w) \neq H(v,w')$.
\end{enumerate}

\paragraph{Properties S: structure of the tree.}
\begin{enumerate}[S.1]
\item Every fractional vertex is labelled by a monarchy. Equivalently, every
vertex labelled by a diarchy is integral.
\item A nonleaf integral vertex labelled by a monarchy has exactly one
outgoing arc, labelled $\mathsf I$.
\item A nonleaf vertex labelled by a diarchy has exactly three outgoing arcs,
labelled $\mathsf Z$, $\mathsf P$, and $\mathsf F$.
\item A nonleaf fractional vertex has exactly four outgoing arcs, labelled
$\mathsf N$, $\mathsf L$, $\mathsf R$, and $\mathsf B$.
\end{enumerate}

\begin{df}[The \rul induced by an MD-tree]
\label{tr:def:MD-tree-mechanism}
Fix an MD-tree $\mathcal T$ and a profile $P$. Begin at the root $v^0$ with
residual supply $\vec 1$. At each visited vertex $v$, apply
$\Lambda(v)$ to the current residual supply and to $P$, and
subtract the resulting partial allocation. If $v$ is not a leaf, follow the
unique outgoing arc whose label is the branch condition that holds, as
described in \cref{tr:subsec:branches}. Continue until a leaf is reached. Write
$MD^{\mathcal T}(P)$ for the sum of the partial allocations along the resulting
root-to-leaf path.
\end{df}

Property L.1 guarantees that every agent is served exactly once, so
$MD^{\mathcal T}(P) \in \Pi$ and $MD^{\mathcal T}$ is a \rul. \cref{tr:subsec:branches} verifies
that exactly one branch condition holds at each visited vertex, so the
description above is well defined.

\subsubsection{Interpretation of the branches}
\label{tr:subsec:branches}

Suppose the procedure has reached a vertex $v$ with residual supply
vector $s$ and applies $\Lambda(v)$.

\paragraph{Integral monarchy vertices.}
Suppose $v$ is integral and $\Lambda(v) = m^i$.

\begin{description}
\item[$\mathsf I$, integral.] The monarch $i$ receives their favourite object in
$A(s)$ with probability one. The new residual supply is again integral.
\end{description}

\paragraph{Diarchy vertices.}
Suppose $\Lambda(v) = d^{(i,j,\alpha)}$. By Property S.1 and the induction at
the end of this subsection, $s$ is integral. Let $a_i$ and $b_i$ be the left
diarch's first- and second-ranked objects in $A(s)$, and let $a_j$ and $b_j$
be the right diarch's.

\begin{description}
\item[$\mathsf Z$, zero support overlap.] The diarchs' favourites differ,
$a_i \neq a_j$. Each receives their own favourite with probability one, so
their supports do not overlap. The new residual supply is integral.

\item[$\mathsf P$, partial support overlap.] The diarchs share a favourite,
$a_i = a_j = a$, but their second-ranked objects differ, $b_i \neq b_j$.
Write $\ell = b_i$ and $\varrho = b_j$. Then
\[
\pi_i = \alpha a + (1-\alpha)\ell, \qquad
\pi_j = (1-\alpha) a + \alpha \varrho,
\]
so the supports overlap in exactly one object. The only fractional
coordinates of the new residual supply are $s'_\ell = \alpha$ and
$s'_\varrho = 1-\alpha$, with anchors $i$ and $j$.

\item[$\mathsf F$, full support overlap.] The diarchs share both their favourite
and their second-ranked object, $a_i = a_j = a$ and $b_i = b_j = b$. They
receive
\[
\pi_i = \alpha a + (1-\alpha)b, \qquad \pi_j = (1-\alpha)a + \alpha b,
\]
so the supports coincide. Both $a$ and $b$ are exhausted and the new
residual supply is integral.
\end{description}

\paragraph{Fractional monarchy vertices.}
Suppose $v$ is fractional and $\Lambda(v) = m^i$. There is then an active
pair $(\ell,\varrho)$ and some $\alpha \in (0,1)$ with $s_\ell = \alpha$ and
$s_\varrho = 1-\alpha$, while every other positive coordinate of $s$ equals
one. The branch is determined by agent $i$'s ranking of $A(s)$.

\begin{description}
\item[$\mathsf N$, neither fractional object taken.] Agent $i$'s favourite object in
$A(s)$ is an object $c \notin \{\ell,\varrho\}$ of unit supply. The monarch
receives $c$, and the active pair and its anchors are unchanged.

\item[$\mathsf L$, left fractional object taken.] Agent $i$ ranks $\ell$ first and
some $c \notin \{\ell,\varrho\}$ of unit supply second. The monarch
receives $\alpha\ell + (1-\alpha)c$. The new active pair is $(c,\varrho)$,
with residual probabilities $\alpha$ and $1-\alpha$.

\item[$\mathsf R$, right fractional object taken.] Agent $i$ ranks $\varrho$ first
and some $c \notin \{\ell,\varrho\}$ of unit supply second. The monarch
receives $(1-\alpha)\varrho + \alpha c$. The new active pair is $(\ell,c)$,
with residual probabilities $\alpha$ and $1-\alpha$.

\item[$\mathsf B$, both fractional objects taken.] Agent $i$'s two favourite
objects in $A(s)$ are $\ell$ and $\varrho$, in either order. The monarch
receives $\alpha\ell + (1-\alpha)\varrho$. The new residual supply is
integral and the active pair ceases to exist.
\end{description}

The labels $\mathsf I, \mathsf Z, \mathsf F, \mathsf B$ therefore leave an integral residual problem
and the labels $\mathsf P, \mathsf N, \mathsf L, \mathsf R$ leave a fractional one. Since a vertex is
integral or fractional according to its incoming arc, it follows by induction
on depth that each vertex is reached with a residual supply of the type its
own label indicates. In particular a diarchy vertex, integral by Property
S.1, is reached with an integral residual supply, as a diarchy requires; and
at each visited vertex exactly one branch condition holds.

Two features of the branch conditions matter later. First, they are invariant
to relabelling of the objects: the conditions $\mathsf Z$, $\mathsf P$, and $\mathsf F$ depend
only on which of the diarchs' first- and second-ranked available objects
coincide, and the conditions $\mathsf N$, $\mathsf L$, $\mathsf R$, and $\mathsf B$ depend only on
the positions of the two fractional objects, relative to the objects of unit
supply, in the monarch's ranking. Second, the left and right roles survive
relabelling: by \cref{tr:lem:shape} the two fractional objects have distinct
anchors, so no permutation of the objects that fixes the preceding history
can interchange them. This is what allows an MD-tree to distinguish the $\mathsf L$
and $\mathsf R$ branches without referring to the labels of the objects. Note that
it relies on $\alpha \in (0,1)$: the anchor of each fractional object holds a
strictly positive complementary share.

\subsubsection{An example}
\label{tr:subsec:example}

\cref{tr:fig:example} shows an MD-tree for $N = \{1,2,3,4\}$ and
$A = \{a,b,c,d\}$. Agents $1$ and $2$ are diarchs at the root, with weight
$\sfrac13$. If their favourites differ, the tree serves agent $4$ and then
agent $3$. If their favourites and second choices coincide, agents $3$ and
$4$ become diarchs with equal weight. If their favourites coincide but their
second choices differ, a fractional pair is created; agent $3$ is then served
as a monarch, and agent $4$ receives whatever is left.

\begin{figure}[t]
\centering
\begin{tikzpicture}[
  >={Stealth[length=2.2mm]},
  vert/.style={draw, rounded corners=2pt, inner sep=4pt, minimum height=7mm,
               font=\small},
  lab/.style={font=\small, fill=white, inner sep=1.5pt}
]
\node[vert] (root) at (0,0) {$d^{(1,2,\sfrac13)}$};

\node[vert] (z1) at (5.0,4.6) {$m^{4}$};
\node[vert] (z2) at (10.0,4.6) {$m^{3}$};

\node[vert] (p1) at (5.0,0.3) {$m^{3}$};
\node[vert] (pn) at (10.0,2.6) {$m^{4}$};
\node[vert] (pl) at (10.0,1.2) {$m^{4}$};
\node[vert] (pr) at (10.0,-0.2) {$m^{4}$};
\node[vert] (pb) at (10.0,-1.6) {$m^{4}$};

\node[vert] (f1) at (5.0,-3.8) {$d^{(3,4,\sfrac12)}$};

\draw[->] (root) -- node[lab, pos=0.5] {$\mathsf Z$} (z1);
\draw[->] (z1) -- node[lab, pos=0.5] {$\mathsf I$} (z2);
\draw[->] (root) -- node[lab, pos=0.5] {$\mathsf P$} (p1);
\draw[->] (p1) -- node[lab, pos=0.32] {$\mathsf N$} (pn);
\draw[->] (p1) -- node[lab, pos=0.52] {$\mathsf L$} (pl);
\draw[->] (p1) -- node[lab, pos=0.52] {$\mathsf R$} (pr);
\draw[->] (p1) -- node[lab, pos=0.32] {$\mathsf B$} (pb);
\draw[->] (root) -- node[lab, pos=0.5] {$\mathsf F$} (f1);
\end{tikzpicture}
\caption{An MD-tree for four agents. Vertices carry monarchies and diarchies;
arcs carry branch conditions. The root and the vertices reached by $\mathsf Z$,
$\mathsf I$, $\mathsf F$, and $\mathsf B$ are integral; those reached by $\mathsf P$, $\mathsf N$, $\mathsf L$,
and $\mathsf R$ are fractional and are therefore labelled by monarchies, as
Property S.1 requires.}
\label{tr:fig:example}
\end{figure}

The example also illustrates a general point. Because a fractional vertex has
four children, the continuation below a fractional vertex is replicated once
for each of the four branches, even when the sequencing rule does not
distinguish among them. An MD-tree is therefore a faithful but not a compact
description, and a sequencing rule written as a table is often the more
readable presentation. What the tree buys is uniqueness, and with it the
ability to count.

\subsection{MD-trees and HMDs describe the same \ruls}
\label{tr:sec:equivalence}

Say that two histories $\eta$ and $\eta'$ are \emph{relabelling equivalent},
written $\eta \sim_{A} \eta'$, if $\eta' = \rho(\eta)$ for some permutation
$\rho: A \to A$. Since permutations of $A$ act on histories by
$(\rho,\eta) \mapsto \rho(\eta)$, with
$\rho(\pi)_{i\rho(a)} = \pi_{ia}$, the relation $\sim_{A}$ is an equivalence
relation, and neutrality of an MD-sequencing rule says exactly that its
selection is constant on each $\sim_{A}$-equivalence class of reachable
histories. Write $[\eta]$ for the $\sim_{A}$-equivalence class of $\eta$.

\begin{prop}
\label{tr:prop:equivalence}
\begin{enumerate}
\item For every neutral MD-sequencing rule $\sigma$ there is an MD-tree
$\mathcal T^\sigma$ with $MD^\sigma(P) = MD^{\mathcal T^\sigma}(P)$ for every $P \in \mathcal P^N$.
\item For every MD-tree $\mathcal T$ there is a neutral MD-sequencing rule
$\sigma^{\mathcal T}$ with $MD^{\mathcal T}(P) = MD^{\sigma^{\mathcal T}}(P)$ for every $P \in \mathcal P^N$.
\end{enumerate}
\end{prop}

\begin{cor}
\label{tr:cor:equivalence}
A \rul is a hierarchy of monarchies and diarchies if and only if it is
$MD^{\mathcal T}$ for some MD-tree $\mathcal T$. Consequently, by \cref{thm: characterization}, a
\rul is \strongefficient, \sdsp, non-bossy, neutral, and boundedly invariant
if and only if it is $MD^{\mathcal T}$ for some MD-tree $\mathcal T$.
\end{cor}

The proof of \cref{tr:prop:equivalence} occupies the rest of this subsection. For a
reachable history $\eta$, write
\[
\mathcal P^N(\eta) = \{P \in \mathcal P^N: \eta \text{ is reached in the execution of }
MD^\sigma \text{ at } P\},
\]
which is nonempty by definition and, by \cref{tr:lem:free-tail}, is closed under
arbitrary changes to the preferences of agents outside $N(\eta)$.

\subsubsection{Step 1: successor histories have the branch structure of
\texorpdfstring{\cref{tr:def:MD-tree}}{the definition of an MD-tree}}

Fix a neutral MD-sequencing rule $\sigma$ and a reachable non-terminal
history $\eta$. Let $s = r(\eta)$, let $f = \sigma(\eta)$, and for
$P \in \mathcal P^N(\eta)$ let $\pi(P) = f(s,P)$ be the partial allocation at $\eta$.
Let
\[
\mathcal S(\eta) = \{\eta\doubleplus\pi(P): P \in \mathcal P^N(\eta)\}
\]
be the set of successor histories of $\eta$. We show that $\mathcal S(\eta)$ meets
exactly one, three, or four $\sim_{A}$-equivalence classes, according to the
three cases of Properties S.2 to S.4, and that the classes are distinguished
by conditions that do not refer to the labels of the objects.

\paragraph{Case: $s$ integral.}
Every object in $A(s)$ has received zero probability at every preceding
stage. Hence any permutation of $A$ that fixes $A \setminus A(s)$ pointwise
leaves $\eta$ unchanged.

If $f = m^i$, then $i$ receives one object of $A(s)$ with probability one.
Since the objects of $A(s)$ may be permuted arbitrarily while fixing $\eta$,
all successor histories lie in a single $\sim_{A}$-equivalence class. This is
the branch $\mathsf I$, and the new residual supply is again integral.

Now suppose $f = d^{(i,j,\alpha)}$, and let $a_i,b_i$ and $a_j,b_j$ be the
first and second choices of $i$ and of $j$ in $A(s)$. There are three
possibilities.
\begin{enumerate}
\item If $a_i \neq a_j$, then $\pi_i(P) = a_i$ and $\pi_j(P) = a_j$. Both the
partial allocation and the new residual supply are integral.
\item If $a_i = a_j = a$ and $b_i = b_j = b$, then
$\pi_i(P) = \alpha a + (1-\alpha)b$ and $\pi_j(P) = (1-\alpha)a + \alpha b$.
The partial allocation is not integral, but the new residual supply is,
because $a$ and $b$ are exhausted.
\item If $a_i = a_j = a$ and $b_i = \ell \neq \varrho = b_j$, then
$\pi_i(P) = \alpha a + (1-\alpha)\ell$ and
$\pi_j(P) = (1-\alpha)a + \alpha\varrho$, and the only fractional coordinates
of the new residual supply are $s'_\ell = \alpha$ and
$s'_\varrho = 1-\alpha$.
\end{enumerate}
These are the branches $\mathsf Z$, $\mathsf F$, and $\mathsf P$. Within each, any two
successor histories are $\sim_{A}$-equivalent: map the relevant ordered pair or
ordered triple of available objects in one realization to that in the other
and fix everything else, which fixes $\eta$. The three are pairwise
inequivalent: $\mathsf Z$ is distinguished from the other two by integrality of the
partial allocation, and $\mathsf F$ from $\mathsf P$ by integrality of the new residual
supply. Both distinctions are preserved by relabelling.

\paragraph{Case: $s$ not integral.}
By \cref{tr:lem:shape} there is an active pair $(\ell,\varrho)$ with
$s_\ell = \alpha$, $s_\varrho = 1-\alpha$ for some $\alpha \in (0,1)$, with
distinct anchors $i_\ell$ and $i_\varrho$ in $N(\eta)$, and every other
positive coordinate of $s$ equals one.

No permutation of $A$ that fixes $\eta$ can interchange $\ell$ and $\varrho$.
Indeed, such a permutation fixes $X(\eta)$, so interchanging $\ell$ and
$\varrho$ would require the columns of $X(\eta)$ indexed by $\ell$ and
$\varrho$ to be identical. By \cref{tr:lem:shape} these columns have their
unique positive entries at the distinct agents $i_\ell$ and $i_\varrho$, and
both entries are positive because $\alpha \in (0,1)$.

Condition (\emph{ii}) gives $f = m^k$ for some $k \notin N(\eta)$. Every
object of $A(s)$ other than $\ell$ and $\varrho$ has a zero column in
$X(\eta)$, so all such objects may be permuted arbitrarily among themselves
while fixing $\eta$, $\ell$, and $\varrho$. There are four possibilities.
\begin{enumerate}
\item The monarch's favourite object in $A(s)$ is some object $c$ of unit
supply. The monarch receives $c$ and the active pair is unchanged.
\item The monarch ranks $\ell$ first and some object $c$ of unit supply
second, and receives $\alpha\ell + (1-\alpha)c$. The new active pair is
$(c,\varrho)$.
\item The monarch ranks $\varrho$ first and some object $c$ of unit supply
second, and receives $(1-\alpha)\varrho + \alpha c$. The new active pair is
$(\ell,c)$.
\item The monarch's two favourite objects in $A(s)$ are $\ell$ and $\varrho$,
in either order, and the monarch receives $\alpha\ell + (1-\alpha)\varrho$.
The new residual supply is integral.
\end{enumerate}
These are the branches $\mathsf N$, $\mathsf L$, $\mathsf R$, and $\mathsf B$. Within each, all
successor histories are $\sim_{A}$-equivalent, since the object $c$ of unit
supply, when one is involved, may be relabelled arbitrarily. The four are
pairwise inequivalent: $\mathsf N$ is distinguished by integrality of the monarch's
partial allocation; $\mathsf B$ is distinguished from $\mathsf L$ and $\mathsf R$ by
integrality of the new residual supply; and $\mathsf L$ and $\mathsf R$ are distinct
because, as just argued, no permutation preserving $\eta$ interchanges $\ell$
and $\varrho$.

Finally, each of the listed branches is realized by some profile in
$\mathcal P^N(\eta)$. Fix $P^0 \in \mathcal P^N(\eta)$ and modify only the preferences of
the agents in $\operatorname{agents}(f)$, which lie outside $N(\eta)$; by
\cref{tr:lem:free-tail} the modified profile still reaches $\eta$. At an
integral history the unassigned agents are in bijection with $A(s)$, so a
nonleaf diarchy vertex is reached with $|A(s)| \geq 3$, and the diarchs'
preferences can be chosen to produce any of $\mathsf Z$, $\mathsf P$, and $\mathsf F$. At a
nonleaf fractional history at least two agents are unassigned and the active
pair accounts for total probability one, so at least one object of unit
supply remains, and the monarch's preferences can be chosen to produce any of
$\mathsf N$, $\mathsf L$, $\mathsf R$, and $\mathsf B$.

\subsubsection{Step 2: from a neutral sequencing rule to an MD-tree}

Let $\sigma$ be a neutral MD-sequencing rule. Reachability is preserved by
relabelling: if $\eta$ is reached at $P$, then $\rho(\eta)$ is reached at
$\rho(P)$, by induction along the execution using neutrality of $\sigma$ and
of monarchies and diarchies. So each $\sim_{A}$-equivalence class consists
entirely of reachable histories or entirely of unreachable ones.

Take as vertices the $\sim_{A}$-equivalence classes of reachable non-terminal
histories,
\[
V^\sigma = \{[\eta]: \eta \text{ reachable and non-terminal}\},
\]
with root $v^0 = [()]$, and set $\Lambda^\sigma([\eta]) = \sigma(\eta)$. This
is well defined by neutrality. The vertex $[\eta]$ is a leaf if
$N(\eta) \cup \operatorname{agents}(\sigma(\eta)) = N$, a condition that does not depend on the
representative. Otherwise set
\[
Q^\sigma = \{([\eta],[\eta']): [\eta] \in V^\sigma \text{ is not a leaf
and } \eta' \in \mathcal S(\eta)\},
\]
and let $H^\sigma([\eta],[\eta'])$ be the label of the branch of Step 1 that
contains $\eta'$. Both $Q^\sigma$ and $H^\sigma$ are well defined: if
$\eta'' = \rho(\eta)$, then neutrality of $\sigma$ and of monarchies and
diarchies give $\mathcal S(\eta'') = \rho(\mathcal S(\eta))$, so $\eta$ and $\eta''$ determine
the same set of successor classes, and the branch conditions of Step 1 are
invariant to relabelling, so the label attached to a successor class does not
depend on the representative.

The resulting directed graph is a rooted tree. Each nonroot history has a
unique predecessor, obtained by deleting its last partial allocation, and
relabelling equivalent histories have relabelling equivalent predecessors, so
every nonroot vertex has a unique parent. Along every arc the number of
served agents strictly increases, so the graph is acyclic; and it is finite
because every path has at most $n$ vertices and every vertex has at most four
children.

Property L.1 follows from condition (\emph{i}) and from termination: no agent
is selected twice along an execution, and every agent is selected before a
terminal history is reached. Property L.2 holds by construction, since Step 1
shows that the successor classes at a vertex are pairwise inequivalent and
distinct labels are attached to distinct classes. Property S.1 follows from
condition (\emph{ii}). Properties S.2 to S.4 follow from Step 1 together with
the realizability argument at the end of it. So
$\mathcal T^\sigma = (V^\sigma, Q^\sigma, v^0, \Lambda^\sigma, H^\sigma)$ is an
MD-tree.

For each profile $P$, the path followed in $\mathcal T^\sigma$ records exactly the
sequence of $\sim_{A}$-equivalence classes of the histories generated by
$\sigma$ at $P$, and at every corresponding vertex the same monarchy or
diarchy is applied. Hence $MD^{\mathcal T^\sigma}(P) = MD^\sigma(P)$.

\subsubsection{Step 3: from an MD-tree to a neutral sequencing rule}

Let $\mathcal T = (V,Q,v^0,\Lambda,H)$ be an MD-tree. For each profile $P$ and each
vertex $v$ visited during the execution of $MD^{\mathcal T}$ at $P$, let $\eta(P,v)$
be the history accumulated along the path from the root to $v$, not including
the partial allocation made at $v$, and let
\[
\mathcal H^{\mathcal T} = \{\eta(P,v): P \in \mathcal P^N,\ v \text{ is visited at } P\}
\]
be the set of histories reachable under $\mathcal T$.

A reachable history determines its vertex uniquely. This follows by induction
on length: the empty history corresponds to the root, and if $\eta$
corresponds to $v$, then the partial allocation appended to it satisfies
exactly one of the mutually exclusive branch conditions at $v$, and therefore
determines a unique child of $v$. Define
\[
\sigma^{\mathcal T}(\eta) = \Lambda(v) \qquad \text{whenever } \eta \in \mathcal H^{\mathcal T}
\text{ corresponds to } v .
\]

This rule is neutral on reachable histories. Fix a permutation
$\rho: A \to A$. We show by induction along the path that if $P$ reaches $v$
with history $\eta$, then $\rho(P)$ reaches the same vertex $v$ with history
$\rho(\eta)$. The claim is immediate at the root. Suppose it holds at $v$.
Monarchies and diarchies are neutral, so the partial allocation made at
$\rho(P)$ is the $\rho$-relabelling of the partial allocation made at $P$.
Moreover, every branch condition is invariant to relabelling of the objects,
as noted at the end of \cref{tr:subsec:branches}. Hence $P$ and $\rho(P)$ follow
arcs with the same label and reach the same child. Therefore
$\sigma^{\mathcal T}(\rho(\eta)) = \sigma^{\mathcal T}(\eta)$ for every $\eta \in \mathcal H^{\mathcal T}$.

It remains to define $\sigma^{\mathcal T}$ at histories that are not reachable under
$\mathcal T$. Since $\mathcal H^{\mathcal T}$ is closed under relabelling, the unreachable histories
are a union of complete $\sim_{A}$-equivalence classes. For each unreachable
non-terminal history $\eta$, let $k(\eta)$ be the lowest-indexed agent not in
$N(\eta)$ and set $\sigma^{\mathcal T}(\eta) = m^{k(\eta)}$. This satisfies conditions
(\emph{i}) and (\emph{ii}), and it is constant on $\sim_{A}$-equivalence
classes because relabelling does not change $N(\eta)$.

Thus $\sigma^{\mathcal T}$ is a neutral MD-sequencing rule on all of $\mathcal H^{NT}$, and
by construction it selects the label of the corresponding vertex at every
history reached during the execution of $\mathcal T$. An induction along the
realized path gives $MD^{\sigma^{\mathcal T}}(P) = MD^{\mathcal T}(P)$ for every $P$, which
completes the proof of \cref{tr:prop:equivalence}. \hfill\qedsymbol

\begin{rmk}[Off-path values]
\label{tr:rmk:offpath}
\emph{The values of $\sigma^{\mathcal T}$ at unreachable histories are irrelevant to
$MD^{\sigma^{\mathcal T}}$, and the particular choice above is made only for
concreteness. Any assignment satisfying conditions (\emph{i}) and
(\emph{ii}) and constant on $\sim_{A}$-equivalence classes would do. This is
precisely the redundancy that MD-trees remove.}
\end{rmk}

\subsection{Uniqueness of the representation}
\label{tr:sec:rigidity}

\subsubsection{Every vertex matters}

We first record that an MD-tree carries no superfluous vertices.

\begin{lem}
\label{tr:lem:visited}
Every vertex of an MD-tree is visited at some preference profile.
\end{lem}

\begin{proof}
Induct on the depth of the vertex. The root is visited at every profile.
Suppose $v$ is visited at $P$ and is not a leaf, and let $\eta$ be the
history accumulated on the way to $v$, so that
$\operatorname{agents}(\Lambda(v)) \cap N(\eta) = \emptyset$ by Property L.1. By
\cref{tr:lem:free-tail}, the preferences of the agents in $\operatorname{agents}(\Lambda(v))$ may
be changed without disturbing the path to $v$. It remains to check that each
of the branch conditions prescribed by Properties S.2 to S.4 can then be
produced. If $v$ is integral and labelled by a monarchy, this is immediate.
If $v$ is labelled by a diarchy, then $v$ is integral, the unassigned agents
are in bijection with $A(r(\eta))$, and $v$ is not a leaf, so
$|A(r(\eta))| \geq 3$; preferences of the two diarchs can then be chosen to
produce each of $\mathsf Z$, $\mathsf P$, and $\mathsf F$. If $v$ is fractional, then at least
two agents are unassigned and the active pair accounts for total probability
one, so at least one object of unit supply remains; preferences of the
monarch can then be chosen to produce each of $\mathsf N$, $\mathsf L$, $\mathsf R$, and $\mathsf B$.
\end{proof}

\subsubsection{The construction is canonical}

Next we record that the passage from sequencing rules to trees involves no
arbitrary choices.

\begin{prop}
\label{tr:prop:canonical}
Let $\sigma$ and $\sigma'$ be neutral MD-sequencing rules.
\begin{enumerate}
\item If $\sigma$ and $\sigma'$ agree at every history reachable under
$\sigma$, then they have the same reachable histories,
$\mathcal T^\sigma = \mathcal T^{\sigma'}$, and $MD^\sigma = MD^{\sigma'}$.
\item For every MD-tree $\mathcal T$, $\mathcal T^{\sigma^{\mathcal T}} \cong \mathcal T$, where $\cong$
denotes an isomorphism of rooted trees preserving vertex and arc labels.
\end{enumerate}
\end{prop}

\begin{proof}
For the first part, note that the vertices of $\mathcal T^\sigma$ are the
$\sim_{A}$-equivalence classes of reachable non-terminal histories themselves,
not chosen representatives of those classes, so $\mathcal T^\sigma$ is determined by
$\sigma$. Moreover, whether a history of length $k+1$ is reachable depends on
$\sigma$ only through its values at histories of length at most $k$. An
induction on $k$ therefore shows that $\sigma$ and $\sigma'$ have the same
reachable histories, and the two trees coincide because the data defining
them coincide. The claim about induced \ruls follows.

For the second part, every vertex of $\mathcal T$ is visited at some profile by
\cref{tr:lem:visited}. Two histories that reach the same vertex of $\mathcal T$ are
$\sim_{A}$-equivalent, by the within-branch equivalence of Step 1. Hence the
map sending a vertex $v$ of $\mathcal T$ to the class of histories reaching it is a
bijection onto the vertices of $\mathcal T^{\sigma^{\mathcal T}}$, and it preserves both
vertex labels, by the definition of $\sigma^{\mathcal T}$, and arc labels, by Step 1.
The two trees are not literally equal only because the vertices of
$\mathcal T^{\sigma^{\mathcal T}}$ are built as sets of histories while those of $\mathcal T$ are
abstract, which is a matter of names.
\end{proof}

\subsubsection{Diarchy reversal}

By \eqref{tr:eq:reversal}, $d^{(i,j,\alpha)}$ and $d^{(j,i,1-\alpha)}$ are the
same partial allocation mechanism. They differ only in which diarch is called
the left one, and that choice orients the active pair created on the $\mathsf P$
branch. Rewriting a diarchy vertex therefore forces a relabelling of the
$\mathsf L$ and $\mathsf R$ arcs for as long as that pair remains active.

\begin{df}[Reversal]
\label{tr:def:reversal}
Let $\mathcal T = (V,Q,v^0,\Lambda,H)$ be an MD-tree and let $v \in V$ with
$\Lambda(v) = d^{(i,j,\alpha)}$. If $v$ is a leaf, the \emph{reversal of
$\mathcal T$ at $v$} replaces $\Lambda(v)$ by $d^{(j,i,1-\alpha)}$. If $v$ is not a
leaf, let $w$ be the child of $v$ along the $\mathsf P$ arc and let $D(v)$ be the
set of vertices $u$ such that every arc on the path from $w$ to $u$, if any,
is labelled $\mathsf N$, $\mathsf L$, or $\mathsf R$; in particular $w \in D(v)$. The
\emph{reversal of $\mathcal T$ at $v$} replaces $\Lambda(v)$ by
$d^{(j,i,1-\alpha)}$ and, at each nonleaf $u \in D(v)$, exchanges the labels
of the two arcs leaving $u$ that are labelled $\mathsf L$ and $\mathsf R$.
\end{df}

Every vertex of $D(v)$ is fractional, hence labelled by a monarchy, so a
reversal never touches the label of another diarchy vertex. It is immediate
that the result is again an MD-tree.

\begin{lem}
\label{tr:lem:reversal-invariance}
If $\mathcal T'$ is a reversal of $\mathcal T$, then $MD^{\mathcal T'} = MD^{\mathcal T}$.
\end{lem}

\begin{proof}
Fix $P$. The two trees have the same vertices and the same vertex labels
except at $v$, where the two labels are the same partial allocation
mechanism by \eqref{tr:eq:reversal}. So the executions produce the same partial
allocation at every visited vertex, provided they visit the same vertices. We
check the latter. Outside $D(v)$ the arc labels agree and the branch
conditions are unchanged, so the executions agree. At $v$, the branch
condition holding at $P$ is the same in both trees, since $\mathsf Z$, $\mathsf P$, and
$\mathsf F$ do not refer to which diarch is called left. If the $\mathsf P$ branch is
taken, the active pair created is $(\ell,\varrho)$ in $\mathcal T$ and
$(\varrho,\ell)$ in $\mathcal T'$, with the corresponding weights $\alpha$ and
$1-\alpha$. Hence at each $u \in D(v)$ the branch that $\mathcal T$ labels $\mathsf L$ is
the branch that $\mathcal T'$ labels $\mathsf R$, and conversely, which is exactly the
exchange made by the reversal. The $\mathsf N$ and $\mathsf B$ branches are unaffected,
and the $\mathsf B$ branch leaves $D(v)$. An induction along $D(v)$ completes the
argument.
\end{proof}

\subsubsection{Reversal is the only freedom}

\begin{prop}
\label{tr:prop:rigidity}
Let $\mathcal T$ and $\mathcal T'$ be MD-trees. Then $MD^{\mathcal T} = MD^{\mathcal T'}$ if and only if
$\mathcal T'$ is obtained from $\mathcal T$, up to isomorphism of labelled rooted trees, by
finitely many reversals.
\end{prop}

The ``if'' direction is \cref{tr:lem:reversal-invariance}. The ``only if''
direction rests on the following lemma, which says that an HMD determines not
only the allocation it produces but the way that allocation is built up.

\begin{lem}
\label{tr:lem:decomposition}
Let $\mathcal T$ and $\mathcal T'$ be MD-trees with $MD^{\mathcal T} = MD^{\mathcal T'}$. Then, for every
profile $P$, the executions of $\mathcal T$ and of $\mathcal T'$ at $P$ produce the same
sequence of partial allocations.
\end{lem}

\begin{proof}
Write $\varphi = MD^{\mathcal T} = MD^{\mathcal T'}$. We show by induction on $t \geq 0$
that, at every profile, the first $t$ partial allocations produced by the two
executions coincide. The case $t = 0$ is vacuous.

Assume the claim for $t$ and fix a profile $P$. Let $\eta$ be the common
history after $t$ steps, let $s = r(\eta)$, and let $N_0 = N \setminus
N(\eta)$. If $N_0 = \emptyset$ there is nothing to prove, so suppose not.
Let $f$ and $f'$ be the partial allocation mechanisms labelling the vertices
that $\mathcal T$ and $\mathcal T'$ have reached. Since $\sum_{a\in A} s_a = |N_0|$, we may
regard the continuation as an allocation problem for $N_0$ over $A(s)$. We
claim $f = f'$; the inductive step follows, since then
$f(s,P) = f'(s,P)$.

Let $P^0 \in \mathcal P$ be any preference that ranks the objects of $A(s)$ above
those of $A \setminus A(s)$, and let $\overline P$ be the profile that agrees
with $P$ on $N(\eta)$ and gives every agent of $N_0$ the preference $P^0$. By
\cref{tr:lem:free-tail} and the induction hypothesis, both executions reach
$\eta$ at $\overline P$, so both continue with $f$ and $f'$ respectively.
Because each agent is served exactly once, and because a served agent
receives nothing at any later step, the row $\varphi_k(\overline P)$ of the
allocation is exactly the lottery that $k$ receives at the step at which $k$
is served. Both $f$ and $f'$ must therefore be consistent with
$\varphi(\overline P)$, and we now show that $\varphi(\overline P)$ pins them
down.

Suppose first that $\integ(s) = 1$, and let $c_1$ and $c_2$ be the two
objects of $A(s)$ that $P^0$ ranks highest, with $c_2$ omitted if
$|A(s)| = 1$. Since $s_{c_1} = s_{c_2} = 1$, no agent of $N(\eta)$ holds any
probability of $c_1$ or of $c_2$. If the step at $\eta$ is a monarchy $m^k$
with $k \in N_0$, then $k$ receives $c_1$ with probability one and $c_1$ is
exhausted, so $\varphi_{kc_1}(\overline P) = 1$ and
$\varphi_{k'c_1}(\overline P) = 0$ for every other $k'$. If instead the step
is a diarchy $d^{(k,k',\beta)}$, then
$k$ and $k'$ share the favourite $c_1$ and the second choice $c_2$, so the
step is a full-overlap step and
\[
\varphi_k(\overline P) = \beta c_1 + (1-\beta) c_2, \qquad
\varphi_{k'}(\overline P) = (1-\beta) c_1 + \beta c_2 ,
\]
while every other agent receives zero probability of $c_1$ and of $c_2$.
These two possibilities are mutually exclusive and exhaustive, and each is
detected from the column of $\varphi(\overline P)$ indexed by $c_1$: in the
first, some agent holds all of $c_1$; in the second, two agents hold
$\beta$ and $1-\beta$ of it, with $\beta \in (0,1)$. In the first case the
served agent $k$ is identified, so $f = m^k = f'$. In the second case the
pair $\{k,k'\}$ is identified, and so is the weight, since the agent
holding $\beta$ of $c_1$ holds $1-\beta$ of $c_2$; hence
$f = d^{(k,k',\beta)}$ and $f' \in \{d^{(k,k',\beta)},
d^{(k',k,1-\beta)}\}$, and these coincide by \eqref{tr:eq:reversal}.

Suppose now that $\integ(s) = 0$. By Property S.1 both $f$ and $f'$ are
monarchies, say $f = m^k$ and $f' = m^{k'}$. If $|N_0| = 1$ then
$k = k'$ is the single remaining agent. Otherwise, by \cref{tr:lem:shape},
$A(s)$ consists of the active pair together with $|N_0| - 1 \geq 1$ objects
of unit supply. Choose $P^0$ to rank all objects of unit supply in $A(s)$
above the two fractional ones, and let $c_1$ be the highest ranked. Since
$s_{c_1} = 1$, no agent of $N(\eta)$ holds any probability of $c_1$, and the
agent served at the step in question receives $c_1$ with probability one and
exhausts it. That agent is therefore identified by the column of
$\varphi(\overline P)$ indexed by $c_1$. Hence $k = k'$ and $f = f'$.
\end{proof}

\begin{proof}[Proof of \cref{tr:prop:rigidity}]
Only the ``only if'' direction remains. Suppose $MD^{\mathcal T} = MD^{\mathcal T'}$. By
\cref{tr:lem:decomposition}, at every profile the two executions produce the
same history at every stage. Since a reachable history determines its vertex
uniquely in either tree, the assignment
\[
\psi:\ v \longmapsto v' \quad \text{whenever some profile reaches $v$ in
$\mathcal T$ and $v'$ in $\mathcal T'$ at the same stage}
\]
is a well-defined map from the vertices of $\mathcal T$ to those of $\mathcal T'$. It is
surjective and injective because every vertex of either tree is visited at
some profile (\cref{tr:lem:visited}) and because the histories reaching $v$
and those reaching $\psi(v)$ are the same set. It carries the root to the
root and preserves the parent relation, so it is an isomorphism of rooted
trees. By the last paragraph of the proof of \cref{tr:lem:decomposition}, it
preserves vertex labels except that a diarchy vertex $v$ with
$\Lambda(v) = d^{(i,j,\alpha)}$ may satisfy
$\Lambda'(\psi(v)) = d^{(j,i,1-\alpha)}$, written in the reversed way.

Let $W$ be the set of diarchy vertices at which the two descriptions differ
in this way, and let $\mathcal T''$ be the tree obtained from $\mathcal T$ by reversing at
every vertex of $W$. Reversals at distinct diarchy vertices commute and do
not interfere. A reversal at $v$ alters only $\Lambda(v)$ and the $\mathsf L$ and
$\mathsf R$ arc labels inside $D(v)$, and $D(v)$ contains no diarchy vertex, since
every vertex of $D(v)$ is fractional. Moreover, if $v_1 \neq v_2$ are diarchy
vertices then $D(v_1) \cap D(v_2) = \emptyset$. To see this, note first that
if $u \in D(v)$, then every vertex on the path from the $\mathsf P$ child of $v$ to
$u$ also lies in $D(v)$. Suppose now that $u \in D(v_1) \cap D(v_2)$. Then
$u$ descends from the $\mathsf P$ child of $v_1$ and from the $\mathsf P$ child of
$v_2$, and the ancestors of $u$ form a chain, so one of $v_1$ and $v_2$ lies
on the path from the other's $\mathsf P$ child to $u$; say $v_2$ does. Then
$v_2 \in D(v_1)$. But $v_2$ is a diarchy vertex, hence integral, hence not in
$D(v_1)$, a contradiction. So $\mathcal T''$ is well defined, and by
construction $\psi$ preserves vertex labels as maps from $\mathcal T''$ to $\mathcal T'$.

It remains to check that $\psi$ preserves arc labels from $\mathcal T''$ to $\mathcal T'$.
Consider an arc $(u,z)$ of $\mathcal T''$. The label of $(u,z)$ is determined by the
branch condition satisfied by the partial allocation appended at $u$, given
the orientation of the active pair, if any, that is in force at $u$. By
\cref{tr:lem:decomposition} the partial allocation appended is the same in both
trees, and by construction the orientation in force at $u$ in $\mathcal T''$ agrees
with that in $\mathcal T'$, since the two trees now describe every diarchy the same
way. Hence the branch conditions agree and
$H''(u,z) = H'(\psi(u),\psi(z))$. Therefore $\mathcal T'' \cong \mathcal T'$, and $\mathcal T''$
was obtained from $\mathcal T$ by finitely many reversals.
\end{proof}

\begin{rmk}[Reversal is a genuine ambiguity]
\label{tr:rmk:half}
\emph{Reversal at a diarchy vertex always produces a different MD-tree, since
the ordered pair labelling that vertex changes. This holds for every
$\alpha \in (0,1)$, including $\alpha = \sfrac12$. When $\alpha = \sfrac12$
the two diarchs receive symmetric partial allocations, but the orientation of
the active pair created on the $\mathsf P$ branch still records which diarch was
designated the left one, and the $\mathsf L$ and $\mathsf R$ arcs below refer to it.}
\end{rmk}

\subsection{How many HMDs are there?}
\label{tr:sec:counting}

An HMD carries continuously many parameters, namely the weight attached to
each diarchy, so we count combinatorial types. Call two MD-trees
\emph{equivalent} if one is obtained from the other by finitely many
reversals and by changing the weights at diarchy vertices, and let $M(n)$ be
the number of equivalence classes on an agent set of size $n$. By
\cref{tr:prop:rigidity}, two MD-trees in different classes induce different
HMDs, and the HMDs induced by trees in a single class form a family
parameterized by the weights. So $M(n)$ counts HMDs up to the weights
attached to diarchies.

It is convenient to phrase the count in terms of continuations. Fix a set of
agents $N'$ with $|N'| = m$ and think of a subtree hanging below a vertex.
Call it an \emph{integral-state tree} if it satisfies all the conditions of
\cref{tr:def:MD-tree} with $N$ replaced by $N'$ and with an integral supply
vector, and a \emph{fractional-state tree} if it satisfies them with a supply
vector having exactly two fractional coordinates in a designated order. Let
$I(m)$ and $F(m)$ be the numbers of such trees, counted up to reversals
inside them and up to weights. An MD-tree is an integral-state tree on $N$,
so
\[
M(n) = I(n).
\]

\begin{prop}
\label{tr:prop:count}
For every $m \geq 2$,
\[
\begin{aligned}
I(m) &= m\, I(m-1) + \tbinom{m}{2}\, I(m-2)^2\, F(m-2),\\
F(m) &= m\, I(m-1)\, F(m-1)^3,
\end{aligned}
\]
with $I(0) = I(1) = F(0) = F(1) = 1$.
\end{prop}

\begin{proof}
Consider an integral-state tree on $N'$ with $m \geq 2$. Its root is labelled
by a monarchy or by a diarchy, and the two cases are mutually exclusive.

If the root is a monarchy, there are $m$ choices of monarch. By Property S.2
the root has a single child, reached by $\mathsf I$, and by the analysis of
\cref{tr:subsec:branches} the residual supply there is integral. The subtree
below is an arbitrary integral-state tree on the remaining $m-1$ agents. This
gives the term $m I(m-1)$.

If the root is a diarchy, it has three children, reached by $\mathsf Z$, $\mathsf P$, and
$\mathsf F$. The residual supply is integral at the first and the third and has an
oriented pair of fractional coordinates at the second. The three subtrees are
therefore an integral-state tree, a fractional-state tree, and an
integral-state tree on the remaining $m-2$ agents, chosen independently. As
for the root label, a diarchy is specified by an ordered pair of agents,
$m(m-1)$ choices, and by a weight, which we are ignoring. Reversal at the
root acts on this data by exchanging the two agents of the ordered pair and
by reorienting the fractional pair below the $\mathsf P$ arc. It is a free
involution, since $i \neq j$, so the number of choices up to reversal at the
root is $\tbinom{m}{2}$, with the fractional-state subtree counted with its
orientation fixed. Reversals at diarchy vertices strictly below the root act
only inside the subtrees, since $D(v_1) \cap D(v_2) = \emptyset$ for distinct
diarchy vertices, and are already accounted for by $I(m-2)$ and $F(m-2)$.
This gives the term $\tbinom{m}{2} I(m-2)^2 F(m-2)$.

Now consider a fractional-state tree on $N'$ with $m \geq 2$. By Property S.1
its root is a monarchy, giving $m$ choices of monarch, and by Property S.4 it
has four children, reached by $\mathsf N$, $\mathsf L$, $\mathsf R$, and $\mathsf B$. By the analysis
of \cref{tr:subsec:branches} the residual supply is integral at the $\mathsf B$ child
and has an oriented pair of fractional coordinates at the other three, with
the orientation determined in each case. The four subtrees are chosen
independently, giving $m I(m-1) F(m-1)^3$.

The base cases are immediate: with no agents left the continuation is empty,
and with one agent left the continuation is the single vertex labelled by the
monarchy of that agent.
\end{proof}

\begin{table}[t]
\centering
\begin{tabular}{rrrr}
\hline
$m$ & $I(m)$ & $F(m)$ & $m!$ \\
\hline
$0$ & $1$ & $1$ & $1$\\
$1$ & $1$ & $1$ & $1$\\
$2$ & $3$ & $2$ & $2$\\
$3$ & $12$ & $72$ & $6$\\
$4$ & $156$ & $17{,}915{,}904$ & $24$\\
$5$ & $104{,}460$ & $\approx 4.5\times 10^{24}$ & $120$\\
$6$ & $6{,}540{,}022{,}222{,}920$ & $\approx 5.7 \times 10^{79}$ & $720$\\
\hline
\end{tabular}
\caption{The number $M(m) = I(m)$ of HMDs on $m$ agents, up to the weights
attached to diarchies, against the number $m!$ of serial dictatorships.}
\label{tr:tab:counts}
\end{table}

\cref{tr:tab:counts} reports the first few values. Serial dictatorships are the
HMDs whose trees use no diarchy, so that every nonleaf vertex has a single child
reached by $\mathsf I$; there are $n!$ of them. The gap opens immediately and then
grows without bound.

\begin{prop}
\label{tr:prop:growth}
For every $m \geq 2$, $F(m) \geq 2^{3^{m-2}}$, and for every $n \geq 4$,
\[
M(n) > 2^{3^{n-4}} .
\]
In particular $\log_2 M(n)$ grows at least like $3^n$, while
$\log_2 n! \leq n\log_2 n$, so the ratio $M(n)/n!$ is doubly exponential
in $n$.
\end{prop}

\begin{proof}
Since $m \geq 2$ and $I(m-1) \geq 1$, the second recursion gives
$F(m) \geq 2 F(m-1)^3 > F(m-1)^3$. As $F(2) = 2 = 2^{3^0}$, induction gives
$F(m) \geq 2^{3^{m-2}}$ for $m \geq 2$. For $n \geq 4$ the first recursion
gives $M(n) = I(n) \geq \tbinom{n}{2} I(n-2)^2 F(n-2) > F(n-2) \geq
2^{3^{n-4}}$, using $\tbinom n2 \geq 6$. The final statement follows from
$n! \leq n^n$.
\end{proof}

\begin{rmk}[Where the size comes from]
\emph{The doubly exponential growth is driven by $F$, that is, by the
branching that occurs while a fractional pair is active. A single diarchy
that generates a partial overlap creates a pair of fractional objects, and
every subsequent monarch may then be chosen contingent on which of the two
that monarch's predecessor consumed. Each of the three surviving branches
carries its own continuation, and the cubing in the recursion for $F$
compounds across steps. By contrast, the freedom in $I$ is only the choice of
who is served next and whether to pair two agents, which by itself would
give a merely exponential count.}
\end{rmk}

\begin{rmk}[Reading the count]
\emph{\cref{thm: characterization} is naturally read as an impossibility result: an
HMD can treat at most two agents even-handedly, so it cannot satisfy any of
the usual symmetry requirements once $n \geq 3$. The count qualifies that
reading. The class of \ruls satisfying the five axioms is not a thin
perturbation of the serial dictatorships; it is astronomically larger, and
almost all of its members are genuine contingent plans in which the order of
service responds to what has already been allocated. What the axioms rule out
is even-handedness, not flexibility.}
\end{rmk}

We conclude with three remarks. First, the count is of combinatorial types, and each
type is realized by a continuum of HMDs indexed by the weights at its diarchy
vertices. A count of \ruls satisfying the axioms is therefore a count of a
continuum. \cref{tr:prop:count} counts the discrete skeleton.

Second, MD-trees are faithful but not compact. As noted in
\cref{tr:subsec:example}, an MD-tree replicates a subtree once for each of
the four branches at a nonleaf fractional vertex, even when the sequencing rule does
not distinguish among them. A sequencing rule presented as a table remains
the more economical description. The tree is the canonical one.

Third, the analogy with \citeposs{papai2000} inheritance trees is suggestive
but limited. Both devices encode a contingent plan as a labelled tree, and in
both cases the plan is pinned down only up to a symmetry of the labelling.
The objects being tracked differ, however: an inheritance tree branches on
which agent claims which object, while an MD-tree branches only on features
of the allocation so far that survive relabelling of the objects. Whether the
two representations can be placed inside a common framework, perhaps by way
of the results of \citet{PyciaUnverTE2017}, is a question we leave open.

\setstretch{1.18}

\bibliography{refs.bib}

\end{document}